\documentclass[lettersize,journal]{IEEEtran}
\normalsize
\ifCLASSINFOpdf
   \usepackage[pdftex]{graphicx}
   \graphicspath{{../pdf/}{../jpeg/}}
   \DeclareGraphicsExtensions{.pdf,.jpeg,.png}
\else
   \usepackage[dvips]{graphicx}
   \graphicspath{{../eps/}}
   \DeclareGraphicsExtensions{.eps}
\fi
\usepackage[table]{xcolor}
\usepackage{cite}
\usepackage{bm}
\usepackage[caption=false]{subfig}
\usepackage{dsfont}
\usepackage{multirow}
\usepackage{amssymb}
\usepackage{makecell}
\usepackage{color}
\usepackage{arydshln}
\usepackage{amsthm}
\usepackage{amsmath}
\usepackage[colorlinks,linkcolor=blue,anchorcolor=blue,citecolor=blue]{hyperref}
\usepackage{mathrsfs}
\usepackage{extarrows}
\usepackage[linesnumbered,ruled]{algorithm2e}
\theoremstyle{remark}
\newtheorem{definition}{Definition}
\newtheorem{example}{Example}
\newtheorem{remark}{Remark}
\newtheorem{theorem}{Theorem}
\newtheorem{lemma}{Lemma}
\newtheorem{corollary}{Corollary}

\makeatletter
\newcommand{\biggg}{\bBigg@{3}}
\newcommand{\qbin}{\genfrac{[}{]}{0pt}{}}
\makeatother
\begin{document}
\title{Performance Bounds and Degree-Distribution Optimization of Finite-Length BATS Codes}

\author{
	Mingyang~Zhu,~\IEEEmembership{Member,~IEEE,} Shenghao~Yang,~\IEEEmembership{Member,~IEEE,} Ming~Jiang,~\IEEEmembership{Member,~IEEE,} and~Chunming~Zhao,~\IEEEmembership{Member,~IEEE}
	\thanks{This work was supported in part by the National Natural Science Foundation of China under Grant 62331002, Grant 62171399, and Grant 62171238; in part by the National Key Research and Development Program of China under Grant 2022YFA1005000; in part by the Fundamental Research Funds for the Central Universities of China under Grant NKU 050-63233070; and in part by the Fellowship Award from the Research Grants Council of the Hong Kong Special Administrative Region, China, under Grant CUHK SRFS2223-4S03. \textit{(Corresponding author: Shenghao Yang.)}}
	\thanks{Mingyang Zhu was with the National Mobile Communications Research Laboratory, Southeast University, Nanjing 210096, China. He is now with the Institute of Network Coding, The Chinese University of Hong Kong, Hong Kong SAR, China. (e-mail: mingyangzhu@cuhk.edu.hk).}
	\thanks{Shenghao Yang is with the School of Science and Engineering, The Chinese University of Hong Kong, Shenzhen, Guangdong 518172, China (e-mail: shyang@cuhk.edu.cn).}
	\thanks{Ming Jiang and Chunming Zhao are with the National Mobile Communications Research Laboratory, Southeast University, Nanjing 210096, China, and also with the Purple Mountain Laboratories, Nanjing 211111, China (e-mail: jiang\_ming@seu.edu.cn; cmzhao@seu.edu.cn).}
}
\maketitle

\begin{abstract}
	Batched sparse (BATS) codes were proposed as a reliable communication solution for networks with packet loss. In the finite-length regime, the error probability of BATS codes under belief propagation (BP) decoding has been studied in the literature and can be analyzed by recursive formulae. However, all existing analyses have not considered precoding or have treated the BATS code and the precode as two separate entities. In this paper, we analyze the word-wise error probability of finite-length BATS codes with a precode under \textit{joint decoding}, including BP decoding and maximum-likelihood (ML) decoding. The joint BP decoder performs peeling decoding on a joint Tanner graph constructed from both the BATS and the precode Tanner graphs, and the joint ML decoder solves a single linear system with all linear constraints implied by the BATS code and the precode. We derive closed-form upper bounds on the error probability for both decoders. Specifically, low-density parity-check (LDPC) precodes are used for BP decoding, and any generic precode can be used for ML decoding. Even for BATS codes without a precode, the derived upper bound for BP decoding is more accurate than the approximate recursive formula, and easier to compute than the exact recursive formula. The accuracy of the two upper bounds has been verified by many simulation results. Based on the two upper bounds, we formulate an optimization problem to optimize the degree distribution of LDPC-precoded BATS codes, which improves BP performance, ML performance, or both. In our experiments, to transmit $128$ packets over a line network with packet loss, the optimized LDPC-precoded BATS codes reduce the transmission overhead to less than 50\% of that of standard BATS codes under comparable decoding complexity constraints.
\end{abstract}

\begin{IEEEkeywords}
	Batched sparse codes, finite-length analysis, upper bounds, belief propagation, maximum likelihood, error probability.
\end{IEEEkeywords}
\IEEEpeerreviewmaketitle

\section{Introduction}
Batched sparse (BATS) codes are an efficient solution for network communications with packet loss~\cite{BATS,FL_analysis_BATS}. BATS codes can be regarded as \textit{concatenated codes} over a finite field $\mathbb{F}_q$, where the outer code is a generalized fountain code using the matrix-form Luby Transform (LT)~\cite{LT} and the inner code is a random linear network code (RLNC). More precisely, the outer code generates a sequence of \textit{batches}, each of which consists of $M$ linear combinations of specific input symbols, and $M$ is called the \textit{batch size}. Then, the inner code operates at the intermediate network nodes, generating linear combinations of the symbols belonging to the same batch.

BATS codes share some similarities with many low-complexity linear network coding schemes, such as L-chunked codes~\cite{Yang2014Lchunked,Tang2018Lchunked}, overlapped chunked codes~\cite{5191397, overlapChunk, 5695118, Tang2012EOC}, and Gamma codes~\cite{GammaCodes2012}. A common feature of these low-complexity network coding schemes is the utilization of additional check constraints to \textit{couple} chunked network codes, improving the achievable rate by sharing information through global or local check constraints. Compared to the aforementioned low-complexity schemes, BATS codes generally achieve higher rate in the asymptotic regime and have the ability to generate unlimited batches. Numerical results of the achievable rates of these schemes can be found in~\cite{Tang2018Lchunked}.

The asymptotic performance of BATS codes under belief propagation (BP) decoding\footnote{The BP decoding of BATS codes still involves Gaussian eliminations. However, the Gaussian eliminations only apply to relatively small linear systems, each of which is associated to a batch.} has been analyzed in~\cite{BATS,Yang2016tree}. Both works provided a sufficient condition such that the BP decoder can recover a given fraction of the input symbols with high probability when the number of input symbols tends to infinity. The proof of this sufficient condition in~\cite{BATS} is based on a differential equation approach, while the proof in~\cite{Yang2016tree} uses a tree-based approach similar to the density evolution for low-density parity-check (LDPC) codes~\cite{luby1998analysis}.

The degree distribution of BATS codes plays a crucial role in their decoding performance. In~\cite{BATS}, some optimization problems based on the asymptotic analysis were raised to design asymptotically optimal degree distributions, which can be approximately solved by linear programming. In~\cite{FL_analysis_BATS}, an exact and an approximate recursive formulae on the error probability of BP decoding of finite-length BATS codes have been derived. Using one of these recursive formulae as the objective function, a general framework~\cite{Raptor,Lazaro2017inactivation} to numerically optimize the degree distribution for fountain codes can be adopted to optimize the degree distribution for BATS codes. Some heuristic methods to design the degree distribution of finite-length BATS codes have been studied in~\cite{Xu2017QUBATS,Zhu2023LDPCBATS}. The numerical results in~\cite{FL_analysis_BATS, Xu2017QUBATS, Zhu2023LDPCBATS} demonstrate that reasonable degree-distribution designs in the finite-length regime can lead to significant performance gains compared to using asymptotically optimal degree distributions.

However, to date, there has been no analysis of the error probability for BATS codes \textit{with a precode} in the finite-length regime, either for BP decoding or maximum-likelihood (ML) decoding. As a generalization of fountain codes, we can empirically infer that the precode is significant for BATS codes. Therefore, we need theoretical guidance to optimize precoded BATS codes. The main contributions of this paper are the establishment of performance upper bounds for precoded BATS codes in the finite-length regime and a method for optimizing the degree distribution based on these bounds, as summarized in Sec.~\ref{subsection:introduction_main_results}. Before presenting our results, we review some existing results on the finite-length analyses of fountain codes, LDPC codes, and BATS codes.

\subsection{Related Works}
There are several existing studies on the finite-length performance of BP decoding for fountain codes, LDPC codes, and BATS codes. Karp \textit{et al.} developed a recursive formula to compute the error probability of LT codes based on a finite-state machine model~\cite{FL_analysis_LT}. Shokrollahi provided an analysis on the error probability of Raptor codes, utilizing the finite-length analyses of LT codes and a class of variable-regular LDPC codes~\cite{Raptor,shokrollahi2001finite}. His analysis is based on a cascade of an LT decoder and an LDPC decoder, which neglects the joint iteration between the two parts, and so becomes an upper bound on the overall error probability. Moreover, several methods for finite-length analysis of regular LDPC codes on the erasure channel presented in~\cite{FL_analysis_LDPC,Modern_Coding_Theory,Johnson2009finite} can also be directly applied to improve Shokrollahi's upper bound for Raptor codes with more general LDPC precodes. Yang \textit{et al.} provided two recursive formulae to compute the error probability of BATS codes for a given number of received batches~\cite{FL_analysis_BATS}. These two recursive formulae also follow the model of a finite-state machine, but in comparison to the recursive formula for LT codes, they are more computationally complex. Note that the precode of BATS codes are not taken into account in these two recursive formulae.

A few works have addressed the error probability of LT and Raptor codes under ML decoding~\cite{Rahnavard2007, Birgit2013LT, Wang2016performanceRaptorML, Zhang2017boundsRaptorQ, Lazaro2017inactivation, Lazaro2021bounds}. In~\cite{Rahnavard2007}, Rahnavard \textit{et al.} derived upper and lower bounds on the bit error probability for binary LT and Raptor codes, where all entries of the parity-check matrix of the precode are independent and identically distributed (i.i.d.) Bernoulli random variables. In~\cite{Birgit2013LT}, Schotsch \textit{et al.} provided the upper and lower bounds on the error probability of LT codes on word as well as on symbol level, where the upper bound on the symbol-wise error probability is a generalization of~\cite{Rahnavard2007}. Wang \textit{et al.} provided upper and lower bounds on error probability of binary Raptor codes with a systematic low-density
generator-matrix precode~\cite{Wang2016performanceRaptorML}. This work was elegantly extended in~\cite{Zhang2017boundsRaptorQ}, where $q$-ary Raptor codes are considered. Both~\cite{Wang2016performanceRaptorML} and~\cite{Zhang2017boundsRaptorQ} consider very short Raptor codes (the number of input symbols $\le 100$). Furthermore, upper and lower bounds for $q$-ary Raptor codes with generic $q$-ary linear precodes were derived in~\cite{Lazaro2021bounds}. 

\textit{Inactivation decoding} is an efficient ML decoding for sparse linear systems, and its complexity is typically measured by the number of inactive symbols (see~\cite{InactivationPatent,Raptor_codes_foundations_and_trends,Lazaro2017inactivation} for inactivation decoding and inactive symbols). In~\cite{FL_analysis_BATS}, Yang \textit{et al.} developed a recursive formula to compute the expected number of inactive symbols for BATS codes (the inactivation decoding of BATS codes is similar to that of fountain codes, which is introduced in Sec.~\ref{subsection:decoding_std_BATS}). This recursive formula becomes essentially the same as the recursive formula for fountain codes presented in~\cite{Lazaro2017inactivation} when the batch size is one. However, there has been no analysis on the error probability of BATS codes under ML decoding, even without a precode.

\subsection{Summary of Results}\label{subsection:introduction_main_results}
We derive two upper bounds on the error probability\footnote{In this paper, the error probability stands for word-wise error probability, i.e., the probability that the decoder cannot recover all input symbols.} of \textit{precoded BATS codes} under BP decoding and ML decoding, and both bounds result from application of the union bound. For both BP decoding and ML decoding, \textit{joint decoders} that simultaneously decode the BATS code and the precode are considered. A degree-distribution optimization method based on the derived performance bounds is developed.

For BP decoding, the precode needs to be sparse so that it can effectively participate in the decoding process. Thus, we consider the precode to be an LDPC code. To facilitate the analysis, we do not focus on particular LDPC codes, but rather consider randomly selecting the precode from a $q$-ary regular LDPC \textit{ensemble} (this technique is widely used in the analysis of LDPC codes). The derived upper bound reveals the graph structures (in the graphs of precoded BATS codes) that lead to BP decoding failures. When the number of input symbols is relatively small, e.g., 128, this upper bound is generally tight at the error probability below $0.01$. The accuracy of this upper bound is independent of $q$, i.e., the order of the finite field from which the code is constructed. Even when considering BATS codes without a precode, this upper bound (with a trivial precode of rate $1$) still has some advantages compared to the two recursive formulae in~\cite{FL_analysis_BATS}. On one hand, this upper bound is easier to compute than the \textit{exact recursive formula}~\cite[Theorem 1]{FL_analysis_BATS}. On the other hand, this upper bound is more accurate than the \textit{approximate recursive formula}~\cite[Theorem 16]{FL_analysis_BATS}. To the best of the authors' knowledge, this is the first work which analyzes the BP performance of precoded BATS codes in the finite-length regime.


For ML decoding, the derived upper bound is applicable for BATS codes with any generic precode as long as the weight enumerator, i.e., the number of codewords having a certain Hamming weight, of the precode is known. While there have been previous studies on the performance of joint ML decoding of fountain codes \cite{Wang2016performanceRaptorML, Zhang2017boundsRaptorQ, Lazaro2021bounds}, these bounds cannot be easily extended to our case due to the matrix generalization of LT codes and the RLNCs used in BATS codes. This upper bound is a function of the number $n$ of transmitted batches, and clearly reveals how the error probability of ML decoding decreases exponentially with $n$. Although the upper bound we provide for ML decoding has certain limitations, such as its applicability to BATS codes with a relatively small number of input symbols and low-order finite fields (unlike the upper bound for BP decoding, this upper bound becomes looser as $q$ increases), it represents the first analytical result on the error probability of BATS codes under ML decoding.

Using the two upper bounds to design the objective function, we formulate an optimization problem to optimize the degree distribution of LDPC-precoded BATS (LDPC-BATS) codes. This optimization problem aims to find the degree distribution that minimizes the BP performance upper bound under the constraint that the ML performance upper bound remains below a certain threshold $\epsilon^*$. For $\epsilon^* = 1$, the resulting degree distribution is near optimal for BP decoding. For $\epsilon^* < 1$ and is chosen appropriately, the resulting degree distribution ensures that inactivation decoding has guaranteed performance and low decoding complexity, with the latter achieved through the joint BP and ML decoding property of inactivation decoding. Similar to \cite{FL_analysis_BATS}, this optimization problem can be numerically solved by iterative optimization. In our experience, this approach converges to a degree distribution that is significantly better than the initial one after very few iterations (e.g., $1 \sim 5$). The optimized LDPC-BATS codes outperform the conventional BATS codes (optimized by the existing methods) under BP decoding, ML decoding, or both. For example, when transmitting $128$ packets over a network with packet loss, the optimized LDPC-BATS codes can reduce the transmission overhead to below 50\% of that of the standard BATS codes optimized by~\cite{FL_analysis_BATS}. Furthermore, the optimized LDPC-BATS codes can achieve low complexity in inactivation decoding, even with a relatively low-rate precode.

\subsection{Paper Organization}
The rest of the paper is organized as follows. Section~\ref{sec:Preliminaries} introduces the fundamentals of BATS codes. Section~\ref{Sec_main_results} presents two upper bounds on the error probability of BP decoding and ML decoding, respectively, and provides numerical results as validation. The proofs of the bounds in Section~\ref{Sec_main_results} are provided in appendices. Section~\ref{Sec_opt_DD} presents how to use the two upper bounds to optimize the degree distribution of LDPC-BATS codes. At last, Section~\ref{sec:Conclusion} summarizes the paper.

\subsection{Notations}
\textit{Set and field notations:} The integer set is denoted by $\mathbb{Z}$ and the positive integer set is denoted by $\mathbb{Z}^+$. The ordered integer set from $i$ to $j$ ($j \ge i \ge 0$) is represented by $[i:j]$, i.e., $[i:j] = \{k \in  {\mathbb{Z}}: i \le k \le j\}$. The set difference of ${\cal A}$ and ${\cal B}$ is denoted by ${\cal A} \setminus {\cal B}$, which is the set of all elements belonging to ${\cal A}$ but not ${\cal B}$. Let $\mathbb{F}_q$ denote the finite field consisting of $q$ elements. 

\textit{Vectors and matrix notations:} The indices of vectors and matrices start from $1$. For a vector ${\bf a}$, denote by ${\bf a}[i:j]$ the subvector formed by the components with indices in $[i:j]$. For an $m \times n$ matrix $\bf A$, ${\bf A}[i_1:i_2,j_1:j_2]$ is the $(i_2-i_1+1)\times(j_2-j_1+1)$ submatrix formed by the components with the row indices in $[i_1:i_2]$ and column indices in $[j_1:j_2]$. Define two shorthand notations $[i]$ and $[:]$ related to $[i:j]$, where $[i] \triangleq [i:i]$ and $[:]$ is the universal set of row indices or column indices. Therefore, we can write ${\bf a}[i]$, ${\bf A}[i,j]$, ${\bf A}[:,j_1:j_2]$, ${\bf A}[i_1:i_2,:]$, etc. The rank of a matrix $\bf A$ is denoted by ${\rm rk}({\bf A})$, and the transpose of a matrix $\bf A$ is denoted by ${\bf A}^{\textsf{T}}$.

\section{Preliminaries}\label{sec:Preliminaries}
As a starting point, Sec.~\ref{subsec:BATS}, \ref{subsection:Tanner_std_BATS}, and \ref{subsection:decoding_std_BATS} review the encoding, the Tanner graph representation, and the decoding of BATS codes, respectively. Then, Sec.~\ref{sec_preliminaries_LDPCBATS} provides a natural extension of the three aforementioned aspects for LDPC-BATS codes. In Sec.~\ref{subsction:ensembles}, we define the ensembles for LDPC-BATS codes, which are useful for the proofs of the performance bounds.

\subsection{BATS Codes}\label{subsec:BATS}
Following \cite{BATS}, we first introduce BATS codes without a precode, referred to as \textit{standard BATS codes}. Fix a finite field $\mathbb{F}_q$, referred to as the \textit{base field}. $K$ input symbols ${\bf B}= (B_1,B_2,\ldots,B_K) \in \mathbb{F}_q^K$ are going to be transmitted through a network with erasure channels. Consider a BATS code with $K$ input symbols\footnote{In general, we may consider $K$ input packets, each of which is a vector over the base field. However, considering $K$ input symbols does not affect the performance analysis in this paper.} of the base field and $n$ batches\footnote{In practice, BATS codes can be \textit{rateless}, that is, continuous batches can be transmitted until the decoder declares a decoding success. Here, we consider a fixed number of batches merely for theoretical analysis.}. 

Fix an integer $M \ge 1$ called the \textit{batch size}. The outer code of a BATS code generates a sequence of $n$ batches ${\bf X}_1,{\bf X}_2,\ldots,{\bf X}_n$:
\begin{equation}\label{eq_outer_encode}
	{\bf X}_i = {\bf B}_i {\bf G}_i,
\end{equation}
where ${\bf B}_i$ is a row vector consisting of ${\rm dg}_i$ distinct input symbols chosen from $\bf B$, and ${\bf G}_i$ is a ${\rm dg}_i \times M$ \textit{totally random matrix} over the base field (i.e., all components are independently and uniformly chosen from $\mathbb{F}_q$ at random), called the \textit{generator matrix}\footnote{The concept of a ``generator matrix'' can sometimes be confused with the generator matrix of a precode. In this paper, unless explicitly specified as the ``generator matrix of a precode'', the term ``generator matrix'' refers to ${\bf G}_i$ defined here.}. ${\rm dg}_i$ is called the \textit{degree of the $i$-th batch}. The degrees ${\rm dg}_i$, $i = 1,2,\ldots$, are i.i.d. random variable following a given distribution ${\bm \Psi} = (\Psi_1,\ldots,\Psi_K)$, i.e., $\Pr\{{\rm dg}_i = d\} = \Psi_d$. ${\bm \Psi} $ is called the \textit{degree distribution of the BATS code}. Let ${\cal I}_i \subset [1:K]$ be the index set containing the indices of ${\rm dg}_i$ input symbols in ${\bf B}_i$. The outer encoding equation (\ref{eq_outer_encode}) is a matrix generalization of fountain codes. 

Then, the batches are transmitted through a network, where each node performs RLNC to linearly combine symbols belonging to the same batch. At a destination node, the received row vector of the $i$-th batch can be written as
\begin{equation}\label{eq_batch_equation}
	{\bf Y}_i = {\bf X}_i  {\bf H}_i = {\bf B}_i {\bf G}_i {\bf H}_i,
\end{equation}
where ${\bf H}_i$ is an $M$-row random matrix over the base field called the \textit{transfer matrix}. Each column of ${\bf H}_i$ is a coefficient vector of a linear combination of the $M$ outer-encoded symbols, i.e., ${\bf X}_i$. Thus, the RLNC is also called the inner code of a BATS code. Eq. (\ref{eq_batch_equation}) is referred to as the \textit{batch equation} in this paper. 

Assume that the ranks of transfer matrices are i.i.d., and define $h_k = \Pr\{{\rm rk}({\bf H}_i) = k\}$. We call ${\bf h} = (h_0,h_1,\ldots,h_M)$ the \textit{empirical rank distribution}, which is determined by the network and the RLNC scheme. Define the \textit{rate} of a standard BATS code as $K/n$. The capacity of a destination node with empirical rank distribution $\bf h$, i.e., the maximum rate at which error-free transmission can be achieved, is~\cite{LOC2010}:
\begin{equation}\label{eq:h_capacity}
	\bar{h} \triangleq \sum_{i=1}^{M} i h_i.
\end{equation}
We may consider $\bf h$ as a part of the code. Then, a standard BATS code is specified by a 6-tuple $(K,n,M,{\bm \Psi},{\bf h},q)$, denoted by $\mathscr{C}_{\rm std}^{(q)}(K,n,M,{\bm \Psi},{\bf h})$. Note that $\mathscr{C}_{\rm std}^{(q)}(K,n,M,{\bm \Psi},{\bf h})$ is a \textit{random code}.

\subsection{Tanner Graph Representation of BATS Codes}\label{subsection:Tanner_std_BATS}
An instance of a BATS code can be uniquely represented by a Tanner graph. Here, ``instance'' means that in a Tanner graph, ${\cal I}_i$, ${\bf G}_i$, and ${\bf H}_i$ are fixed for all $i \in [1:n]$. The Tanner graph of a standard BATS code defined from the encoding perspective is shown in Fig.~\ref{fig:Tanner_std_BATS}(a), where a circle represents a variable node (VN), corresponding to an input symbol over $\mathbb{F}_q$; a square represents an outer code constraint ${\bf X}_i = {\bf B}_i {\bf G}_i$; and a triangle represents an inner code constraint ${\bf Y}_i = {\bf X}_i {\bf H}_i$. From the decoding perspective, solving a batch equation is equivalent to solving the combination of a square and the corresponding triangle in Fig.~\ref{fig:Tanner_std_BATS}(a). Thus, we can regard the combination of a square and the corresponding triangle as a single check node (CN), referred to as a BATS CN (B-CN). The Tanner graph defined from the decoding perspective in shown in Fig.~\ref{fig:Tanner_std_BATS}(b), where a square with~$\times$ is a B-CN (i.e., a batch equation).

\begin{figure*}[!t]
	\centering
	\subfloat[]{\includegraphics[width=2.7in]{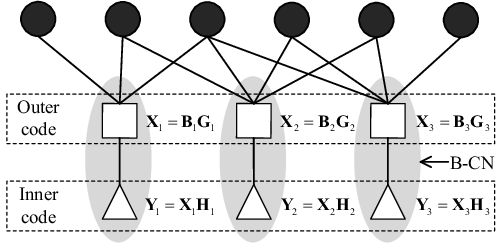}
		\label{}}
	\hfil
	\subfloat[]{\includegraphics[width=2.7in]{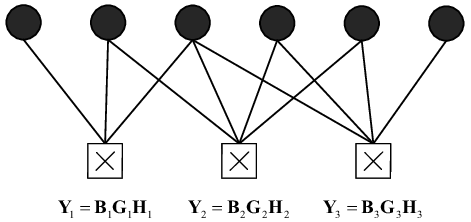}
		\label{}}
	\caption{Tanner graphs of a standard BATS code. (a) Encoding perspective. (b) Decoding perspective.}
	\label{fig:Tanner_std_BATS}
\end{figure*}

\subsection{BP Decoding and ML Decoding of BATS Codes}\label{subsection:decoding_std_BATS}
We first introduce BP decoding of BATS codes. In this paper, since erasure channels are considered, a BP decoder is equivalent to a peeling decoder, which is commonly used to analyze the performance of BP decoding, and which operates based on the Tanner graph defined as Fig.~\ref{fig:Tanner_std_BATS}(b). The peeling decoder progressively removes nodes and edges from the graph. We describe the $K$-step BP decoding for a given number $n$ of batches. At each step $t$, we will remove one decodable input symbol and the edges connected to it. Meanwhile, the generator matrix ${\bf G}_i^{(t)}$, the index set ${\cal I}_i^{(t)}$, and the received vector ${\bf Y}_i^{(t)}$ for each batch $i$ will be updated. Initially, let ${\bf G}_i^{(0)} = {\bf G}_i$ and ${\bf Y}_i^{(0)} = {\bf Y}_i$ for all $i$. We say a batch is \textit{decodable} at step $t$ if ${\rm rk}({\bf G}_i^{(t)} {\bf H}_i) = |{\cal I}_i^{(t)}|$.  We say the $j$-th symbol $B_j$ is \textit{decodable} at step $t$ if there exists a decodable batch $i$ such that $j \in {\cal I}_i^{(t)}$. At each step $t$, the BP decoder operates:

i) Randomly choose a decodable input symbol $B_j$ and remove symbol $B_j$ (i.e., the VN).

ii) For each batch $i$, if $j \in {\cal I}_i^{(t)}$, then ${\cal I}_i^{(t+1)} = {\cal I}_i^{(t)} \backslash \{j\}$, ${\bf G}_i^{(t+1)}$ is formed by removing the row ${\bf g}$ corresponding to the $j$-th symbol $B_j$ in ${\bf G}_i^{(t)}$, and ${\bf Y}_i^{(t+1)} = {\bf Y}_i^{(t)} - B_j {\bf g} {\bf H}_i$; else ${\cal I}_i^{(t+1)} = {\cal I}_i^{(t)}$, ${\bf G}_i^{(t+1)} = {\bf G}_i^{(t)}$, and ${\bf Y}_i^{(t+1)} = {\bf Y}_i^{(t)}$.

iii) For each batch $i$, if $|{\cal I}_i^{(t+1)}| = 0$, then remove the batch $i$ (i.e., the B-CN).

The BP decoding process stops when there is no decodable input symbol. If the BP decoding stops before the step-$K$, there remain undecoded input symbols. In this case, the remaining input symbols are still possible to be recovered by the \textit{ML decoding using Gaussian elimination}, which combines all remaining batches to form a single linear system and solves the linear equations. However, the computational complexity of Gaussian elimination is quite high, proportional to the cube of the number of remaining input symbols. Alternatively, inactivation decoding is an efficient ML erasure decoding for sparse linear systems~\cite{Raptor_codes_foundations_and_trends,InactivationPatent}. 

The decoding process of inactivation decoding is the same as the BP decoding until there is no decodable input symbol in the residual graph. When BP decoding stops at a non-empty residual graph, the inactivation decoder tries to resume the BP decoding process by ``inactivating'' a certain undecoded input symbol. The inactive symbol is marked as ``decoded'' and then removed from the graph. The decoder continuously inactivate undecoded input symbols until decodable input symbols appear. Then the BP decoding continues. The above process repeats until the graph is empty. Let $I_{\rm ina}$ denote the number of inactive symbols after decoding, and $b_1,b_2,\ldots,b_{I_{\rm ina}}$ denote the inactive symbols. If $I_{\rm ina} = 0$, inactivation decoding is the same as BP decoding without inactivation and all input symbols are successfully decoded. Otherwise, each input symbol $B_i$ can be expressed as
\begin{equation}\label{eq_inactive_symbol}
	B_i = c_{i,0} + \sum_{j=1}^{I_{\rm ina}} c_{i,j} b_j,
\end{equation}
where $c_{i,j} \in \mathbb{F}_q$ is the coefficient determined by the decoding process. After decoding, the batches with inactive symbols (i.e., the batch has an input symbol $B_i$ expressed with at least one non-zero coefficient $c_{i,j}$, $j \ge 1$) impose linear constraints on inactive symbols. This linear system of inactive symbols is usually solved by Gaussian elimination. If the inactive symbols can be uniquely solved, all input symbols can be recovered by substituting the inactive symbols back into (\ref{eq_inactive_symbol}) for all $i$.

\subsection{LDPC-BATS Codes: Encoding, Decoding, and Tanner Graphs}\label{sec_preliminaries_LDPCBATS}
LDPC codes are commonly used as a precode for fountain codes, e.g., they have been adopted as one of the precodes\footnote{Note that R10 and RaptorQ codes adopt both an LDPC precode and a high-density parity-check precode, and thus the actual precode is a serial concatenation of the two precodes.} in standardized R10 and RaptorQ codes~\cite{Raptor_codes_foundations_and_trends}. For the theoretical analysis in this paper, we only consider regular LDPC codes as the precode. However, when it comes to encoding and decoding algorithms, we can easily provide a unified description for LDPC-BATS codes, whether regular or irregular LDPC codes are used as the precode.

Consider a $(K,K')$ LDPC code over the same base field $\mathbb{F}_q$ with dimension $K'$ and code rate $R' = \frac{K'}{K}$. In the rest of this paper, we use \textit{input symbols} and \textit{intermediate symbols} to respectively represent the symbols before and after precoding, avoiding any confusion on the coding efficiency. The encoding process of an LDPC-BATS code includes the following steps: 

i) Encode $K'$ input symbols using the $(K,K')$ LDPC code to $K$ LDPC-coded symbols (i.e., intermediate symbols).

ii) Encode the $K$ LDPC-coded symbols into $n$ batches using the encoding process described in Sec.~\ref{subsec:BATS}. 

Considering the overhead of the precode, the \textit{rate} of an LDPC-BATS code is defined as $R = R'K/n = K'/n$. For an $(\mathtt{l},\mathtt{r})$-regular LDPC code (i.e., the parity-check matrix of the LDPC code has the same column weight $\mathtt{l}$ and the same row weight $\mathtt{r}$), the rate satisfies $R' \ge 1-{\mathtt{l}/\mathtt{r}}$. This bound for $R'$ is generally tight, and we can easily construct a regular LDPC code exactly with rate $1-{\mathtt{l}/\mathtt{r}}$. When a regular LDPC code is used as precode, the LDPC-BATS code is denoted by $\tilde{\mathscr{C}}_{\rm pre}^{(q)}(K,n,M,{\bm \Psi},{\bf h},\mathtt{l},\mathtt{r})$. 

There is an equivalence between a \textit{batch equation} and an \textit{LDPC parity-check equation}. This equivalence implies that the BP decoder for any LDPC-BATS codes operates in the same way as that for standard BATS codes. Let ${\bf Q} = [q_{i,j}]_{1 \le i \le K'-K, 1 \le j \le K'}$ be the parity-check matrix of the LDPC code, where $q_{i,j} \in \mathbb{F}_q$. 
Let $\mathtt{r}_i = \sum_{j=1}^{K} \mathds{1}_{\{q_{i,j} \ne 0\}}$ be the weight of the $i$-th row of $\bf Q$, where $\mathds{1}_{\{E\}}$ is the indicator function, which equals one if the event $E$ is true and otherwise zero. The $i$-th parity-check equation is equivalent to the batch equation ${\bf B}_{i}' {\bf G}_{i}' {\bf H}_{i}' = {\bf Y}_{i}'$, where ${\bf B}_{i}'$ is the row vector of $\mathtt{r}_i$ associated intermediate symbols, ${\bf G}_{i}'$ is a column vector consisting of $\mathtt{r}_i$ non-zero elements of the $i$-th row of $\bf Q$, ${\bf H}_{i}'$ is the constant integer $1$, and ${\bf Y}_{i}'$ is the constant integer $0$. Hence, concatenating an LDPC precode is equivalent to adding $K-K'$ extra ``batches'' of size one into the original BATS code. 

By setting ${\bf B}_{i+n} = {\bf B}_{i}'$, ${\bf G}_{i+n} = {\bf G}_{i}'$, ${\bf H}_{i+n} = {\bf H}_{i}'$, and ${\bf Y}_{i+n} = {\bf Y}_{i}'$ for $i \in [1:K-K']$, we have the Tanner graph of an LDPC-BATS code, shown in Fig.~\ref{fig:Tanner_LDPCBATS}. In the Tanner graph of an LDPC-BATS code, the extra $K-K'$ ``batches equations'', which are essentially LDPC parity-check equations, are referred to as LDPC CNs (L-CNs). Based on such a Tanner graph, the BP decoding introduced in Sec.~\ref{subsection:decoding_std_BATS} can jointly decode $K-K'+n$ ``batches'', but one should note that the concept of input symbols in Sec.~\ref{subsection:decoding_std_BATS} needs to be replaced by intermediate symbols. If BP decoding fails, inactivation decoding of an LDPC-BATS code also has the same process as the decoding process described in Sec.~\ref{subsection:decoding_std_BATS}.

\begin{figure}[!t]
	\centering
	\includegraphics[width=3in]{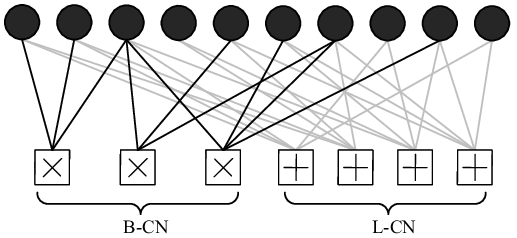}
	\caption{The Tanner graph of an LDPC-BATS code with a $(2,5)$-regular LDPC precode.}
	\label{fig:Tanner_LDPCBATS}
\end{figure}

\subsection{Ensembles}\label{subsction:ensembles}
Following \cite{BATS}, the previous Sec.~\ref{subsec:BATS} introduces BATS codes from the perspective of random codes, i.e., following a probability model. So, a decoder observes an instance of the BATS code (i.e., a specific Tanner graph) with a certain probability. Equivalently, we can define a set consisting of specific Tanner graphs, called the \textit{ensemble}\cite{Ridchardson2001,FL_analysis_LDPC,Weight_distribution_LDPC,Weight_distribution_NBLDPC}, and then a realization of a random code becomes a sampling from the ensemble\footnote{Probability distribution may be assigned to the elements in an ensemble, so the sampling is not necessarily uniform. But for the ensembles defined in this paper, uniform distribution for the elements in the ensemble is preserved.}. Although there is no essential difference between the probability model-based analysis and the ensemble-based analysis of a random code, the latter facilitates a clear expression of the finite-length analysis in this paper.

The $(K,n,M,{\bm \Psi},{\bf h})_q$ BATS ensemble equivalent to the standard BATS code $\mathscr{C}_{\rm std}^{(q)}(K,n,M,{\bm \Psi},{\bf h})$ is defined as follows.

\begin{definition}[BATS Ensemble]\label{def:BATS_ensemble}
	The $(K,n,M,{\bm \Psi},{\bf h})_q$ BATS ensemble is a set of Tanner graphs of the form shown in Fig.~\ref{fig:Tanner_std_BATS}(b). Each Tanner graph in this ensemble has $K$ VNs and $n$ B-CNs. The $(K,0,M,{\bm \Psi},{\bf h})_q$ BATS ensemble contains a single element: a trivial graph consisting only of $K$ VNs. To sample an element $\cal T$ from the $(K,n,M,{\bm \Psi},{\bf h})_q$ BATS ensemble, first sample an element ${\cal T}'$ from the $(K,n-1,M,{\bm \Psi},{\bf h})_q$ BATS ensemble. Next, sample an integer ${\rm dg}$ from $[1:K]$ with $\Pr\{{\rm dg} = d\} = \Psi_d$, and add a B-CN to ${\cal T}'$ to connect ${\rm dg}$ distinct VNs that are chosen uniformly at random. Then, sample a ${\rm dg} \times M$ matrix $\bf G$ from $\mathbb{F}_q^{{\rm dg} \times M}$ uniformly at random. Subsequently, sample an $M$-row matrix $\bf H$ from $\bigcup_{i=0}^{\infty} \mathbb{F}_q^{M \times i}$ with $\Pr\{{\rm rk}({\bf H}) = k\} = h_k$. Finally, associate $\bf G$ and $\bf H$ with the newly added B-CN.
\end{definition}

For a practical LDPC-BATS code, usually the LDPC precode is fixed, which may be specifically constructed to avoid small weaknesses in the graph. However, the finite-length analysis for a fixed LDPC code is difficult. Instead, we analyze the finite-length performance of a \textit{random LDPC-BATS code} whose LDPC precode is chosen from a $q$-ary regular LDPC ensemble uniformly at random. Here, the $q$-ary regular LDPC ensemble is the same as the well-known binary regular LDPC ensemble~\cite[Sec. II-A]{Ridchardson2001} except that each edge is labeled by an element uniformly chosen from $\mathbb{F}_q \setminus \{0\}$. To be specific, the $q$-ary regular LDPC ensemble with $K$ VNs, VN degree $\mathtt{l}$, and CN degree $\mathtt{r}$ is called the $(K,\mathtt{l},\mathtt{r})_q$ LDPC ensemble.

Furthermore, we can \textit{expurgate} the $(K,\mathtt{l},\mathtt{r})_q$ LDPC ensemble to eliminate all stopping sets (see \cite{FL_analysis_LDPC} for the definition of stopping sets) of size in the range $[1:\nu_{\rm min}-1]$ by declaring some bits in the stopping set to be ``dummy'' bits~\cite{Weight_distribution_LDPC}.\footnote{Note that expurgation can be realized by different methods. In \cite{1023273,Johnson2009finite}, expurgation is the direct removal of some specific Tanner graphs in the original ensemble, so the cardinality of the expurgated ensemble is smaller than the original ensemble. In \cite{Weight_distribution_LDPC}, expurgation is accomplished by introducing ``dummy'' bits. If a randomly chosen bit in a stopping set of size $s$ is declared to be ``dummy'', this stopping set of size $s$ disappears. The expurgated ensemble constructed by the latter method has the same cardinality as the original ensemble.} This ensemble is denoted by the $(K,\mathtt{l},\mathtt{r},\nu_{\rm min})_q$ \textit{expurgated} LDPC ensemble.

We define the $(K,n,M,{\bm \Psi},{\bf h},\mathtt{l},\mathtt{r},\nu_{\rm min})_q$ LDPC-BATS ensemble as follows. The random LDPC-BATS code equivalent to this ensemble is denoted by ${\mathscr{C}}_{\rm pre}^{(q)}(K,n,M,{\bm \Psi},{\bf h},\mathtt{l},\mathtt{r},\nu_{\rm min})$.

\begin{definition}[LDPC-BATS Ensemble]\label{def:LDPC_BATS_ensemble}
	The $(K,n,M,{\bm \Psi},{\bf h},\mathtt{l},\mathtt{r},\nu_{\rm min})_q$ LDPC-BATS ensemble is a set of Tanner graphs of the form shown in Fig.~\ref{fig:Tanner_LDPCBATS}. Each Tanner graph in this ensemble has $K$ VNs, $n$ B-CNs, and $K\mathtt{l}/\mathtt{r}$ L-CNs. To sample an element $\tilde{\cal T}$ from this ensemble, first sample a Tanner graph $\cal G$ from the $(K,\mathtt{l},\mathtt{r},\nu_{\rm min})_q$ expurgated LDPC ensemble. Then, sample a Tanner graph $\cal T$ from the $(K,n,M,{\bm \Psi},{\bf h})_q$ BATS ensemble. Finally, $\tilde{\cal T}$ is obtained by sharing the $K$ VNs in ${\cal G}$ and ${\cal T}$.
\end{definition}

At the end of this subsection, we briefly discuss reasonable choices of the expurgation parameter $\nu_{\rm min}$. The following lemma gives the expected stopping set distribution for the regular LDPC ensemble.
\begin{lemma}[\cite{Weight_distribution_LDPC}]
	Consider a graph $\cal G$ in the $(K,\mathtt{l},\mathtt{r})_q$ LDPC ensemble. Let $A_{\rm ss}({\cal G},w)$ be the number of stopping sets of size $w$ in the code defined by $\cal G$. Then,
	\begin{equation}\label{eq_expected_ss_dist}
		\mathbb{E}\left[ A_{\rm ss}({\cal G},w) \right] = \binom{K}{w}\frac{{\rm coef}\left\{\left((1+x)^\mathtt{r}-\mathtt{r}x\right)^{K\mathtt{l}/\mathtt{r}}, x^{w\mathtt{l}}\right\}}{\binom{n\mathtt{l}}{w\mathtt{l}}}.
	\end{equation}
\end{lemma}

Note that the original version of the above lemma is for the $(K,\mathtt{l},\mathtt{r})_2$ binary LDPC ensemble, but changing binary to $q$-ary does not affect this result\footnote{This modification is equivalent to multiplying both the numerator and denominator in (\ref{eq_expected_ss_dist}) by a factor of $(q-1)^{w\mathtt{l}}$.}. As pointed out in \cite[p. 138]{Modern_Coding_Theory}, there exist graphs in the $(K,\mathtt{l},\mathtt{r})_q$ LDPC ensemble which do not contain stopping sets of size smaller than $\nu$ if $\sum_{w=1}^{\nu-1} \mathbb{E}\left[ A_{\rm ss}({\cal G},w) \right] < 1$. Therefore, the maximum expurgation parameter we may choose is
\begin{equation}\label{eq:max_nu_min}
	\nu_{\rm min}^* \triangleq \sup\left\{\nu \in \mathbb{Z}^+: \sum_{w=1}^{\nu-1} \mathbb{E}\left[ A_{\rm ss}({\cal G},w) \right] < 1\right\}.
\end{equation}

\section{Performance Bounds for BP and ML Decoding}\label{Sec_main_results}
This section presents two upper bounds on the error probability of LDPC-BATS codes: Theorem \ref{theorem_LDPC_BATS_BP} about a bound for the error probability of random LDPC-BATS codes under BP decoding, and Theorem \ref{theorem_LDPC_BATS_INA} about a bound for the error probability of random LDPC-BATS codes under ML decoding. When the regular LDPC precode has rate $1$, i.e., set $\mathtt{l} = 0$, the aforementioned bounds reduce to the bounds for standard BATS codes. Furthermore, Theorem \ref{theorem_LDPC_BATS_INA} can, in fact, be applied to BATS codes with \textit{any} precode as long as the precode's weight enumerator is known, because the precode does not play a role in the proof of Theorem \ref{theorem_LDPC_BATS_INA} except for its weight enumerator.

The numerical results to verify these bounds are provided in Sec.~\ref{subsection:numerical_1}--\ref{subsection:numerical_3}. Subsequently, in Sec.~\ref{Sec_opt_DD}, we present a practical application of these two upper bounds: Optimizing the degree distribution $\bm \Psi$ for LDPC-BATS codes.

\subsection{Upper Bounds}
Before presenting the performance upper bounds, we introduce some useful functions and symbols. Recalling that ${\bf \Psi} = (\Psi_1,\Psi_2,\ldots,\Psi_K)$ is the degree distribution of BATS codes, let $D_{\rm min}$ be the smallest degree $d$ with $\Psi_d > 0$. Recall that ${\bf h} = (h_0,h_1,\ldots,h_M)$ is the \textit{empirical rank distribution} of the transfer matrix, i.e., $\Pr\{{\rm rk}({\bf H}_i) = r\} = h_r$, and assume $\bf h$ is known. Define 
\begin{equation}\label{eq_zeta_function}
	\begin{small}
	\begin{aligned}
	\zeta_r^m = \begin{cases}
		(1-q^{-m})(1-q^{-m+1})\cdots(1-q^{-m+r-1}) & \text{if }r > 0,\\
		1 &\text{if }r=0,
	\end{cases}
	\end{aligned}
	\end{small}
\end{equation}
and
\begin{equation}\label{eq_hbar_prime}
	\hbar_k' = \begin{cases}
		\sum_{s=k}^{M} \zeta_k^s h_s & \text{if }k \in [0:M],\\
		0 & \text{if }k > M.
	\end{cases}
\end{equation}	
According to \cite{BATS,FL_analysis_BATS}, $\zeta_r^m$ is the probability that $r$ vectors independently picked from $\mathbb{F}_q^m$ are linearly independent, and $\hbar_k'$ is the probability that a batch with degree $k$ is decodable.

The p.m.f. of the hypergeometric distribution is denoted by
\begin{multline}\nonumber
			{\rm hyge}(k;n,i,j) \\ \triangleq 
			\begin{cases}
			\frac{\binom{i}{k}\binom{n-i}{j-k}}{\binom{n}{j}} & \text{if }\max\{0,i+j-n\} \le k \le \min\{i,j\},\\
			0 & \mbox{otherwise.}
		\end{cases}
\end{multline}

Let ${\rm coef}\left\{f(x),x^i\right\}$ denote the coefficient of $x^i$ in $f(x)$, e.g., ${\rm coef}\left\{2x^2+x, x^2\right\} = 2$. 

\begin{theorem}\label{theorem_LDPC_BATS_BP}
	Let $P_{\textsf{E}}^{\rm BP}({\mathscr{C}}_{\rm pre})$ denote the error probability of BP decoding of the random LDPC-BATS code ${\mathscr{C}}_{\rm pre}^{(q)}(K,n,M,{\bm \Psi},{\bf h},\mathtt{l},\mathtt{r},\nu_{\rm min})$. $P_{\textsf{E}}^{\rm BP}({\mathscr{C}}_{\rm pre})$ can be bounded by
	\begin{equation}\label{eq_theorem_2}
		\begin{aligned}
			P_{\textsf{E}}^{\rm BP}({\mathscr{C}}_{\rm pre}) \le& \min\biggg\{ 1,~\sum_{a = 1}^{K} \sum_{b = 0}^{n} \binom{K}{a} \binom{n}{b} \\
			&\cdot \left(\sum_{d = 1}^{K} \varphi_{d,K-a}\right)^{n-b} \left(\sum_{d=1}^{K} \Psi_d \sum_{k=1}^{\min\{a,d\}} \xi_{d,k} \right)^b \\
			&\cdot \biggg( \sum_{j=0}^{s_{\rm max}(a)} L(a,j;K,\mathtt{l},\mathtt{r}) \\
			&\cdot U(K-a-j;K-a,n-b,{\bm \Phi}^{(K-a)}) \biggg) \biggg\}.
		\end{aligned}
	\end{equation}
	Here, $s_{\rm max}(a) = \left\lfloor \frac{(K-a)\mathtt{l}}{\mathtt{r}} \right\rfloor$.	For $d \in [1:K]$, $k \in [0:K]$,
	\begin{equation}\label{eq_varphi}
		\varphi_{d,k} = \Psi_d \frac{\binom{k}{d}}{\binom{K}{d}},
	\end{equation}
	\begin{equation}\label{eq_xi}
		\xi_{d,k}= (1-\hbar_k') {\rm hyge}(k;K,a,d).
	\end{equation}
	For $k \in [0:K]$,
	\begin{equation}\label{eq_phi_vector}
		{\bm \Phi}^{(k)} = ({\Phi}_{0}^{(k)},{\Phi}_{1}^{(k)},{\Phi}_{2}^{(k)},\ldots,{\Phi}_{k}^{(k)}),
	\end{equation}
	where for $d \in [0:K]$,
	\begin{equation}\label{eq_Phi}
		{\Phi}_{d}^{(k)} = \begin{cases}
			\frac{\varphi_{d,k}}{\sum_{d' = 1}^{k} \varphi_{d',k}} &\text{if }d \ge 1 \text{ and } k \ge D_{\rm min},\\
			1 &\text{if }d=0 \text{ and } k < D_{\rm min},\\
			0 & \text{otherwise.}
		\end{cases} 
	\end{equation}
	For $j \in [0:m]$, $m,s \in \mathbb{Z}$,
	\begin{multline}\label{eq_U}
		U(j;m,s,{\bm \Phi}^{(m)}) \\ = 
		\binom{m}{j} \sum_{k=0}^{j} (-1)^{j-k} \binom{j}{k} \left(\sum_{d = 0}^{k} \Phi_d^{(m)} \frac{\binom{k}{d}}{\binom{m}{d}}\right)^s.
	\end{multline}	
	For $j \in [0:s_{\rm max}(a)]$,
	\begin{multline}
		L(a,j;K,\mathtt{l},\mathtt{r}) \\= \sum_{c = 0}^{K\mathtt{l}/\mathtt{r}} \min \Bigg\{ \binom{K\mathtt{l}/\mathtt{r}}{c} \hat{p}(a,0,c),
		\binom{K\mathtt{l}/\mathtt{r}}{j+c} \hat{p}(a,j,c) \Bigg\},
	\end{multline}
where
\begin{equation}\label{eq:hat_p}
	\begin{aligned}
		&\hat{p}(a,a',c) =\\
		&\begin{cases}
			\begin{footnotesize}	
			\begin{aligned}
			&\frac{{\rm coef}\{ ((1+x)^\mathtt{r}-1-\mathtt{r}x)^{c},x^{a\mathtt{l}} \}}{\binom{K\mathtt{l}}{a\mathtt{l}}} \\
			&\cdot \prod_{t=1}^{a'} \frac{t(c+t) \mathtt{r} \binom{(K\mathtt{l}/\mathtt{r}-a'+t-1)\mathtt{r}-(a+t-1)\mathtt{l}}{\mathtt{l}-1} (\mathtt{l})!}{\prod_{i=0}^{\mathtt{l}-1} (K\mathtt{l}-(a+t-1)\mathtt{l}-i)} 
			\end{aligned}
		    \end{footnotesize}
			&\begin{aligned}
				&\text{if }a \ge \nu_{\rm min}\\ &~~\text{and } a' \ge 1,
			\end{aligned}\\[1cm]
			\begin{footnotesize}
			\begin{aligned}
			\frac{{\rm coef}\{ ((1+x)^\mathtt{r}-1-\mathtt{r}x)^{c},x^{a\mathtt{l}} \}}{\binom{K\mathtt{l}}{a\mathtt{l}}} 
			\end{aligned}
		    \end{footnotesize}
			&\begin{aligned}
					&\text{if }a \ge \nu_{\rm min}\\ &~~\text{and } a' = 0,
			\end{aligned}\\
			0 & \text{if }a < \nu_{\rm min}.	
		\end{cases}
	\end{aligned}
\end{equation}
\end{theorem}
\begin{proof}
	The proof is given in Appendix~\ref{Sec_Proof_Theorem_BP}. Here is the sketch. First of all, we need to understand what kind of graph structures lead to BP decoding failures. Assume that $\cal A$ is the set containing the remaining VNs after BP (i.e., peeling) decoding. On one hand, we immediately know that $\cal A$ is a stopping set \cite{FL_analysis_LDPC} of the LDPC precode; otherwise, some VNs in $\cal A$ can be decoded by the LDPC precode. On the other hand, $\cal A$ does not have neighboring B-CNs that are decodable. Informally, we say that $\cal A$ is also a ``stopping set of the BATS code''. Assume that ${\bf Y} = {\bf B} {\bf G} {\bf H}$ is the original batch equation of a certain B-CN that is a neighbor of $\cal A$. After peeling all VNs except for those in ${\cal A}$, this batch equation becomes ${\bf Y} - {\bf B}^{({\cal A}^c)} {\bf G}^{({\cal A}^c)} {\bf H} = {\bf B}^{(\cal A)} {\bf G}^{(\cal A)} {\bf H}$, where ${\cal S}^c$ is the complement set of $\cal S$, and ${\bf B}^{(\cal S)}$ (resp. ${\bf G}^{(\cal S)}$) is the submatrix formed by the columns (resp. rows) of $\bf B$ (resp. $\bf G$) that have corresponding VNs in $\cal S$. To ensure that the modified batch equation is undecodable, we need ${\rm rk}({\bf G}^{(\cal A)} {\bf H}) < s$, where $s$ is the number of rows of ${\bf G}^{(\cal A)}$. Thus, each B-CN that is the neighbor of $\cal A$ needs to satisfy this constraint on the rank. At this point, we have comprehended the graph structures that lead to BP decoding failures.
	
	Subsequently, we can bound the probability that a randomly chosen $\cal A$ is a stopping set of both the LDPC code and the BATS code. Considering all choices of $\cal A$, we can obtain a union bound. However, such an upper bound may be loose, because the graph structure associated with VNs outside of $\cal A$ is not considered at all. To tighten the bound, we impose certain constraints on the graph structure outside of $\cal A$ to increase the likelihood that $\cal A$ is the ``maximal stopping set'' of the Tanner graph, thereby reducing the double counting in the union bound.
\end{proof}

\begin{remark}
	Unless otherwise specified, we use the convention $\binom{n}{k} = 0$ for $k>n$, $k<0$, or $n<0$. To cover the case $\mathtt{l} = 0$, we adopt the convention $\prod_{i=0}^{\mathtt{l-1}} x_i = 1$ if $\mathtt{l} \le 0$ for $\hat{p}(a,a',c)$.
\end{remark}

\begin{theorem}\label{theorem_LDPC_BATS_INA}
	Let $P_{\textsf{E}}^{\rm ML}({\mathscr{C}}_{\rm pre})$ denote the error probability of ML decoding of the random LDPC-BATS code ${\mathscr{C}}_{\rm pre}^{(q)}(K,n,M,{\bm \Psi},{\bf h},\mathtt{l},\mathtt{r},\nu_{\rm min}=1)$. $P_{\textsf{E}}^{\rm ML}({\mathscr{C}}_{\rm pre})$ can be bounded by
	\begin{equation}\label{eq:ML_UB}
		\begin{aligned}
			P_{\textsf{E}}^{\rm ML}({\mathscr{C}}_{\rm pre}) \le &\min\biggg\{1, \sum_{l = 1}^{K} \frac{A_l}{q-1} \left(\tilde{\pi}_l^n - \left(\sum_{d=1}^{K} \varphi_{d,K-l}\right)^n\right)  \\
			&+ \min\left\{\frac{A_l}{q-1},\binom{K}{l}\right\} \left(\sum_{d=1}^{K} \varphi_{d,K-l}\right)^n\biggg\},
		\end{aligned}
	\end{equation}
	where 
	\begin{equation}
		\tilde{\pi}_l = \sum_{d = 1}^{K} \Psi_d \left(\frac{\binom{K-l}{d}}{\binom{K}{d}} + \left(1-\frac{\binom{K-l}{d}}{\binom{K}{d}}\right) \left(\sum_{k = 0}^{M} h_k q^{-k}\right)\right),
	\end{equation}
	and $A_l$ is the expected weight enumerator for the $(K,\mathtt{l},\mathtt{r},1)_q$ LDPC ensemble, given below
	\begin{equation}\label{eq_WE}
		A_l = \binom{K}{l} \frac{{\rm coef}\left( \left(\frac{(1+(q-1)x)^{\mathtt{r}}+(q-1)(1-x)^{\mathtt{r}}}{q}\right)^{K\frac{\mathtt{l}}{\mathtt{r}}}, x^{\mathtt{l}l} \right)}{\binom{K\mathtt{l}}{\mathtt{l}l}(q-1)^{(\mathtt{l}-1)l}}.
	\end{equation}
\end{theorem}
\begin{proof}
	The proof is given in Appendix~\ref{Sec_proof_Theorem_3}. Here is the sketch. Let $\tilde{\bf G}_i'$ be the $K \times M$ matrix obtained by adding $K - {\rm dg}_i$ all-zero rows to ${\bf G}_i$ such that the ${\rm dg}_i$ rows of $\tilde{\bf G}_i'$ indexed by ${\cal I}_i$ form ${\bf G}_i$, where ${\cal I}_i$ contains the indices of B-CN $i$'s neighboring VNs, and ${\rm dg}_i = |{\cal I}_i|$. Thus, the batch equation for the batch $i$ can be written as ${\bf Y}_i = {\bf v} \tilde{\bf G}_i' {\bf H}_i$, where ${\bf v} = (v_1,v_2,\ldots,v_K)$ is a codeword in $\mathscr{C}$, and $\mathscr{C}$ is the codebook of the precode. Without loss of generality, assume the zero codeword is transmitted. We then see that the necessary and sufficient condition such that the ML decoding fails is $\exists {{\bf v} \in \mathscr{C} \setminus \{{\bf 0}\}}, {\bf v} \tilde{\bf G}_i' {\bf H}_i = {\bf 0}, \forall i \in [1:n]$. Given $\bf v$, noting that this probability only depends on the Hamming weight $w_H({\bf v})$ of $\bf v$, it can be computed as $\left(\Pr\{{\bf v} \tilde{\bf G}_i' {\bf H}_i = {\bf 0} \mid w_H({\bf v}) = l\}\right)^n \triangleq {\tilde{\pi}}_l^n$. Utilizing some properties of random matrices, ${\tilde{\pi}}_l$ can be derived. Then, we can obtain a union bound of the form like $\sum_{l=1}^{K} A_l {\tilde{\pi}}_l^n$. By eliminating some double counting in the union bound, Theorem~\ref{theorem_LDPC_BATS_INA} follows.
\end{proof}

\begin{remark}
Since Theorem \ref{theorem_LDPC_BATS_INA} considers randomly choosing a precode from the LDPC ensemble, $A_l$ is the expected weight enumerator. If the precode is fixed, $A_l$ becomes the weight enumerator for the fixed precode.
\end{remark}

\begin{remark}
	According to the proof of Theorem \ref{theorem_LDPC_BATS_INA}, the specific choice of the precode does not play a role in the proof except for its weight enumerator. Thus, Theorem \ref{theorem_LDPC_BATS_INA} can, in fact, be applied to BATS codes with any precode as long as $A_l$ is known.
\end{remark}

\begin{remark}\label{remark:minimum_ML_bound}
	In our experiments, we observed that the minimum of the ML performance upper bound with respect to the degree distribution $\bm \Psi$ is always attained when $\Psi_K = 1$. Although this fact is not readily apparent from (\ref{eq:ML_UB}), we can consider a slightly relaxed upper bound to support it to some extent. By replacing the term $\min\left\{\frac{A_l}{q-1},\binom{K}{l}\right\}$ on the right hand side of (\ref{eq:ML_UB}) with $\frac{A_l}{q-1}$, we obtain a relaxed upper bound $P_{\textsf{E}}^{\rm ML}({\mathscr{C}}_{\rm pre}) \le \min\left\{1, \sum_{l=1}^{K} \frac{A_l}{q-1} \tilde{\pi}_l^n\right\}$. Since $\sum_{k = 0}^{M} h_k q^{-k}$ is a constant in $[0,1]$, it is straightforward to see that the minimum of $\tilde{\pi}_l$ is achieved when $\Psi_K = 1$ for all $l$.
\end{remark}

\begin{corollary}\label{corollary_std_BATS_BP}
	Let $P_{\textsf{E}}^{\rm BP}({\mathscr{C}}_{\rm std})$ denote the error probability of BP decoding of the standard BATS code $\mathscr{C}_{\rm std}^{(q)}(K,n,M,{\bm \Psi},{\bf h})$. $P_{\textsf{E}}^{\rm BP}({\mathscr{C}}_{\rm std})$ can be bounded by
	\begin{equation}
		\begin{aligned}
			P_{\textsf{E}}^{\rm BP}({\mathscr{C}}_{\rm std}) \le &\min\biggg\{ 1,~\sum_{a = 1}^{K} \sum_{b = 0}^{n} \binom{K}{a} \binom{n}{b} \\
			&\cdot \left(\sum_{d = 1}^{K} \varphi_{d,K-a}\right)^{n-b}  \left(\sum_{d=1}^{K} \Psi_d \sum_{k=1}^{\min\{a,d\}} \xi_{d,k} \right)^b \\
			&\cdot U(K-a;K-a,n-b,{\bm \Phi}^{(K-a)})\biggg\}.
		\end{aligned}
	\end{equation}
\end{corollary}
\begin{proof}
	If $\mathtt{l} = 0$, the random LDPC-BATS code reduces to the standard BATS code. By noting that the summation in the third and fourth lines of (\ref{eq_theorem_2}) is equal to $U(K-a;K-a,n-b,{\bm \Phi}^{(K-a)})$ for $\mathtt{l} = 0$, this corollary follows.
\end{proof}

\begin{corollary}\label{corollary_std_BATS_INA}
	Let $P_{\textsf{E}}^{\rm ML}({\mathscr{C}}_{\rm std})$ denote the error probability of ML decoding of the standard BATS code $\mathscr{C}_{\rm std}^{(q)}(K,n,M,{\bm \Psi},{\bf h})$. $P_{\textsf{E}}^{\rm ML}({\mathscr{C}}_{\rm std})$ can be bounded by
	\begin{multline}
		P_{\textsf{E}}^{\rm ML}({\mathscr{C}}_{\rm pre}) \le \\
		\min\biggg\{1, \sum_{l = 1}^{K} \binom{K}{l}(q-1)^{l-1} \left(\tilde{\pi}_l^n - \left(\sum_{d=1}^{K} \varphi_{d,K-l}\right)^n\right) \\
		+ \binom{K}{l} \left(\sum_{d=1}^{K} \varphi_{d,K-l}\right)^n\biggg\},
	\end{multline}
\end{corollary}
\begin{proof}
	Applying Theorem \ref{theorem_LDPC_BATS_INA} with $\mathtt{l}=0$, the expected weight enumerator given by (\ref{eq_WE}) reduces to $A_l = \binom{K}{l} (q-1)^l$, and then this corollary follows.
\end{proof}

\subsection{Numerical Results: BP Performance of Standard BATS Codes}\label{subsection:numerical_1}
The purpose of this subsection is to verify the tightness of the BP performance bound (Theorem~\ref{theorem_LDPC_BATS_BP}) for standard BATS codes, that is, LDPC-BATS codes with a trivial precode of rate $1$. As mentioned before, BP performance of standard BATS codes can be analyzed by the two recursive formulae in \cite{FL_analysis_BATS}, which provide the exact and the approximate BP performance, respectively. Therefore, we make a comparison among Theorem~\ref{theorem_LDPC_BATS_BP} and the two recursive formulae in \cite{FL_analysis_BATS}. Note that for standard BATS codes, Theorem~\ref{theorem_LDPC_BATS_BP} can be written in a more concise form, see Corollary~\ref{corollary_std_BATS_BP}.

In numerical evaluation, we consider a length-2 line network with packet loss rate $0.2$, and employ the RLNC at the intermediate node as the scheme provided in \cite[Sec. VII-A]{BATS}. All used degree distributions ${\bm \Psi}_i$ and empirical rank distributions ${\bf h}_i$ are given in Appendix~\ref{tables}.

\begin{example}\label{example_1_1}	
	This example considers the standard BATS code ${\mathscr{C}}_{\rm std}^{(256)}(128,n,16,{\bf \Psi}_1,{\bf h}_1)$ under BP decoding. The numerical results are given in Fig.~\ref{fig_pure_BATS_K128_performance}, where the three error probability curves are computed by the exact recursive formula \cite[Theorem 1]{FL_analysis_BATS}, the approximate recursive formula \cite[Theorem 16]{FL_analysis_BATS}, and the upper bound in Theorem~\ref{theorem_LDPC_BATS_BP}, respectively. 
	
	In the relatively-low error probability region (e.g., below $10^{-1}$), the approximate recursive formula overestimates the error probability, and the estimation bias increases as $n$ increases or error probability decreases. In Fig.~\ref{fig_pure_BATS_K128_performance}, at the error probability of $10^{-6}$, the approximate recursive formula overestimates the error probability by about one order of magnitude. While, the upper bound provides an excellent match to the actual error probability (i.e., the exact recursive formula) for error probability below $10^{-1}$. 
	
	Moreover, we can expect that the upper bound is more efficient to compute than the exact recursive formula for large $n$, because the computational complexity of the upper bound is $\Theta(n)$ while that of the exact recursive formula is $\Theta(n^2)$. In Table~\ref{table_cpu_time}, we provide the CPU time for computing the two recursive formulae and the upper bound in MATLAB, utilizing an i9-12900K CPU. We observe that the computation of the approximate recursive formula has the fastest execution speed, with a computational complexity independent of $n$. The CPU time for the upper bound follows closely behind the CPU time for the approximate recursive formula with only a small gap. The computation of the exact recursive formula, on the other hand, is the slowest, requiring several times the CPU time for the other two algorithms. As a conclusion, the upper bound for standard BATS codes achieves a good tradeoff between accuracy and computational complexity.
\end{example}

\begin{figure}[!t]
	\centering
	\includegraphics[width=3.5in]{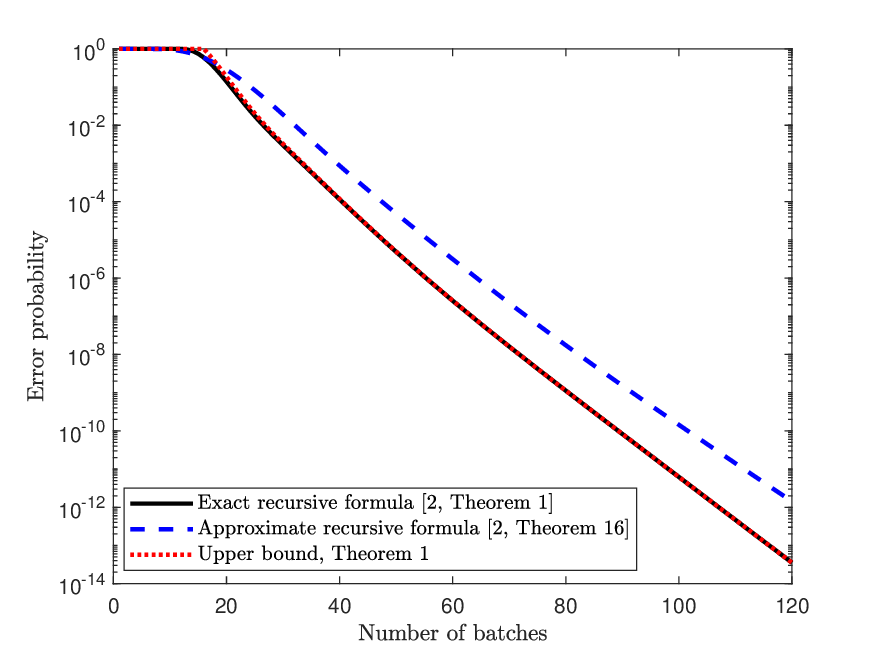}
	\caption{The error probabilities of the standard BATS code ${\mathscr{C}}_{\rm std}^{(256)}(128,n,16,{\bf \Psi}_1,{\bf h}_1)$ under BP decoding. The error probabilities are computed by the two recursive formulae in \cite{FL_analysis_BATS} and Theorem~\ref{theorem_LDPC_BATS_BP}.}
	\label{fig_pure_BATS_K128_performance}
\end{figure}

\begin{table*}[t]
	\caption{CPU Time for Computing the Upper Bound, the Exact Recursive Formula, and the Approximate Recursive Formula for ${\mathscr{C}}_{\rm std}^{(256)}(128,n,16,{\bf \Psi}_1,{\bf h}_1)$ BATS Code in MATLAB.}
	\normalsize
	\centering
	\rowcolors{1}{lightgray!20}{}
	\begin{tabular}{lllllll}
		\hline
		Algorithm / $n$  &  $20$ &   $30$ &    $40$  &  $50$ &   $75$&    $100$ \\
		Upper bound & 1.784 s  &  1.807 s  &   1.902 s   & 1.986 s  & 2.302 s  & 2.520 s \\
		Exact recursive formula & 2.244 s  &  3.883 s  &  6.517 s    &  10.959 s & 31.799 s  & 74.400 s \\
		Approximate recursive formula & 1.210 s  &  1.214 s  &  1.211 s    &  1.220 s & 1.210 s  & 1.213 s \\
		\hline
	\end{tabular}\label{table_cpu_time}
\end{table*}

\begin{example}\label{example_1_2}	
	This example verifies the BP performance bound for the standard BATS code ${\mathscr{C}}_{\rm std}^{(256)}(256,n,16,{\bf \Psi}_2,{\bf h}_1)$. The numerical results are shown in Fig.~\ref{fig_pure_BATS_K256_performance}, where the three error probability curves are obtained in the same ways as Example~\ref{example_1_1}. 
	
	In Fig.~\ref{fig_pure_BATS_K256_performance}, we can observe a phenomenon similar to that in Example~\ref{example_1_1}. When the error probability is below $10^{-1}$, the upper bound is much closer to the accurate result compared to the approximate recursive formula. However, as $K$ increases, the upper bound aligns with the exact recursive formula only when the error probability is below $10^{-2}$, while in Example~\ref{example_1_1} the upper bound is sufficiently accurate for the error probability below $10^{-1}$. It is a common phenomenon for the union bound to become looser as the code length increases.
\end{example}

\begin{figure}[!t]
	\centering
	\includegraphics[width=3.5in]{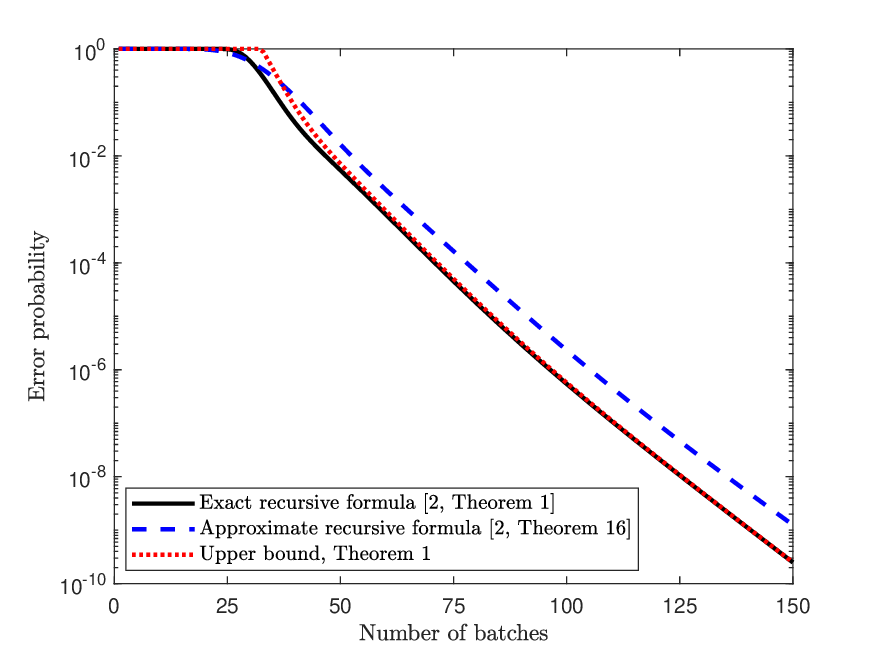}
	\caption{The error probabilities of the standard BATS code ${\mathscr{C}}_{\rm std}^{(256)}(256,n,16,{\bf \Psi}_2,{\bf h}_1)$ under BP decoding. The error probabilities are computed by the two recursive formulae in \cite{FL_analysis_BATS} and Theorem~\ref{theorem_LDPC_BATS_BP}.}
	\label{fig_pure_BATS_K256_performance}
\end{figure}

\subsection{Numerical Results: BP Performance of LDPC-BATS Codes}\label{subsection:numerical_2}
The goal of this subsection is to verify the BP performance bound (Theorem~\ref{theorem_LDPC_BATS_BP}) for LDPC-BATS codes. The network is the same as that in Sec.~\ref{subsection:numerical_1} and the empirical rank distributions ${\bf h}_i$ can be found in Appendix~\ref{tables}. The degree distributions used in this subsection are optimized using the method presented in Sec.~\ref{Sec_opt_DD} and the details of them can be found in Appendix~\ref{tables}.

Due to the absence of any analytical formulae for LDPC-BATS codes at the time of writing this paper, we employ Monte Carlo simulations in the following two examples to validate the accuracy of Theorem~\ref{theorem_LDPC_BATS_BP}. The simulation settings merit some discussion. To verify the bound, we perform \textit{simulation for random LDPC-BATS codes}, where more than $10,000$ LDPC codes in the $(K,\mathtt{l},\mathtt{r},\nu_{\rm min})_q$ LDPC ensemble are pre-generated and each transmission randomly choose one to be the precode. In addition, noting that practical communication systems usually employ well-constructed LDPC codes rather than randomly choosing codes from the ensemble, we also consider LDPC-BATS codes with a regular LDPC precode constructed by the PEG algorithm~\cite{PEG}. With some abuse of notation, let $\tilde{\mathscr{C}}_{\rm pre}^{(q)}(K,n,M,{\bm \Psi},{\bf h},\mathtt{l},\mathtt{r})$ denote an LDPC-BATS code with a fixed $(\mathtt{l},\mathtt{r})$-regular LDPC precode constructed by the PEG algorithm. The simulations for such codes are denoted as \textit{simulation for PEG-LDPC-BATS} in the legends of Fig.~\ref{fig_LDPC_BATS_K128_performance} and Fig.~\ref{fig_LDPC_BATS_K256_performance}. 

However, Theorem~\ref{theorem_LDPC_BATS_BP} is only an upper bound for ${\mathscr{C}}_{\rm pre}^{(q)}(K,n,M,{\bf \Psi},{\bf h},\mathtt{l},\mathtt{r},\nu_{\rm min})$, but may not necessarily apply as an upper bound for $\tilde{\mathscr{C}}_{\rm pre}^{(q)}(K,n,M,{\bm \Psi},{\bf h},\mathtt{l},\mathtt{r})$. As illustrated in the following two examples, we can adjust the expurgation parameter $\nu_{\rm min}$ in the upper bound to estimate the performance of $\tilde{\mathscr{C}}_{\rm pre}^{(q)}(K,n,M,{\bm \Psi},{\bf h},\mathtt{l},\mathtt{r})$. Eq. (\ref{eq:max_nu_min}) offers a conservative way to determine $\nu_{\rm min}$.\footnote{Empirically speaking, the size of the minimum stopping set in a PEG-LDPC code is likely to be larger than $\nu_{\rm min}^*$ by (\ref{eq:max_nu_min}).}

\begin{example}\label{Example_LDPC_BATS_BP_1}
	This example considers LDPC-BATS codes with $K = 128$, $M = 16$, and $(3,6)$-regular precoding (including random LDPC precoding and PEG-LDPC precoding). More precisely, the codes are ${\mathscr{C}}_{\rm pre}^{(256)}(128,n,16,{\bf \Psi}_3/{\bf \Psi}_{17},{\bf h}_1,3,6,1/4)$ and $\tilde{\mathscr{C}}_{\rm pre}^{(256)}(128,n,16,{\bm \Psi}_3,{\bf h}_1,3,6)$ (here, $a/b$ stands for $a$ or $b$). For these codes, the number of input symbols is approximately $K'=64$. 
	
	The upper bounds and simulation results are shown in Fig.~\ref{fig_LDPC_BATS_K128_performance}. We see that the accuracy of the upper bound may vary with the degree distribution. For ${\bf \Psi}_3$, the upper bounds are tight compared to the simulation results when the error probability is below $10^{-1}$. However, for ${\bf \Psi}_{17}$, noting that the error probability of this code decreases more rapidly, the upper bound is tight below the error probability of $0.003$. As mentioned before, the upper bound with $\nu_{\rm min} = 4$ may be regarded as an approximation for $\tilde{\mathscr{C}}_{\rm pre}^{(256)}(128,n,16,{\bm \Psi}_3,{\bf h}_1,3,6)$. We see that this approximation well matches the actual performance for ${\bm \Psi}_3$.
\end{example}

\begin{figure}[!t]
	\centering
	\includegraphics[width=3.5in]{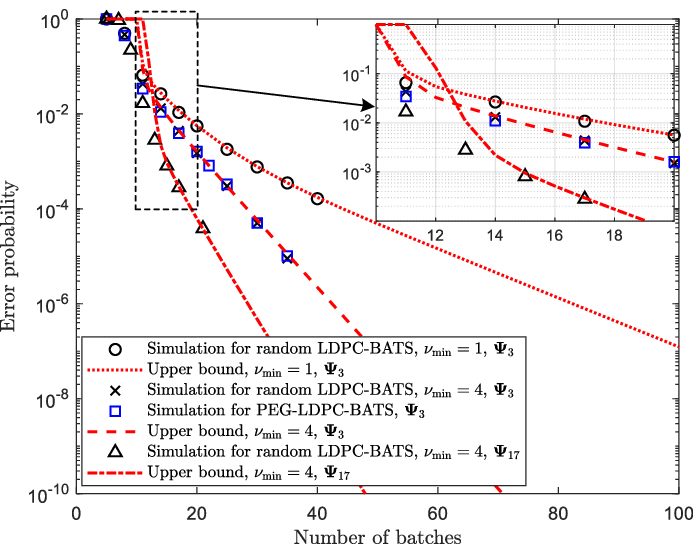}
	\caption{The simulation results and the upper bounds (Theorem \ref{theorem_LDPC_BATS_BP}) for the LDPC-BATS codes ${\mathscr{C}}_{\rm pre}^{(256)}(128,n,16,{\bf \Psi}_3/{\bf \Psi}_{17},{\bf h}_1,3,6,1/4)$ and $\tilde{\mathscr{C}}_{\rm pre}^{(256)}(128,n,16,{\bm \Psi}_3,{\bf h}_1,3,6)$ under BP decoding.}
	\label{fig_LDPC_BATS_K128_performance}
\end{figure}

\begin{example}\label{Example_LDPC_BATS_BP_2}	
	This example verifies the accuracy of the BP performance bound for LDPC-BATS codes with $K = 256$, $M = 16$, and $(3,6)$-regular precoding (including random LDPC precoding and PEG-LDPC precoding). The codes in this example are ${\mathscr{C}}_{\rm pre}^{(256)}(256,n,16,{\bf \Psi}_4/{\bf \Psi}_5,{\bf h}_1,3,6,1/7)$ and $\tilde{\mathscr{C}}_{\rm pre}^{(256)}(256,n,16,{\bf \Psi}_4/{\bf \Psi}_5,{\bf h}_1,3,6)$ (here, $a/b$ stands for $a$ or $b$).
	
	The simulation results and the upper bounds are shown in Fig.~\ref{fig_LDPC_BATS_K256_performance}. For ${\bf \Psi}_4$, when the error probability is below $0.03$, the upper bound is fairly accurate for the random LDPC-BATS code regardless of whether $\nu_{\rm min} = 1$ or $\nu_{\rm min} = 7$. When $\nu_{\rm min} = 7$, the upper bound also aligns well with the performance of the LDPC-BATS code with a PEG-LDPC precode. However, for ${\bf \Psi}_5$, the upper bound does not become tight until the error probability is below $10^{-4}$. This may be because in this case, the error probability of the code decreases extremely rapidly (on the other hand, the gap measured by $n$ is not large). 
	
	Additionally, we observe that when $\nu_{\rm min}$ is set according to (\ref{eq:max_nu_min}), Theorem~\ref{theorem_LDPC_BATS_BP} may not always be an accurate approximation for PEG-LDPC-precoded BATS codes. For ${\bf \Psi}_5$, the upper bound excessively overestimates the error probability of $\tilde{\mathscr{C}}_{\rm pre}^{(256)}(256,n,16,{\bf \Psi}_5,{\bf h}_1,3,6)$ even when the error probability is low to $10^{-6}$. Although we can continue to increase $\nu_{\rm min}$ to make the upper bound a better approximation, it does not have much theoretical significance.
\end{example}

\begin{figure}[!t]
	\centering
	\includegraphics[width=3.5in]{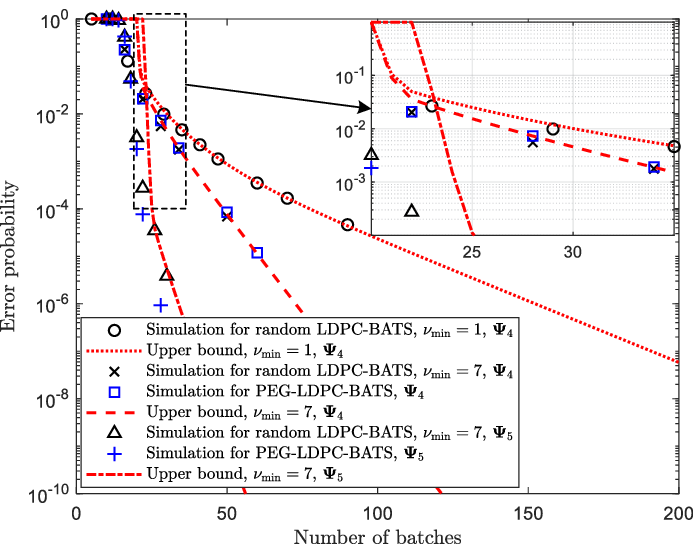}
	\caption{The simulation results and the upper bounds (Theorem \ref{theorem_LDPC_BATS_BP}) for the LDPC-BATS codes ${\mathscr{C}}_{\rm pre}^{(256)}(256,n,16,{\bf \Psi}_4/{\bf \Psi}_5,{\bf h}_1,3,6,1/7)$ and $\tilde{\mathscr{C}}_{\rm pre}^{(256)}(256,n,16,{\bf \Psi}_4/{\bf \Psi}_5,{\bf h}_1,3,6)$ under BP decoding.}
	\label{fig_LDPC_BATS_K256_performance}
\end{figure}

\subsection{Numerical Results: ML Performance of LDPC-BATS Codes}\label{subsection:numerical_3}
This subsection aims to verify the tightness of the ML performance bound (Theorem~\ref{theorem_LDPC_BATS_INA}) for LDPC-BATS codes. The network is the same as that in Sec.~\ref{subsection:numerical_1} and the empirical rank distributions ${\bf h}_i$ can be found in Appendix~\ref{tables}. The details of all degree distributions used in this subsection are given in Appendix~\ref{tables} and most of them are optimized using the method introduced in Sec.~\ref{Sec_opt_DD}.

Similar to Sec.~\ref{subsection:numerical_2}, we verify the bound by Monte Carlo simulation. In the simulation, the LDPC precode is randomly chosen from more than $10000$ pre-generated codes in the $(K,\mathtt{l},\mathtt{r})_q$ LDPC ensemble. Here we only consider the standard LDPC ensemble, because the weight enumerator for expurgated LDPC ensembles is unknown.

\begin{example}\label{example_LDPC_BATS_ML_1}	
	In this example, we consider two random LDPC-BATS codes based on $\mathbb{F}_4$ and with batch size $M = 8$. In this case, the empirical rank distribution is ${\bf h}_2$, which is provided in Table~\ref{table_h2}. Both codes have $K' = 64$, but their procodes are $(3,6)$-regular and $(3,15)$-regular, respectively. The simulation results and the upper bounds are shown in Fig.~\ref{fig_inactivation_sim_1}. We see that for both codes with different precodes, the upper bounds are tight when error probability is below $10^{-2}$. 
\end{example}

\begin{example}\label{example_LDPC_BATS_ML_2}
	This example considers two \textit{binary} random LDPC-BATS codes with $K'=128$ and batch size $M = 32$. The precodes of the two codes are different: one is $(3,6)$-regular and the other is $(3,15)$-regular. The empirical rank distribution ${\bf h}_3$ is provided in Table~\ref{table_h3}. In Fig.~\ref{fig_inactivation_sim_2}, it is evident that the upper bound by Theorem \ref{theorem_LDPC_BATS_INA} is very tight at error probabilities below $10^{-2}$. In addition, the simulation results in this figure illustrate that the binary random LDPC-BATS codes still perform well under ML decoding, when the batch size $M$ is relatively big. For the rank distribution ${\bf h}_3$ used in this figure, the capacity is $\bar{h} = 24.04$ and so the least number of batches required for error-free transmission is $K'/\bar{h} = 5.32$.
\end{example}

\begin{figure}[!t]
	\centering
	\includegraphics[width=3.5in]{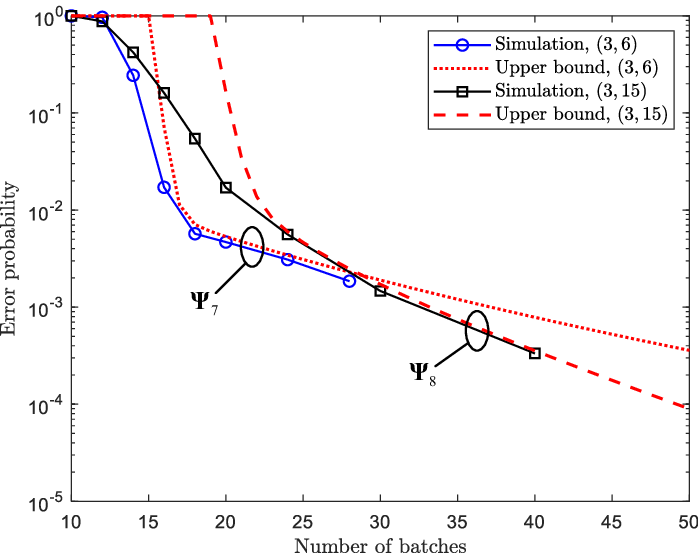}
	\caption{The simulation results and the upper bounds (Theorem \ref{theorem_LDPC_BATS_INA}) for ${\mathscr{C}}_{\rm pre}^{(4)}(128,n,8,{\bm \Psi}_7,{\bf h}_2,3,6,1)$ and ${\mathscr{C}}_{\rm pre}^{(4)}(80,n,8,{\bm \Psi}_8,{\bf h}_2,3,15,1)$ under ML decoding. Both codes have the number of input symbols $K'=64$.}
	\label{fig_inactivation_sim_1}
\end{figure}

\begin{figure}[!t]
	\centering
	\includegraphics[width=3.5in]{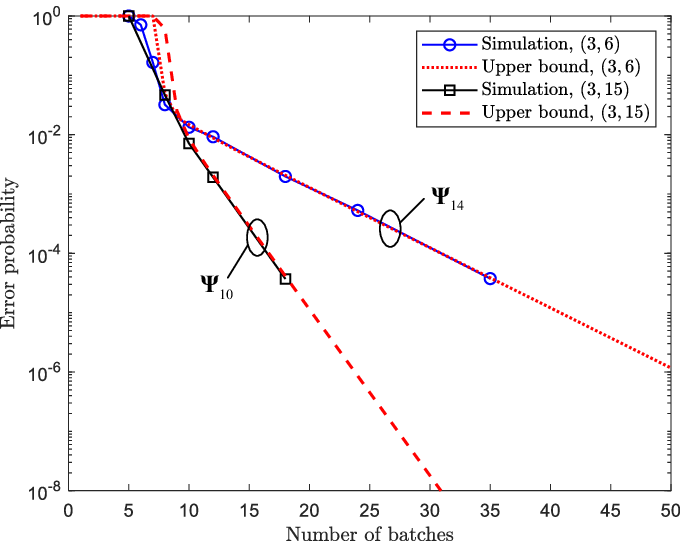}
	\caption{The simulation results and the upper bounds (Theorem \ref{theorem_LDPC_BATS_INA}) for ${\mathscr{C}}_{\rm pre}^{(2)}(256,n,32,{\bm \Psi}_{14},{\bf h}_3,3,6,1)$ and ${\mathscr{C}}_{\rm pre}^{(2)}(160,n,32,{\bm \Psi}_{10},{\bf h}_3,3,15,1)$ under ML decoding. Both codes have the number of input symbols $K'=128$.}
	\label{fig_inactivation_sim_2}
\end{figure}

\section{Optimization of the Degree Distribution}\label{Sec_opt_DD}
As a practical application of Theorem \ref{theorem_LDPC_BATS_BP} and Theorem \ref{theorem_LDPC_BATS_INA}, we introduce how to optimize the degree distribution of LDPC-BATS codes with the proposed performance bounds.

\subsection{Optimization Problem}
Our goal is to optimize the degree distribution of LDPC-BATS codes for BP decoding and ML decoding, respectively. For BP decoding, we can naturally use the BP performance upper bound as the objective function, ensuring good BP performance by minimizing this upper bound. However, for ML decoding, it is improper to directly minimize the ML performance upper bound, as this results in a trivial degree distribution $\Psi_K = 1$ in all the cases we evaluated (see Remark~\ref{remark:minimum_ML_bound} for some discussion). A BATS code with $\Psi_K = 1$ is a trivial dense code with the outer code being a random linear code. If each packet consists of $T$ finite-field symbols, such a dense code requires $O(K^2+TK)$ finite-field operations per packet to decode with Gaussian elimination, thus losing the low complexity property of BATS codes. When optimizing ML performance, low decoding complexity should be preserved. As an efficient ML decoding algorithm, inactivation decoding significantly reduces the dimension of Gaussian elimination by incorporating BP decoding. Empirically speaking, good BP performance is a prerequisite for inactivation decoding to achieve low complexity. Based on this experience, we still need to consider the BP performance when optimizing the performance of inactivation decoding.

In the optimization, the code parameters $q$, $K$, $M$, $\bf h$, $\mathtt{l}$, and $\mathtt{r}$ are fixed. Let $f_{\rm BP}(n,{\bm \Psi},\nu_{\rm min})$ and $f_{\rm ML}(n,{\bm \Psi},\nu_{\rm min})$ denote the upper bounds in Theorem \ref{theorem_LDPC_BATS_BP} and Theorem \ref{theorem_LDPC_BATS_INA}, respectively, where the fixed parameters are omitted. Based on the earlier mention that inactivation decoding possesses the property of joint BP and ML decoding, we can formulate the following optimization problem to address both BP decoding and inactivation decoding.
	\begin{equation}\label{eq:optimization}\tag{P1}
		\begin{aligned}
			\min~~& f_{\rm BP}(n_1,{\bm \Psi},\nu_{\rm min}) + \mathds{1}_{\{f_{\rm ML}(n_2,{\bm \Psi},1) > \epsilon^* \}}\\
			{\rm s.t.}~~ & \sum_{d=1}^{K} \Psi_d = 1\\
			& \Psi_d \ge 0, \forall d
		\end{aligned}
	\end{equation}
Here, $\mathds{1}_{\{\cdot\}}$ is the indicator function, and the parameters $n_1,n_2,\nu_{\rm min} \in \mathbb{Z}$ and $0 < \epsilon^* \le 1$ need to be properly set. In the second term of the objective function, $\nu_{\rm min}$ is fixed to $1$. This is because the weight enumerator for the expurgated LDPC ensemble is unknown to date. When optimizing the degree distribution for BP decoding, $\epsilon^* = 1$ is chosen. In this case, the second term in the objective function equals $0$ and can be ignored, so the optimization problem finds the degree distribution possibly with the optimal BP performance. When optimizing the degree distribution for inactivation decoding, the BP performance bound is still considered into the objective function to ensure low decoding complexity, and additionally, $\epsilon^* <1$ needs to be appropriately chosen to guarantee that the ML performance does not fall below a certain threshold. There may be other objective functions for inactivation decoding or other ML decoding algorithms that differ from the one we provide, but the core principle remains the same, that is, they must include considerations of ML decoding complexity.

Next, we provide some heuristics for selecting $n_1$, $n_2$, $\nu_{\rm min}$, and $\epsilon^*$.

\textit{Selecting $n_1$ and $n_2$:} Both $n_1$ and $n_2$ should be larger than $n_{\bar{h}} \triangleq {K(1-\mathtt{l}/\mathtt{r})}/{\bar{h}}$, where $K(1-\mathtt{l}/\mathtt{r})$ is the number of input symbols, and $\bar{h}$ is the capacity of the linear operation channel specified by $\bf h$, see (\ref{eq:h_capacity}). Since finite-length BP performance generally has a gap to $n_{\bar{h}}$, $n_1$ can not be set very close to $n_{\bar{h}}$. Empirically speaking, $n_1 \in [1.5n_{\bar{h}}:3n_{\bar{h}}]$ is a good choice for short LDPC-BATS codes. If the LDPC-BATS code is expected to have a low error floor, a larger $n_1$ is better for the optimization. $n_2$ needs to be set close to $n_{\bar{h}}$ because we want the ML decoding to be capacity approaching. Usually, we set $n_2 \approx 1.5n_{\bar{h}}$.

\textit{Selecting $\nu_{\rm min}$ and $\epsilon^*$:} $\nu_{\rm min} = \nu_{\rm min}^*$ is a good setting for $f_{\rm BP}(\cdot,\cdot,\cdot)$, see (\ref{eq:max_nu_min}) for $\nu_{\rm min}^*$. $\epsilon^*$ is useful only for low-order base field, e.g., $q = 2,4,8,16$, because the upper bound for ML decoding becomes looser as $q$ increases. For $q > 16$, we set $\epsilon^* = 1$ to ignore the second term in the objective function. For $q \le 16$, $\epsilon^*$ should not be set to a very small value such as $10^{-8}$. This is because $\nu_{\rm min} = 1$ is used for $f_{\rm ML}(\cdot,\cdot,\cdot)$ and so the ML decoding may have an error floor, see Fig.~\ref{fig_inactivation_sim_1} and Fig.~\ref{fig_inactivation_sim_2}. We observe that $0.001 \le \epsilon^* \le 0.1$ is suitable.

To solve the optimization problem (\ref{eq:optimization}), we follow the iterative optimization approach presented in \cite[Sec. V]{FL_analysis_BATS}. At the beginning of the iterative optimization, we choose an initial degree distribution ${\bm \Psi}^{(0)}$, which can be chosen according to asymptotic analysis or heuristic methods. For example, we can solve the optimization problem \cite[(P1)]{BATS} to obtain an asymptotically optimal degree distribution ${\bm \Psi}^{(0)}$, where the recovery rate $\bar{\eta}$ used in \cite[(P1)]{BATS} can be set to $\min\{0.99,1-\mathtt{l}/\mathtt{r}\}$ empirically. In each iteration $\ell$, a potential new degree distribution is updated by
	\begin{equation}\nonumber
		{\bm \Psi}^{(\ell + 1)} = \frac{{\bm \Psi}^{(\ell)} + \delta {\bf e}_{d}}{1 + \delta},
	\end{equation}
	where $d \in [1:K]$ is the selection of a degree, $\delta \in [-{\bm \Psi}^{(\ell)}[d],\delta_{\rm max}]$ is the adjustment for the selected degree, and ${\bf e}_{d}$ is the all-zero vector except that the $d$-th component is one. Here, $\delta_{\rm max}$ is an empirical value, e.g., $1$. In each iteration, we try to find the optimal $\delta$ and $d$ such that ${\bm \Psi}^{(\ell + 1)}$ reduces the objective function by the largest amount. We use the brute-force search as follows.

Assuming that the adjustment step of $\delta$ is $s$, the following candidate $\delta$ will be generated for each degree $d$ in iteration $\ell$:
	\begin{multline}\nonumber
		{\rm For~each~}{\bf e}_{d}:\\
		 -{\bm \Psi}^{(\ell)}[d],-{\bm \Psi}^{(\ell)}[d]+s,\ldots, -{\bm \Psi}^{(\ell)}[d] + \left\lfloor\frac{\delta_{\rm max}+{\bm \Psi}^{(\ell)}[d]}{s}\right\rfloor s.
	\end{multline}
	Based on $d \in [1:K]$ and the corresponding candidate $\delta$, we have a lot of candidate degree distributions ${\bm \Psi}^{(\ell + 1)}$. Then, compute the objective function for each candidate degree distribution and find the optimal ${\bm \Psi}^{(\ell + 1)}$. When a new degree distribution that reduces the objective function cannot be found or the maximum number of iterations is reached, the optimization is terminated. Obviously, the above iterative optimization can only result in a sub-optimal solution.

\begin{remark}
	In order to efficiently compute $f_{\rm BP}(n_1,{\bm \Psi},\nu_{\rm min})$ in the iterative optimization, a total of $K \times \left(\left\lfloor \frac{(K-1)\mathtt{l}}{\mathtt{r}} \right\rfloor+1\right)$ values of $L(a,j;K,\mathtt{l},\mathtt{r})$ ($a\in [1:K]$, $j \in [0:\lfloor (K-1)\mathtt{l}/\mathtt{r}\rfloor]$) can be pre-generated and stored.
\end{remark}

\subsection{Numerical Results}
We provide two examples of the degree-distribution optimization. To be specific, we consider the precode rate $R' = 1$ (i.e., without precode) and $R' = 0.5$. For $R' = 0.5$, we use the $q$-ary $(3,6)$-regular LDPC precode, which is constructed by the PEG algorithm~\cite{PEG} and the non-zero elements on the edges are randomly chosen. Therefore, in the following examples, the LDPC-BATS codes have a \textit{fixed precode} rather than a random precode, which is more practical. All examples consider the line network of length 2 (i.e., 2 links and 3 nodes). Unless otherwise stated, the links have the same erasure probability. The degree distributions used in this subsection can be found in Appendix~\ref{tables}.

\begin{example}
	Consider $q = 256$, $K' = 128$, $K = 128$ (standard BATS) or $256$ (LDPC-BATS), and $M = 16$. In this case, the empirical rank distribution is ${\bf h}_1$, see Table~\ref{table_h1}. In the optimization, we set $\epsilon^* = 1$, i.e., only consider the BP decoding performance. This is because, as mentioned before, the upper bound for ML decoding is very loose for $q = 256$.
	
	First, let us see the BP performance improvement produced by the proposed optimization method. The simulation results are shown in Fig.~\ref{fig:DD_opt_q256_1}. In this figure, the three curves without a precode are computed by the exact recursive formula~\cite[Theorem 1]{FL_analysis_BATS}, and the remaining curves are obtained via Monte Carlo simulation. It is worthy noting that ``finite ($n_1 = n$)'' in the legend represents that the code is optimized by the \textit{approximate recursive formula} \cite[Theorem 16]{FL_analysis_BATS} at $n$.
	
	For the standard BATS codes, we optimize the degree distribution at the operating point $n_1 = 30$. In Sec.~\ref{subsection:numerical_1}, we see that the approximate recursive formula overestimates the error probability noticeably, and the upper bound is more accurate when $n$ is relatively large. Thus, in Fig.~\ref{fig:DD_opt_q256_1}, we observe that the degree distribution optimized by (\ref{eq:optimization}) leads to a smaller error probability when $n \ge n_1$.
	
	For LDPC-BATS codes with $R'=0.5$, both the asymptotic optimization method~\cite{BATS} and the finite-length optimization method~\cite{FL_analysis_BATS} show that the LDPC-BATS codes (with ${\bm \Psi}_{6}$ and ${\bm \Psi}_{2}$) perform \textit{worse} than the standard BATS codes (with ${\bm \Psi}_{1}$ and ${\bm \Psi}_{12}$). Does this mean that using a low-rate precode is detrimental to BP decoding? The answer is no. The reason is that the existing degree optimization methods do not provide a proper degree distribution for the LDPC-BATS code with $R' = 0.5$. Using (\ref{eq:optimization}), the LDPC-BATS code with ${\bm \Psi}_5$ has the best performance among all the codes in Fig.~\ref{fig:DD_opt_q256_1}. The performance improvement at low error probability region is significant. Define $n-K'/\bar{h}$ as the transmission overhead. To achieve an error probability below $0.001$, the standard BATS code and the optimized LDPC-BATS code need $n = 32$ and $n=21$, respectively. Thus, for BP decoding, the optimized LDPC-BATS code lowers the transmission overhead to 48.2\% of that of the standard BATS code.
	
	Then, we discuss the impact of the parameter selections in the optimization. Since this example only involves BP decoding, we will focus on $n_1$ and $\nu_{\rm min}$. Recall that the $(3,6)$-regular LDPC code of length 256 is used as the precode. For the corresponding LDPC ensemble, we have $\nu_{\rm min}^* = 7$, see (\ref{eq:max_nu_min}). Changing the selections of $n_1$ and $\nu_{\rm min}$ with $n_1 \in \{21,26\}$ and $\nu_{\rm min} \in \{1,7\}$, we construct four LDPC-BATS codes. The simulation results of these codes are shown in Fig.~\ref{fig:DD_opt_q256_2}. All the codes perform well when the error probability is above $0.05$. Below this error probability, the two codes with $n_1 = 26$ have a lower error floor, and the one with $\nu_{\rm min} = 7$ performs better. The reason behind this phenomenon is quite simple: When the LDPC-BATS code uses a PEG-LDPC precode, the error probability decreases \textit{rapidly} with $n$. When the error probability is not low enough, the upper bound overestimates the error probability, see Fig.~\ref{fig_LDPC_BATS_K256_performance}. Selecting a larger $n_1$ can avoid this issue. Choosing a larger $\nu_{\rm min}$ ensures that the upper bound for the ensemble is closer to the performance of the LDPC-BATS code with a PEG-LDPC precode.
	
	Finally, we briefly discuss the impact of the erasure probability of the links on BP decoding performance. The LDPC-BATS code with ${\bm \Psi}_5$ is considered, which is optimized for erasure probability of $0.2$. By varying the erasure probability of the links, we present the simulation results in Fig.~\ref{fig:DD_opt_q256_3}. We see that the LDPC-BATS code exhibits a certain level of robustness to small-scale variations in the erasure probability, but it is not a universal code. At the erasure probability of $0.2$, this code performs best (measured by the gap between error probability of $0.005$ and $K'/\bar{h}$). How to design a quasi-universal LDPC-BATS code requires further research. A potentially feasible way is to follow the optimization framework with multiple $\bf h$ in \cite{Xu2017QUBATS}, but the formulation of the optimization problem in \cite{Xu2017QUBATS} needs some refinements based on the theorems in this paper.
\end{example}

\begin{figure}[!t]
	\centering
	\includegraphics[width=3.5in]{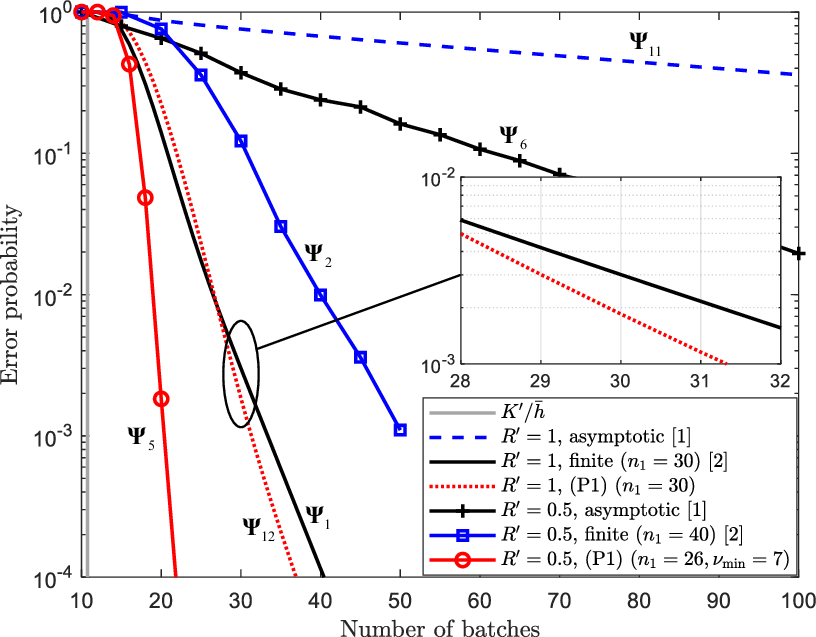}
	\caption{The BP decoding performance of several standard BATS and LDPC-BATS codes with $q=256$, $K' = 128$, and $M=16$.}
	\label{fig:DD_opt_q256_1}
\end{figure}

\begin{figure}[!t]
	\centering
	\includegraphics[width=3.5in]{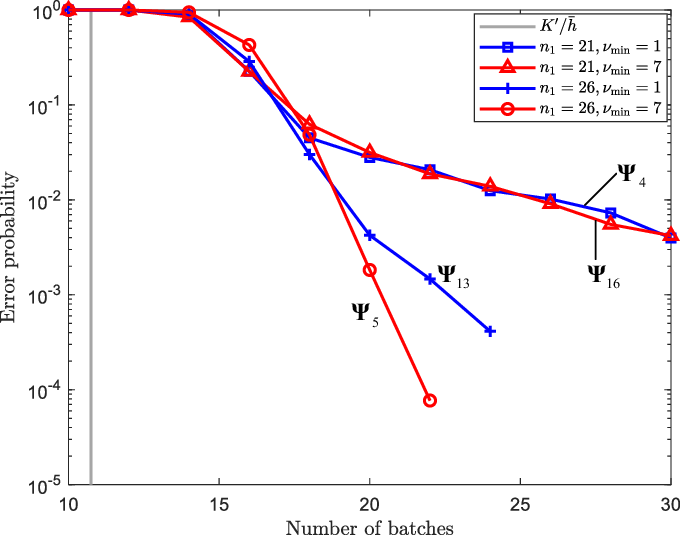}
	\caption{The BP decoding performance of four LDPC-BATS codes with $q=256$, $K' = 128$, and $M=16$, which are optimized at different $n_1$ and $\nu_{\rm min}$.}
	\label{fig:DD_opt_q256_2}
\end{figure}

\begin{figure}[!t]
	\centering
	\includegraphics[width=3.5in]{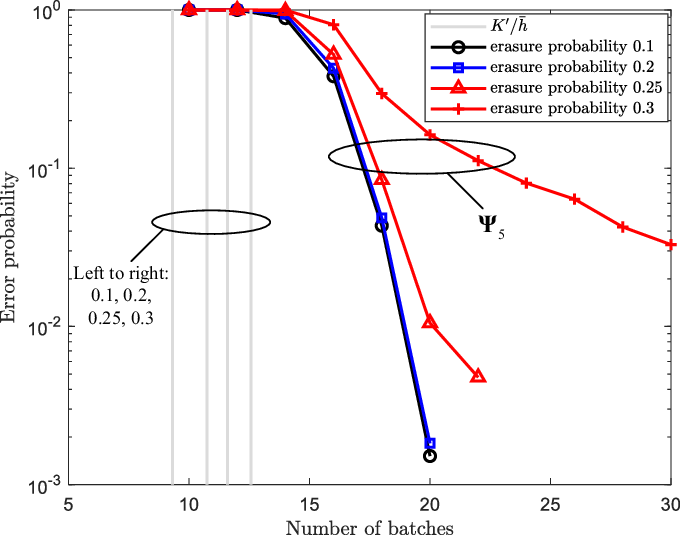}
	\caption{The BP decoding performance of the LDPC-BATS code with $q=256$, $K' = 128$, and $M=16$ through the line networks with different erasure probabilities. (This code is optimized for erasure probability $0.2$.)}
	\label{fig:DD_opt_q256_3}
\end{figure}

\begin{example}
	Consider $q = 2$, $K' = 128$, $K = 128$ (standard BATS) or $256$ (LDPC-BATS), and $M = 32$. In this case, the empirical rank distribution is ${\bf h}_3$, see Table~\ref{table_h3}. For $q = 2$, the upper bound for ML decoding is relatively tight, so we can optimize the degree distribution for both BP decoding and ML decoding. Here, we fix $n_1 = 15$, $\nu_{\rm min} = 7$, and $n_2 = 8$ in the optimization. Two settings of $\epsilon^*$, $0.05$ and $1$, will be compared.
	
	The simulation results of decoding performance are shown in Fig.~\ref{fig:DD_opt_q2_1}. For BP decoding, both the ML-optimized (i.e., $\epsilon^* = 0.05$) and the BP-optimized (i.e., $\epsilon^* = 1$) LDPC-BATS codes significantly outperform the asymptotically optimal LDPC-BATS code obtained by~\cite{BATS}. For achieving an error probability below $0.001$, the BP-optimized LDPC-BATS code reduces the transmission overhead to 37.8\% of that of the standard BATS code optimized by~\cite{FL_analysis_BATS}.
	
	Now consider the performance of ML decoding. To be specific, inactivation decoding \cite{InactivationPatent,Lazaro2017inactivation} is used. By accounting for the performance of ML decoding, the ML-optimized LDPC-BATS code achieves a gap of $1.5$ batches than the BP-optimized one. It is not surprising that the asymptotically optimal LDPC-BATS code has the best ML decoding performance among three LDPC-BATS codes, since its \textit{average degree} (about $31.1$) is larger than the other two LDPC-BATS codes (about $28.2$ and $23.7$). However, later we will see that the asymptotically optimal LDPC-BATS code has much higher ML decoding complexity. For achieving an error probability below $0.001$, the ML-optimized LDPC-BATS code reduces the transmission overhead to 22.0\% of that of the standard BATS code optimized by~\cite{FL_analysis_BATS}.
	
	Next, we discuss the decoding complexity of inactivation decoding, which is measured in terms of the \textit{average number of inactive symbols}. In our simulation, the inactive symbols are chosen uniformly at random from the active symbols\footnote{This strategy is sometimes referred to as ``random inactivation'' in the literature~\cite{FL_analysis_BATS,Lazaro2017inactivation}. This is certainly not the only feasible inactivation strategy, nor is it the one that results in the fewest inactive symbols. For more sophisticated inactivation strategies, we refer readers to~\cite{InactivationPatent,6266770}.}. Empirically speaking, worse BP performance tends to result in a larger average number of inactive symbols. In Fig.~\ref{fig:DD_opt_q2_2}, as expected, the complexity of the asymptotically optimal LDPC-BATS code is significantly higher than the other three codes. Due to the finite-length optimization for BP decoding, the average number of inactive symbols for the two LDPC-BATS codes optimized by (\ref{eq:optimization}) is less than $10$ when $n \ge 8$ and can even be fewer than that for the standard BATS code. This suggests a somewhat counterintuitive conclusion: As long as the degree distribution is appropriate, using a low-rate LDPC precode does not necessarily increase the complexity of inactivation decoding.\footnote{In contrast, when using a high-density parity-check precode, the complexity of inactivation decoding increases significantly with the number of parity-check symbols of the precode, regardless of how the degree distribution of the BATS code is optimized.} Interestingly, we observe that the intersection points marked in Fig.~\ref{fig:DD_opt_q2_1} and Fig.~\ref{fig:DD_opt_q2_2} are quite close, which somewhat supports the empirical view that improving BP performance helps reduce the complexity of inactivation decoding.
\end{example}

\begin{figure}[!t]
	\centering
	\includegraphics[width=3.5in]{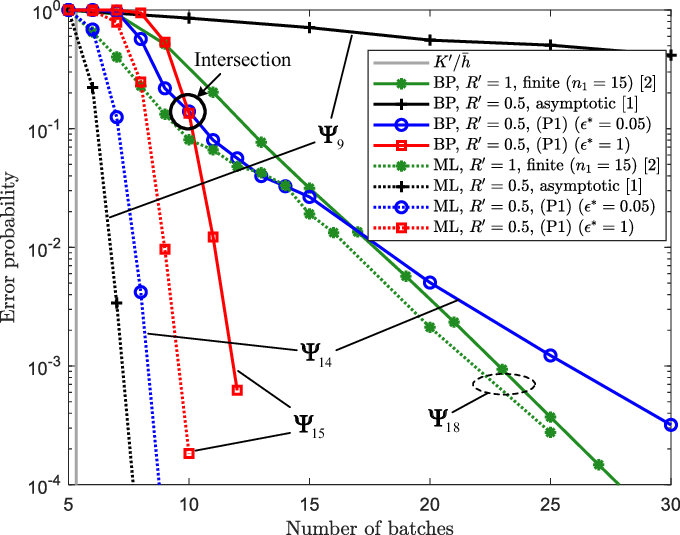}
	\caption{The BP decoding performance and the ML (i.e., inactivation) decoding performance of several standard BATS and LDPC-BATS codes with $q=2$, $K' = 128$, and $M=32$.}
	\label{fig:DD_opt_q2_1}
\end{figure}

\begin{figure}[!t]
	\centering
	\includegraphics[width=3.5in]{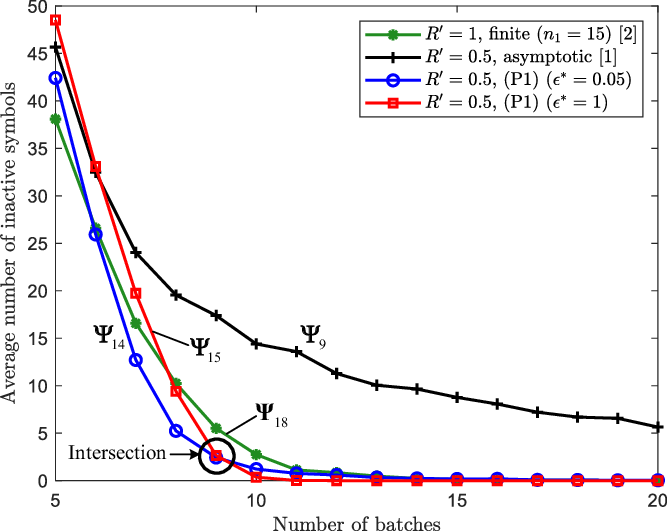}
	\caption{The inactivation decoding complexity of several standard BATS and LDPC-BATS codes with $q=2$, $K' = 128$, and $M=32$. (Inactive symbols are chosen uniformly at random.)}
	\label{fig:DD_opt_q2_2}
\end{figure}

\section{Conclusion}\label{sec:Conclusion}
In this paper, we derive two performance upper bounds for precoded BATS codes under BP decoding and ML decoding, respectively, and utilize the bounds to optimize the degree distribution of the codes. When a non-trivial precode is used, i.e., the precode rate is less than $1$, our work presents the first analytical results on both BP decoding and ML decoding of precoded BATS codes in the finite-length regime.

For BP decoding, the upper bound is applicable to LDPC-BATS codes, where the precode is an $(\mathtt{l},\mathtt{r})$-regular LDPC code chosen from the regular LDPC ensemble uniformly at random. We provide a necessary and sufficient condition such that the BP decoder for LDPC-BATS codes fails. This condition is based on a stopping-set-like structure in the Tanner graph, which is proposed in this paper. The accuracy of this upper bound is independent of $q$, i.e., the order of the finite field from which the code is constructed. The numerical results indicate that this upper bound is generally tight for relatively low error probability, e.g., $10^{-2}$ in general and $10^{-4}$ in some special cases.

For ML decoding, the upper bound is suitable for precoded BATS codes with any generic code as long as the precode's weight enumerator is known. This upper bound is a function of the number $n$ of transmitted batches, and clearly reveals how the error probability of ML decoding decreases exponentially with $n$. Although this upper bound is tight only for precoded BATS codes with a relatively small number of input symbols and low-order base fields, it represents the first analytical result on the error probability of precoded BATS codes (even with the precode rate $1$) under ML decoding. Tighter bounds for ML decoding of BATS codes with high-order base fields (e.g., $\mathbb{F}_{256}$) may deserve further research.

Degree-distribution optimization is a practical application of the derived upper bounds. We formulate an optimization problem based on the two upper bounds, and show how to numerically solve it by iterative optimization. The optimized LDPC-BATS codes outperform the conventional BATS codes that are optimized by the existing methods under BP decoding, inactivation decoding, or both. Even with a relatively low-rate precode, the optimized LDPC-BATS codes still achieve low complexity in inactivation decoding (i.e., have a small number of inactive symbols).

\appendices
\section{Derivation of the Bound on the Error Probability of BATS Codes Under BP Decoding}\label{Sec_Proof_Theorem_BP}
This appendix proves Theorem~\ref{theorem_LDPC_BATS_BP}. For clarity, we divide the proof into three subsections. As a starting point, we first discuss how a BP decoding failure of standard BATS codes occurs and derive an upper bound for standard BATS codes in Appendix~\ref{subsection:proof_BP_std_BATS} (Corollary~\ref{corollary_std_BATS_BP} can be proved independently by this subsection). Then, we derive two upper bounds on the probability that a subgraph of the LDPC precode is an extended stopping set (Definition \ref{def:ess}) in Appendix~\ref{subsection:proof_extended_stopping_set}. Finally, in Appendix~\ref{subsection:proof:theorem_1}, we prove Theorem~\ref{theorem_LDPC_BATS_BP} based on the bounding techniques for the aforementioned two parts.

The following notations for Tanner graphs will be used: For a node $o$ in the Tanner graph, given an edge set ${\cal E}' \subset {\cal E}$, where $\cal E$ is the set consisting of all edges in the Tanner graph, let ${\rm Neib}_{{\cal E}'}(o)$ be the set of neighbors of $o$ which can be reached through the edges in ${\cal E}'$. If the edge set is not specified for ${\rm Neib}(\cdot)$, let ${\rm Neib}(o)$ be the set of all neighbors of $o$ in the Tanner graph. The degree of a node $o$ is defined as ${\rm dg}(o) = |{\rm Neib}(o)|$.

\subsection{Upper Bound for Standard BATS Codes}\label{subsection:proof_BP_std_BATS}
Let ${\cal T} = ({\cal V},{\cal C},{\cal E})$ be a Tanner graph in the $(K,n,M,{\bm \Psi},{\bf h})_q$ BATS ensemble, where ${\cal V} = \{V_1,\ldots,V_K\}$ is the VN set, ${\cal C} = \{C_1,\ldots,C_n\}$ is the B-CN set, and ${\cal E} = \{(V_{i_1},C_{j_1}),(V_{i_2},C_{j_2}),\ldots\}$ is the edge set. Recall that $\cal T$ is a fixed graph and the batch equation ${\bf B}_i {\bf G}_i {\bf H}_i = {\bf Y}_i$ for $C_i$ is also deterministic. We introduce the rank of a B-CN $C_i$, which is equivalent to the decoding capability of $C_i$.

\begin{definition}
	For a B-CN $C_i$ whose generator matrix and transfer matrix are ${\bf G}_i$ and ${\bf H}_i$, respectively, fix a subset ${\cal S}$ of ${\rm Neib}(C_i)$. The rank of $C_i$ with ${\cal S}$ is defined as 
	\begin{equation}\label{eq_rank_subGH}
		{\rm rk}(C_i^{(\cal S)}) \triangleq {\rm rk}({\bf G}_i^{(\cal S)} {\bf H}_i),
	\end{equation}
	where ${\bf G}_i^{(\cal S)}$ is the submatrix of ${\bf G}_i$ formed by the $|{\cal S}|$ rows of ${\bf G}_i$ corresponding to the VNs in ${\cal S}$. If ${\cal S} = {\rm Neib}(C_i)$, define the shorthand notation 
	\begin{equation}
		{\rm rk}(C_i) \triangleq {\rm rk}(C_i^{{\rm Neib}(C_i)}).
	\end{equation}
\end{definition}

In BP decoding, the VNs in ${\cal S} \subset {\rm Neib}(C_i)$ can be decoded by $C_i$ \textit{only if} ${\rm rk}(C_i^{(\cal S)}) = |{\cal S}|$. $\hbar_k'$ defined in (\ref{eq_hbar_prime}) is in fact the probability that ${\rm rk}(C_i^{(\cal S)}) = k$ given $|{\cal S}| = k$. Using the concepts introduced before, we now define a class of \textit{stopping graphs} ${\cal T}$, which includes at least one non-empty subgraphs that cannot be removed by the peeling decoder.
\begin{definition}[BATS Stopping Graph]\label{Def_stopping_graph_BATS}
	A Tanner graph ${\cal T} = ({\cal V},{\cal C},{\cal E})$ of the BATS code is a stopping graph if there exist a non-empty set ${\cal A} \subset {\cal V}$ and a set (can be empty) ${\cal B} \subset {\cal C}$ such that the following constraints hold:
	\begin{itemize}
		\item[${\rm C1}$:] For any $V \in {\cal A}$, ${\rm Neib}_{\cal E}(V) \cap ({\cal C} \backslash {\cal B}) = \emptyset$.
		\item[${\rm C2}$:] For any $C \in {\cal B}$, $|{\rm Neib}(C) \cap {\cal A}| > {\rm rk}(C^{({\rm Neib}(C) \cap {\cal A})}) \ge 0$.
		\item[${\rm C3}$:] For any $V \in ({{\cal V} \backslash {\cal A}})$, ${\rm Neib}_{\cal E}(V) \cap ({\cal C} \backslash {\cal B}) \ne \emptyset$.
	\end{itemize}
\end{definition}

For a parameter pair $({\cal A},{\cal B})$ in ${\cal T}$, we say that $({\cal A},{\cal B})$ is \textit{valid} for ${\cal T}$ if $({\cal A},{\cal B})$ satisfies the three constraints ${\rm C1}$, ${\rm C2}$, and ${\rm C3}$.

In fact, the subgraph constructed by the nodes in $({\cal A},{\cal B})$ can be regarded as a generalization of \textit{stopping sets} of LDPC codes \cite{FL_analysis_LDPC}. The following lemma demonstrates that valid parameter pairs possess the property of composability, similar to the property exhibited by stopping sets.
\begin{lemma}\label{Lemma_Combine_BATS_SS}
	For a Tanner graph ${\cal T} = ({\cal V},{\cal C},{\cal E})$, if two parameter pairs $({\cal A}_1, {\cal B}_1)$ and $({\cal A}_2, {\cal B}_2)$ are valid for $\cal T$, then $\left({{\cal A}_1 \cup {\cal A}_2} \cup {\cal A}_e,{{\cal B}_1 \cup {\cal B}_2}\right)$ is valid for $\cal T$, where
	\begin{multline}\nonumber
	{\cal A}_e \triangleq \biggg\{ V \in \bigcup_{C \in ({{\cal B}_1 \cup {\cal B}_2})} \left({\rm Neib}(C) \setminus ({\cal A}_1 \cup {\cal A}_2)\right):\\ {\rm Neib}_{\cal E}(V) \cap ({\cal C} \backslash ({\cal B}_1 \cup {\cal B}_2)) = \emptyset \biggg\}.
	\end{multline}
\end{lemma}
\begin{proof}
	See Appendix~\ref{appendix:proofs_for_BP}.
\end{proof}

\begin{remark}
	In general, there are more than one parameter pairs which are valid for a graph $\cal T$. However, for a fixed B-CN set ${\cal B}$, the valid $({\cal A},{\cal B})$ is \textit{unique} (if it exists) according to constraint ${\rm C3}$. Example~\ref{Example_BATS_ss} shows that different valid parameter pairs exist for a Tanner graph.
\end{remark}

\begin{figure*}[!t]
	\centering
	\includegraphics[width=6.5in]{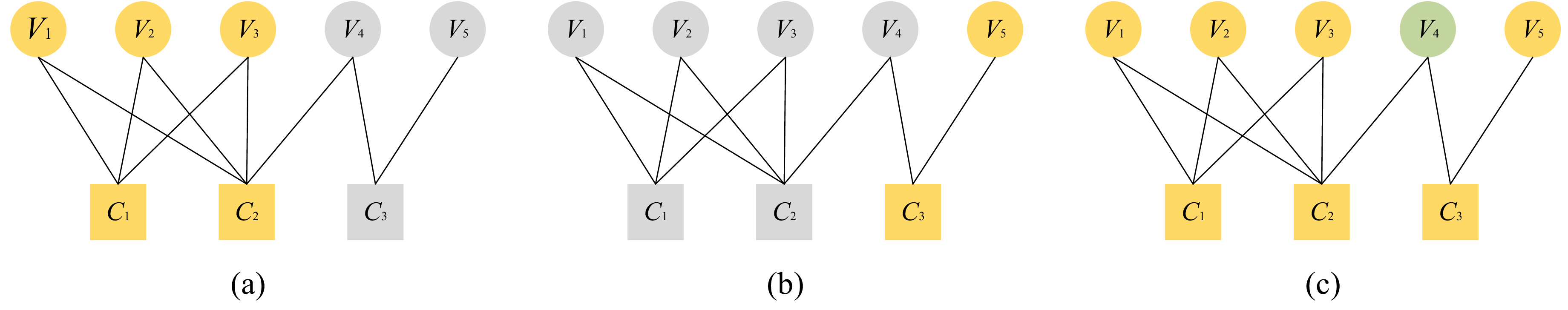}
	\caption{An example of BATS stopping graphs. The nodes in ${\cal A}$ or ${\cal B}$ are marked in yellow, and the node in ${\cal A}_e$ is marked in green.}
	\label{fig_BATS_ss}
\end{figure*}

\begin{example}\label{Example_BATS_ss}
	Fix a base field $\mathbb{F}_2$ and batch size $M = 4$. Consider a graph $\cal T$ shown in Fig.~\ref{fig_BATS_ss}. Assume that this graph has generator matrices ${\bf G}_i$ and transfer matrices ${\bf H}_i$ as follows:
	\begin{equation}\nonumber
		\begin{aligned}
		&{\bf G}_1 = \begin{bmatrix}
			1 & 1 & 1 & 0\\
			1 & 1 & 1 & 1\\
			1 & 1 & 1 & 0
		\end{bmatrix},
		{\bf H}_1 = \begin{bmatrix}
			1 &  1 &  1 &  1\\
			1 &  0 &  1 &  0\\
			0 &  0 &  0 &  1\\
			1 &  1 &  1 &  0
		\end{bmatrix}\\
		&\Rightarrow {\bf G}_1 {\bf H}_1 = \begin{bmatrix}
		    0  & 1 &  0  & 0\\
			1  & 0 &  1  & 0\\
			0  & 1 &  0  & 0
		\end{bmatrix}.
		\end{aligned}
	\end{equation}
	\begin{equation}\nonumber
		\begin{aligned}
	&{\bf G}_2 = \begin{bmatrix}
		0 &  1 &  1  & 1\\
		1 &  1 &  1  & 1\\
		1 &  1 &  0  & 0\\
		1 &  0 &  1  & 0
	\end{bmatrix},
	{\bf H}_2 = \begin{bmatrix}
		0 &  1 &  0 &  0\\
		0 &  0 &  0 &  0\\
		0 &  1 &  1 &  0\\
		1 &  0 &  1 &  1
	\end{bmatrix}\\
	&\Rightarrow{\bf G}_2 {\bf H}_2 = \begin{bmatrix}
   		1 &  1 &  0 &  1\\
		1 &  0 &  0 &  1\\
		0 &  1 &  0 &  0\\
		0 &  0 &  1 &  0
	\end{bmatrix}.
\end{aligned}
\end{equation}
\begin{equation}\nonumber
	\begin{aligned}
	&{\bf G}_3 = \begin{bmatrix}
   		1 &  0 &  1 &  0\\
		0 &  0 &  0 &  0
	\end{bmatrix},
	{\bf H}_3 = \begin{bmatrix}
	   	1 &  1 &  1 &  1\\
		1 &  0 &  0 &  0\\
		0 &  0 &  1 &  1\\
		1 &  0 &  0 &  1
	\end{bmatrix}\\
	&\Rightarrow{\bf G}_3 {\bf H}_3 = \begin{bmatrix}
		1 &  1 &  0  & 0\\
		0 &  0 &  0  & 0
	\end{bmatrix}.
\end{aligned}
\end{equation}

In Fig.~\ref{fig_BATS_ss}(a), we have ${\cal A} = \{V_1,V_2,V_3\}$ and ${\cal B} = \{C_1,C_2\}$. In this case, ${\rm rk}(C_1^{{\rm Neib}(C_1) \cap {\cal A}}) = {\rm rk}({\bf G}_1 {\bf H}_1) = 2$ and ${\rm rk}(C_2^{{\rm Neib}(C_2) \cap {\cal A}}) = {\rm rk}(({\bf G}_2 {\bf H}_2)[1:3,:]) = 2$.
In Fig.~\ref{fig_BATS_ss}(b), we have ${\cal A} = \{V_5\}$ and ${\cal B} = \{C_3\}$. In this case, ${\rm rk}(C_3^{{\rm Neib}(C_3) \cap {\cal A}}) = {\rm rk}(\left({\bf G}_3 {\bf H}_3\right)[2,:]) = 0$.
Fig.~\ref{fig_BATS_ss}(c) is an application of Lemma \ref{Lemma_Combine_BATS_SS}, where ${\cal A}_e = \{V_4\}$ is marked in green.
\end{example}

Next, we will utilize the graph-related concepts introduced earlier to analyze the performance of BP decoding. The following lemma provides a necessary and sufficient condition such that BP decoding fails for ${\cal T}$.

\begin{lemma}\label{Lemma_BATS_Stops_at_Stopping_Graph}
	For a Tanner graph ${\cal T} = ({\cal V},{\cal C},{\cal E})$ of a BATS code, ${\cal T}$ cannot be reduced to an empty graph after BP (peeling) decoding if and only if ${\cal T}$ is a stopping graph.
\end{lemma}
\begin{proof}
	See Appendix \ref{appendix:proofs_for_BP}.
\end{proof}

We form some subsets of the BATS ensemble, such that the union of them contains all stopping graphs in the ensemble. Letting ${\cal A}^{(a)}(i_1,i_2,\ldots,i_a) = \{V_{i_1},V_{i_2},\ldots,V_{i_a}\}$ and ${\cal B}^{(b)}(j_1,j_2,\ldots,j_b) = \{C_{j_1},C_{j_2},\ldots,C_{j_b}\}$, then $({\cal A}^{(a)}(i_1,i_2,\ldots,i_a),{\cal B}^{(b)}(j_1,j_2,\ldots,j_b))$ represents a parameter pair with \textit{fixed} $a$ VNs and \textit{fixed} $b$ B-CNs in $\cal T$. When the specified indices $i_1,i_2,\ldots,i_a$ for VNs and $j_1,j_2,\ldots,j_b$ for B-CNs do not play a role in the context, we write $({\cal A}^{(a)},{\cal B}^{(b)})$ for short. The set of all stopping graphs, for which $({\cal A}^{(a)},{\cal B}^{(b)})$ is valid for \textit{all graphs} in this set, is formed as
	\begin{equation}
		G_{\rm stop}({\cal A}^{(a)},{\cal B}^{(b)}) \triangleq \left\{{\cal T}: ({\cal A}^{(a)},{\cal B}^{(b)}) \text{ is valid for }{\cal T}\right\}.
	\end{equation}
There are totally $\sum_{a=1}^{K} \sum_{b=0}^{n} \binom{K}{a}\binom{n}{b}$ such sets and their union contains all stopping graphs in the BATS ensemble. Note that different sets $G_{\rm stop}({\cal A}^{(a)}(i_1,i_2,\ldots,i_a),{\cal B}^{(b)}(j_1,j_2,\ldots,j_b))$ and $G_{\rm stop}({\cal A}^{(a)}(i_1',i_2',\ldots,i_a'),{\cal B}^{(b)}(j_1',j_2',\ldots,j_b'))$ may not be disjoint. For example, the graph shown in Fig.~\ref{fig_BATS_ss} belongs to three different sets $G_{\rm stop}({\cal A}^{(a)},{\cal B}^{(b)})$. Let ${\cal T}^*$ denote the random variable representing a graph chosen from the $(K,n,M,{\bm \Psi},{\bf h})_q$ BATS ensemble uniformly at random. The ensemble-average error probability $\mathbb{E}\left[P_{\textsf{E}}^{\rm BP}({\cal T})\right]$, i.e., the probability that BP decoding for ${\cal T}^*$ fails, can be bounded by
\begin{equation}\label{eq_P_err_Union}
		\begin{aligned}
			&\mathbb{E}\left[P_{\textsf{E}}^{\rm BP}({\cal T})\right]  = \Pr\left\{{\cal T}^* \in \bigcup_{{\cal A} \ne \emptyset, {\cal A}\subset{\cal V},{\cal B}\subset{\cal C}} {G_{\rm stop}({\cal A},{\cal B})}\right\}\\			
			& \le \min\left\{1, \sum_{a=1}^{K} \sum_{b=0}^{n} \Pr\left\{{\cal T}^* \in \bigcup_{\substack{{\cal A}\subset{\cal V},{\cal B}\subset{\cal C},\\|{\cal A}|=a,|{\cal B}|=b}} {G_{\rm stop}({\cal A},{\cal B})}\right\}\right\}\\
			& \le \min\biggg\{1, \sum_{a=1}^{K} \sum_{b=0}^{n} \binom{K}{a}\binom{n}{b} \\
			&~~~~~~~~~~~~~~~~~~~~~~~~~\cdot \Pr\left\{{\cal T}^* \in {G_{\rm stop}({\cal A}^{(a)},{\cal B}^{(b)})}\right\}\biggg\}.
		\end{aligned}
\end{equation}

The following lemma gives the edge distribution of a B-CN, and will be used in the computation of the probability $\Pr\left\{{\cal T}^* \in {G_{\rm stop}({\cal A}^{(a)},{\cal B}^{(b)})}\right\}$ in the last line of (\ref{eq_P_err_Union}).
\begin{lemma}\label{Lemma_BCN_neighbor_partition}
	Let $\{{\cal V}_1,{\cal V}_2,\ldots,{\cal V}_p\}$ be a partition of $\cal V$. Let $C$ be a randomly selected B-CN in ${\cal T}^*$, where ${\cal T}^*$ is a random graph chosen from the $(K,n,M,{\bm \Psi},{\bf h})_q$ BATS ensemble uniformly at random. For $p$ non-negative integers $k_1,\ldots,k_p$ such that $\Psi_{\sum_{j=1}^{p} k_j} > 0$, we have 
	\begin{multline}\nonumber
		\Pr\left\{\left.\begin{aligned}&|{\rm Neib}(C) \cap {\cal V}_1| = k_1,\\ &\ldots,\\ &|{\rm Neib}(C) \cap {\cal V}_p| = k_p \end{aligned}\right| {\rm dg}(C) = \sum_{j=1}^{p} k_j\right\} \\= \frac{\binom{|{\cal V}_1|}{k_1}\binom{|{\cal V}_2|}{k_2}\cdots \binom{|{\cal V}_p|}{k_p}}{\binom{K}{k_1+k_2+\cdots+k_p}}.
	\end{multline}
\end{lemma}
\begin{proof}
	This lemma follows from the fact that each B-CN $C$ independently chooses ${\rm dg}(C)$ distinct neighbors uniformly at random.
\end{proof}

Now, we study the formula for $\Pr\{{\cal T}^* \in {G_{\rm stop}({\cal A}^{(a)},{\cal B}^{(b)})}\}$. Let $E_i(a,b)$ denote the event that $({\cal A}^{(a)},{\cal B}^{(b)})$ in ${\cal T}^*$ satisfies constraint ${\rm C}i$, ${i} \in \{1,2,3\}$. First of all, we can see that $E_2(a,b)$ is independent of both $E_1(a,b)$ and $E_3(a,b)$, because ${\rm C2}$ restrict the B-CNs in ${\cal B}^{(b)}$, while ${\rm C1}$ and ${\rm C3}$ restrict the other B-CNs. We can develop $\Pr\{{\cal T}^* \in {G_{\rm stop}({\cal A}^{(a)},{\cal B}^{(b)})}\}$ as
\begin{multline}\label{eq_probability_decomposition}
	\Pr\{{\cal T}^* \in {G_{\rm stop}({\cal A}^{(a)},{\cal B}^{(b)})}\} = \Pr\{E_1(a,b)\}\\ \cdot \Pr\{E_2(a,b)\}  \Pr\{E_3(a,b)|E_1(a,b)\}.
\end{multline}

Recall (\ref{eq_varphi}) and (\ref{eq_Phi}) that for $d \in [1:K]$
\begin{equation}\nonumber
	\varphi_{d,k} = \Psi_d \frac{\binom{k}{d}}{\binom{K}{d}},\\
\end{equation}
and for $d \in [0:K]$
\begin{equation}\nonumber
		{\Phi}_{d}^{(k)} = \begin{cases}
	\frac{\varphi_{d,k}}{\sum_{d' = 1}^{k} \varphi_{d',k}}&\text{if }d \ge 1 \text{ and } k \ge D_{\rm min},\\
	1&\text{if }d=0 \text{ and } k < D_{\rm min},\\
	0&\text{otherwise,}
	\end{cases} 
\end{equation}
where $D_{\rm min}$ is the smallest degree $d$ with $\Psi_d > 0$. According to Lemma \ref{Lemma_BCN_neighbor_partition}, $\varphi_{d,k}$ is the probability that a B-CN $C$ is of degree $d$ and all neighbors of $C$ are included in a set $\tilde{\cal V}_k$ of $k$ arbitrary but fixed VNs. ${\Phi}_{d}^{(k)}$, $d \in [1:K]$, is the conditional probability that $C$ is of degree $d$ given ${\rm Neib}(C) \subset \tilde{\cal V}_k$, i.e., $\Pr\{{\rm dg}(C) = d \mid {\rm Neib}(C) \subset \tilde{\cal V}_k\}$. Note that when $k < D_{\rm min}$, the condition event ${\rm Neib}(C) \subset \tilde{\cal V}$ cannot happen, and then the ${\Phi}_{d}^{(k)}$ becomes an undefined term. Let ${\Phi}_{d}^{(k)} = 0$, $d \in [1:K]$, for $k < D_{\rm min}$. Further, a degree-zero term ${\Phi}_{0}^{(k)}$ is introduced to ensure ${\bm \Phi}^{(k)}$ defined by (\ref{eq_phi_vector}) is a degree-distribution vector\footnote{Adding a degree-zero term simplifies the expression in the subsequent proof and does not have any impact on the analysis. When the conditional probability $\Pr\{{\rm dg}(C) = d \mid {\rm Neib}(C) \subset \tilde{\cal V}_k\}$ is well defined, ${\Phi}_{0}^{(k)} = 0$.}, i.e., the sum of all components equals one.

Event $E_1(a,b)$ is equivalent to the event that all B-CNs $\notin {\cal B}^{(b)}$ do not have neighbors in ${\cal A}^{(a)}$. Due to the independence of all B-CNs, one can write that \[\Pr\{E_1(a,b)\} = \left(\sum_{d = 1}^{K} \varphi_{d,K-a}\right)^{n-b}.\]

We now study the probability $\Pr\{E_2(a,b)\}$. For a B-CN $C_i \in {\cal B}^{(b)}$, let ${\cal S}_i = {\rm Neib}(C_i) \cap {\cal A}^{(a)}$. Constraint ${\rm C2}$ requires ${\rm rk}(C_i^{({\cal S}_i)}) = {\rm rk}({\bf G}_i^{({\cal S}_i)}{\bf H}_i) < |{\cal S}_i|$. We first consider a conditional probability $\Pr\left\{|{\cal S}_i| = k, {\rm rk}(C_i^{({\cal S}_i)}) < k \bigm| {\rm dg}(C_i) = d\right\} \triangleq \xi_{d,k}$, where $d \in [1:K]$ and $k \in [1:{\min\{a,d\}}]$. We can write
\begin{equation}\nonumber
	\begin{aligned}
		\xi_{d,k}& = \Pr\left\{{\rm rk}(C_i^{({\cal S}_i)}) < k \bigm| |{\cal S}_i| = k, {\rm dg}(C_i) = d\right\}\\
		&~~~\cdot  \Pr\left\{|{\cal S}_i| = k \mid {\rm dg}(C_i) = d\right\}\\
		& = \Pr\left\{{\rm rk}(C_i^{({\cal S}_i)}) < k \bigm| |{\cal S}_i| = k\right\} \\
		&~~~\cdot \Pr\left\{|{\cal S}_i| = k \mid {\rm dg}(C_i) = d\right\}\\
		& = (1-\hbar_k') {\rm hyge}(k;K,a,d).
	\end{aligned}
\end{equation}
The last equality merits some explanations. Base on the definition of $\hbar_k'$ (see (\ref{eq_hbar_prime})), we can write 
\begin{equation}\nonumber
	\Pr\left\{{\rm rk}(C_i^{({\cal S}_i)}) < k \bigm| |{\cal S}_i| = k\right\} = 1-\hbar_k'.
\end{equation}
The event $|{\cal S}_i| = k$ given ${\rm dg}(C_i) = d$ indicates that $C_i$ has $k$ and $d-k$ neighbors in ${\cal A}^{(a)}$ and ${\cal V}\backslash{\cal A}^{(a)}$, respectively. Since $\{{\cal A}^{(a)},{\cal V}\backslash{\cal A}^{(a)}\}$ is a partition of $\cal V$, with Lemma \ref{Lemma_BCN_neighbor_partition} we have
\begin{equation}\nonumber
	\Pr\left\{|{\cal S}_i| = k \mid {\rm dg}(C_i) = d\right\} = \frac{\binom{a}{k}\binom{K-a}{d-k}}{\binom{K}{d}} = {\rm hyge}(k;K,a,d).
\end{equation}
Further, by summing over $k \in [1:\min\{a,d\}]$ and $d \in [1:K]$, we have
\begin{equation}\nonumber
	\begin{aligned}
		&~~~\Pr\left\{{\rm rk}(C_i^{({\cal S}_i)}) < |{\cal S}_i|\right\} \\
		&= \sum_{d, k} \Pr\{{\rm dg}(C_i) = d\} \\
		&~~~\cdot \Pr\left\{|{\cal S}_i| = k, {\rm rk}(C_i^{({\cal S}_i)}) < k \bigm| {\rm dg}(C_i) = d\right\}\\
		& = \sum_{d=1}^{K} \Psi_d \left(\sum_{k=1}^{\min\{a,d\}} \xi_{d,k} \right).
	\end{aligned}
\end{equation}
Considering the independence of batches, we can write
\begin{equation}\nonumber
	\begin{aligned}
	\Pr\{E_2(a,b)\} &= \left(\Pr\left\{{\rm rk}(C_i^{({\cal S}_i)}) < |{\cal S}_i|\right\}\right)^b \\
	&= \left(\sum_{d=1}^{K} \Psi_d \left(\sum_{k=1}^{\min\{a,d\}} \xi_{d,k} \right)\right)^b.
	\end{aligned}
\end{equation}

Next, we proceed with the derivation of $\Pr\left\{E_3(a,b)|E_1(a,b)\right\}$. Under the condition $E_1(a,b)$, we need to use the ``degree distribution'' ${\bm \Phi}^{(K-a)} = ({\Phi}^{(K-a)}_1,{\Phi}^{(K-a)}_2,\ldots,{\Phi}^{(K-a)}_{K-a})$ for the $n-b$ B-CNs $\notin {\cal B}^{(b)}$, which is conditioned on $E_1(a,b)$. Since ${\Phi}^{(K-a)}_{d} = 0$ for $d > K-a$, we omit these terms from ${\bm \Phi}^{(K-a)}$. Note that ${\bm \Phi}^{(K-a)}$ is a probability vector when $K-a \ge D_{\rm min}$ and an all-zero vector otherwise. $\Pr\{E_3(a,b)|E_1(a,b)\}$ can be obtained using the following lemma.

\begin{lemma}\label{Lemma_Union_Card}
	Given a universal set $U = \{1,2,\ldots,m\}$ and $s$ independent random subsets $S_i \subset U$, $i \in [1:s]$. For each $S_i$, we have $\Pr\{|S_i| = j\} = \beta_j$, where ${\bm \beta} = (\beta_0,\beta_1,\ldots,\beta_m)$ is a probability distribution. Let $S = \bigcup_{i=1}^{s} S_i$. The cardinality of $S$ follows the distribution
	\begin{equation}\label{eq_union_card}
		\Pr\{|S| = j\} = \binom{m}{j} \sum_{k=0}^{j} (-1)^{j-k} \binom{j}{k} \left(\sum_{d = 0}^{k} \beta_d \frac{\binom{k}{d}}{\binom{m}{d}}\right)^s.
	\end{equation}
\end{lemma}
\begin{proof}
Two proofs are provided in Appendix \ref{Appendix_Proof_Lemma_Union_Card}. The first proof is based on a Markov process with a diagonalizable transition probability matrix which is essentially the same as the proof of \cite[Lemma 10]{FL_analysis_BATS}, and the second proof is new. The second proof is more direct and concise, but we still provide the first proof because using some intermediate results in the first proof instead of directly using (\ref{eq_union_card}) may be more convenient for numerical calculations. More details are discussed in Remark~\ref{remark_union_card}.
\end{proof}

Using the function
\begin{multline}
	U(j;m,s,{\bm \Phi}^{(m)}) \\ \triangleq \binom{m}{j} \sum_{k=0}^{j} (-1)^{j-k} \binom{j}{k} \left(\sum_{d = 0}^{k} \Phi_d^{(m)} \frac{\binom{k}{d}}{\binom{m}{d}}\right)^s,
\end{multline}
we can write $\Pr\{E_3(a,b)|E_1(a,b)\} = U(K-a;K-a,n-b,{\bm \Phi}^{(K-a)})$.

By substituting $\Pr\{E_1(a,b)\}$, $\Pr\{E_2(a,b)\}$, and $\Pr\{E_3(a,b)|E_1(a,b)\}$ into (\ref{eq_probability_decomposition}), we obtain $\Pr\{{\cal T}^* \in {G_{\rm stop}({\cal A}^{(a)},{\cal B}^{(b)})}\}$. Subsequently, by substituting $\Pr\{{\cal T}^* \in {G_{\rm stop}({\cal A}^{(a)},{\cal B}^{(b)})}\}$ into (\ref{eq_P_err_Union}), we obtain an upper bound for standard BATS codes, which is essentially Corollary~\ref{corollary_std_BATS_BP}.



\subsection{Upper Bounds for Regular LDPC Ensemble}\label{subsection:proof_extended_stopping_set}
In this subsection, we provide two upper bounds on the probability that a subgraph of a regular LDPC code contains a specific maximum stopping set. The bounding technique in this subsection is inspired by the method introduced in \cite[Ch. 3.22]{Modern_Coding_Theory}.

A \textit{stopping set} of an LDPC code is defined as the set ${\cal S}_v$ of VNs such that all neighbors of ${\cal S}_v$ are connected to ${\cal S}_v$ at least twice \cite{FL_analysis_LDPC}. For ease of explaining some details in the sequel, we may include the neighbors of ${\cal S}_v$ into the definition of a stopping set. Assuming that ${\cal S}_v = \{V_1,\ldots,V_a\}$, let ${\cal S}_c = \bigcup_{V_i \in {\cal S}_v} {\rm Neib}(V_i)$ be the set including all neighbors of ${\cal S}_v$. We may use $({\cal S}_v,{\cal S}_c)$ or $(|{\cal S}_v|,|{\cal S}_c|)$ to denote a stopping set interchangeably. Then, we define the extended stopping set as a generalization of the traditional stopping set.
\begin{definition}[Extended Stopping Set]\label{def:ess}
	An extended stopping set is defined as a pair $({\cal A},{\cal A}')$ of disjoint VN sets such that ${\cal A}$, ${\cal A} \ne \emptyset$, is the maximal stopping set which can be formed by the VNs in ${\cal A} \cup {\cal A}'$.
\end{definition}

Note that ${\cal A}'$ can be empty in an $({\cal A},{\cal A}')$ extended stopping set. In this case, it is identical to a stopping set. When we say an $({\cal A},{\cal A}')$ extended stopping set, the selection of the VNs are fixed, i.e., ${\cal A}$ and ${\cal A}'$ are two fixed sets.

The main results of this subsection are the following two lemmas.
\begin{lemma}\label{lemma:ess_recursive}
	Consider the $(K,\mathtt{l},\mathtt{r},\nu_{\rm min})_q$ LDPC ensemble. Let ${\cal A}$ and ${\cal A}'$ be two disjoint sets of VNs with $|{\cal A}| = a > 0$ and $|{\cal A}'| = a' \ge 0$. The probability $p(a,a')$ that in a randomly chosen graph from the $(K,\mathtt{l},\mathtt{r},\nu_{\rm min})_q$ LDPC ensemble, the pair $({\cal A},{\cal A}')$ is the $({\cal A},{\cal A}')$ extended stopping set can be bounded by
	\begin{equation}\nonumber
		p(a,a') \le \frac{\sum_{t = 0}^{K\mathtt{l}/\mathtt{r}}\sum_{s = 0}^{K\mathtt{l}/\mathtt{r}} \binom{K\mathtt{l}/\mathtt{r}}{t+s} \tilde{A}(a+a',t,s;a)}{\binom{K\mathtt{l}}{(a+a')\mathtt{l}} ((a+a')\mathtt{l})!},
	\end{equation}
	where $\tilde{A}(v,t,s;a)$ is determined by the recursion for $v \in [a+1:K]$ and $t,s \in [0:K\mathtt{l}/\mathtt{r}]$:
	\begin{equation}\nonumber
		\begin{aligned}
			&\tilde{A}(v,t,s;a) =\sum_{\Delta s = 1}^{\mathtt{l}} \sum_{\sigma = 0}^{\mathtt{l}-\Delta s} \sum_{\Delta t = 0}^{\lfloor \frac{\mathtt{l}-\Delta s - \sigma}{2} \rfloor} \sum_{\tau = 0}^{\mathtt{l}-\Delta s -\sigma - 2\Delta t} \\
			&\tilde{A}(v-1,t-\Delta t-\sigma,s+\sigma-\Delta s;a) \\
			& \cdot {(v-a)} (\mathtt{l})! \binom{t+s}{\Delta t + \Delta s} \binom{\Delta t + \Delta s}{\Delta t} \\
			&\cdot {\rm coef}\left\{ \left((1+x)^{\mathtt{r}-1}-1\right)^\sigma \left((1+x)^\mathtt{r}-1-\mathtt{r}x\right)^{\Delta t},x^{\mathtt{l}-\Delta s-\tau} \right\}\\
			& \cdot \binom{(t-\Delta t-\sigma)\mathtt{r}-(v-1)\mathtt{l}+s+\sigma-\Delta s}{\tau} \binom{s+\sigma-\Delta s}{\sigma}\\
			& \cdot \mathtt{r}^{\Delta s} \frac{\Delta s}{s}
		\end{aligned}
	\end{equation}
	with the initial step for $v \in [1:K]$ and $t,s \in [0:K\mathtt{l}/\mathtt{r}]$
	\begin{equation}\nonumber
		\begin{aligned}
		\tilde{A}(v,t,s;a) = &{\rm coef}\{ ((1+x)^\mathtt{r}-1-\mathtt{r}x)^{t},x^{v\mathtt{l}} \}(v\mathtt{l})! \\
		&\cdot \mathds{1}_{\{s=0\}} \mathds{1}_{\{v=a\}} \mathds{1}_{\{v \ge \nu_{\rm min}\}}.
		\end{aligned}
	\end{equation}
\end{lemma}
\begin{proof}
	The proof of this lemma is largely based on \cite[Ch. 3.22]{Modern_Coding_Theory}. Noting that this lemma focuses on the fixed VNs, we first define a subgraph structure for the fixed VNs.
	
	Fix $v$ VNs and order the corresponding $v\mathtt{l}$ VN sockets as well as all $K\mathtt{l}$ L-CN sockets in some fixed but arbitrary way. A \textit{constellation} is a particular choice of how to attach the $v\mathtt{l}$ VN sockets to the $K\mathtt{l}$ L-CN sockets. Here, attaching an edge includes two steps: 1) Place the edge between two sockets, 2) label an non-zero element on the edge. Let $T(v)$ denote the number of all such constellations, for which $v$ VNs are fixed but the neighboring L-CNs are arbitrary. There are $\binom{K\mathtt{l}}{v\mathtt{l}}$ ways to choose L-CN sockets, $(v\mathtt{l})!$ ways to permute the edges, and $(q-1)^{v\mathtt{l}}$ ways to label non-zero elements on the edges, resulting in $T(v) = \binom{K\mathtt{l}}{v\mathtt{l}} (v\mathtt{l})! (q-1)^{v\mathtt{l}}$. A Tanner graph in the $(K,\mathtt{l},\mathtt{r},\nu_{\rm min})_q$ LDPC ensemble can be regarded as a constellation with $K$ VNs, and $T(K) = (K\mathtt{l})! (q-1)^{K\mathtt{l}}$ is consistent with the size of the ensemble.
	
	To compute the probability $p(a,a')$, we need to count how many constellations are an $({\cal A},{\cal A}')$ extended stopping set, denoted by $B(a,a')$. Then, we have 
	\begin{equation}\label{eq:p_a_a_prime}
		p(a,a') = \frac{B(a,a')}{T(a+a')}.
	\end{equation}
	
	Based on the fact that the selection of non-zero elements on the edges of a Tanner graph does not essentially affect the BP decoding of this graph, we only need to consider the $(K,\mathtt{l},\mathtt{r},\nu_{\rm min})_2$ LDPC ensemble in the remainder of the proof. More precisely, when considering the $q$-ary ensemble, both $B(a,a')$ and $T(a+a')$ contain the same factor $(q-1)^{(a+a')\mathtt{l}}$, representing the number of ways to label non-zero elements on the edges.
	
	To bound $B(a,a')$, we use a method based on the recursion in \cite[Ch. 3.22]{Modern_Coding_Theory}. The remainder of the proof is given in Appendix~\ref{appendix:ess}.
\end{proof}

Lemma~\ref{lemma:ess_recursive} does not provide an explicit expression for the upper bound on $p(a,a')$. Then, the following lemma relaxes this upper bound and provides an explicit expression.

\begin{lemma}\label{lemma:ess}
	Consider the $(K,\mathtt{l},\mathtt{r},\nu_{\rm min})_q$ LDPC ensemble. The probability $p(a,a')$ (defined in Lemma~\ref{lemma:ess_recursive}) can be bounded by
	\begin{equation}\nonumber
		p(a,a') \le L(a,a';K,\mathtt{l},\mathtt{r}).
	\end{equation}
	Here,
	\begin{multline}\nonumber
		L(a,a';K,\mathtt{l},\mathtt{r}) \\=  \sum_{c = 0}^{K\mathtt{l}/\mathtt{r}} \min \left\{ \binom{K\mathtt{l}/\mathtt{r}}{c} \hat{p}(a,0,c), \binom{K\mathtt{l}/\mathtt{r}}{a'+c} \hat{p}(a,a',c) \right\},
	\end{multline}
	where
	\begin{equation}\nonumber
		\begin{aligned}
			&\hat{p}(a,a',c) =\\
			&\begin{cases}
				\begin{footnotesize}	
					\begin{aligned}
						&\frac{{\rm coef}\{ ((1+x)^\mathtt{r}-1-\mathtt{r}x)^{c},x^{a\mathtt{l}} \}}{\binom{K\mathtt{l}}{a\mathtt{l}}} \\
						&\cdot \prod_{t=1}^{a'} \frac{t(c+t) \mathtt{r} \binom{(K\mathtt{l}/\mathtt{r}-a'+t-1)\mathtt{r}-(a+t-1)\mathtt{l}}{\mathtt{l}-1} (\mathtt{l})!}{\prod_{i=0}^{\mathtt{l}-1} (K\mathtt{l}-(a+t-1)\mathtt{l}-i)} 
					\end{aligned}
				\end{footnotesize}
				&\begin{aligned}
					&\text{if }a \ge \nu_{\rm min}\\ &~~\text{and } a' \ge 1,
				\end{aligned}\\[1cm]
				\begin{footnotesize}
					\begin{aligned}
						\frac{{\rm coef}\{ ((1+x)^\mathtt{r}-1-\mathtt{r}x)^{c},x^{a\mathtt{l}} \}}{\binom{K\mathtt{l}}{a\mathtt{l}}} 
					\end{aligned}
				\end{footnotesize}
				&\begin{aligned}
					&\text{if }a \ge \nu_{\rm min}\\ &~~\text{and } a' = 0,
				\end{aligned}\\
				0 & \text{if }a < \nu_{\rm min}.	
			\end{cases}
		\end{aligned}
	\end{equation}
\end{lemma}
\begin{proof}
	Similar to the proof of Lemma~\ref{lemma:ess_recursive}, we recursively bound $B(a,a')$. However, the recursion considered here is more relaxed, allowing us to obtain the closed-form upper bound. The proof is given in Appendix \ref{appendix:ess}.
\end{proof}

\subsection{Proof of Theorem \ref{theorem_LDPC_BATS_BP}}\label{subsection:proof:theorem_1}
We now utilize the bounding techniques in the above two subsections to prove Theorem \ref{theorem_LDPC_BATS_BP}. Let ${\tilde{\cal T}} = ({\cal V},{\cal C},{\cal E},\tilde{\cal C},\tilde{\cal E})$ be a Tanner graph in the $(K,n,M,{\bm \Psi},{\bf h},\mathtt{l},\mathtt{r},\nu_{\rm min})_q$ LDPC-BATS ensemble defined in Definition \ref{def:LDPC_BATS_ensemble}, where ${\cal V}$ is the VN set, ${\cal C}$ is the B-CN set, ${\cal E}$ is the edge set consisting of the edges connecting ${\cal V}$ and ${\cal C}$, $\tilde{\cal C}$ is the L-CN set, and $\tilde{\cal E}$ is the edge set consisting of the edges connecting ${\cal V}$ and $\tilde{\cal C}$. Noting the equivalence between L-CNs and B-CNs, we can combine ${\cal C}$ and $\tilde{\cal C}$ as well as ${\cal E}$ and $\tilde{\cal E}$ to be a single set. However, distinguishing them as different sets is beneficial for the clarity of the subsequent proof. 

Following the proof in Appendix~\ref{subsection:proof_BP_std_BATS}, we define a special class of LDPC-BATS Tanner graphs called the \textit{LDPC-BATS stopping graphs}. Then, we prove that BP decoding of ${\tilde{\cal T}}$ fails if and only if ${\tilde{\cal T}}$ is an LDPC-BATS stopping graph.

\begin{definition}[LDPC-BATS Stopping Graph]\label{def:LDPC_BATS_stopping_graph}
	A Tanner graph $\tilde{\cal T} = ({\cal V},{\cal C},{\cal E},\tilde{\cal C},\tilde{\cal E})$ in the LDPC-BATS ensemble is a stopping graph if there exist a non-empty set ${\cal A} \subset {\cal V}$ and a set (can be empty) ${\cal B} \subset {\cal C}$ such that the following constraints holds:
	\begin{itemize}
		\item[${\rm C1}$:] For any $V \in {\cal A}$, ${\rm Neib}_{\cal E}(V) \cap ({\cal C} \backslash {\cal B}) = \emptyset$.
		\item[${\rm C2}$:] For any $C \in {\cal B}$, $|{\rm Neib}(C) \cap {\cal A}| > {\rm rk}(C^{({\rm Neib}(C) \cap {\cal A})}) \ge 0$.
		\item[${\rm C4}$:] Let ${\cal X} \triangleq \{V \in ({\cal V \backslash \cal A}): {\rm Neib}_{\cal E}(V) \cap ({\cal C} \backslash {\cal B}) = \emptyset\}$. The cardinality of $\cal X$ satisfies $|{\cal X}| \le s_{\rm max}(a)$, where $s_{\rm max}(a) = \left\lfloor \frac{(K-a)\mathtt{l}}{\mathtt{r}} \right\rfloor$. If $\mathtt{l} > 0$, $({\cal A},{\cal X})$ is an extended stopping set whose maximal stopping set is ${\cal A}$.
	\end{itemize}
\end{definition}

\begin{remark}
	Compared with Definition \ref{Def_stopping_graph_BATS}, only the last constraint ${\rm C3}$ is replaced by ${\rm C4}$. When $\mathtt{l} = 0$, an LDPC-BATS code is just a standard BATS code. In this case, we have $s_{\rm max}(a) = 0$, and so Definition \ref{def:LDPC_BATS_stopping_graph} becomes identical to Definition \ref{Def_stopping_graph_BATS}.
\end{remark}

For a parameter pair $({\cal A},{\cal B})$ in $\tilde{\cal T}$, we say that $({\cal A},{\cal B})$ is \textit{valid} for $\tilde{\cal T}$ if $({\cal A},{\cal B})$ satisfies the three constraints ${\rm C1}$, ${\rm C2}$, and ${\rm C4}$.

\begin{lemma}\label{Lemma_LDPCBATS_Stops_at_Stopping_Graph}
	An LDPC-BATS Tanner graph $\tilde{\cal T} = ({\cal V},{\cal C},{\cal E},\tilde{\cal C},\tilde{\cal E})$ in the LDPC-BATS ensemble cannot be reduced to an empty graph under BP (peeling) decoding if and only if $\tilde{\cal T}$ is a stopping graph.
\end{lemma}
\begin{proof}
	See Appendix \ref{appendix:proofs_for_BP}.
\end{proof}

In light of the above lemma, the method for bounding the BP performance of standard BATS ensembles can be leveraged to bound the BP performance of LDPC-BATS ensembles. For each $({\cal A}^{(a)},{\cal B}^{(b)})$ (recall that ${\cal A}^{(a)}$ and ${\cal B}^{(b)}$ represent the sets of fixed $a$ VNs and fixed $b$ B-CNs, respectively), form a set of LDPC-BATS stopping graphs as
\begin{equation}\nonumber
	\begin{aligned}
		\tilde{G}_{\rm stop}({\cal A}^{(a)},{\cal B}^{(b)}) \triangleq \left\{\tilde{\cal T}: ({\cal A}^{(a)},{\cal B}^{(b)}) \text{ is valid for }\tilde{\cal T}\right\}.
	\end{aligned}
\end{equation}
Let $\tilde{\cal T}^*$ be a random graph chosen from the $(K,n,M,{\bm \Psi},{\bf h},\mathtt{l},\mathtt{r})_q$ LDPC-BATS ensemble uniformly at random. Then, the ensemble-average error probability $\mathbb{E}\left[P_{\textsf{E}}^{\rm BP}(\tilde{\cal T})\right]$ for the LDPC-BATS ensemble under BP decoding can be bounded by
\begin{equation}\label{eq_union_LDPCBATS}
	\begin{aligned}
		&\mathbb{E}\left[P_{\textsf{E}}^{\rm BP}(\tilde{\cal T})\right] = \Pr\left\{\tilde{\cal T}^* \in \bigcup\limits_{{\cal A} \ne \emptyset,{\cal A} \subset {\cal V},{\cal B}} \tilde{G}_{\rm stop}({\cal A},{\cal B})\right\} \le \\
		&\min\left\{ 1, \sum_{a=1}^{K} \sum_{b=0}^{n} \binom{K}{a} \binom{n}{b} \Pr\left\{\tilde{\cal T}^* \in \tilde{G}_{\rm stop}({\cal A}^{(a)},{\cal B}^{(b)})\right\}\right\}.
	\end{aligned}
\end{equation}

We now study $\Pr\left\{\tilde{\cal T}^* \in \tilde{G}_{\rm stop}({\cal A}^{(a)},{\cal B}^{(b)})\right\}$ in the last line of the above formula. For constraints ${\rm C1}$, ${\rm C2}$, and ${\rm C4}$, let $E_i(a,b)$ be the event that $({\cal A}^{(a)},{\cal B}^{(b)})$ in $\tilde{\cal T}^*$ satisfies constraint ${\rm C}i$, $i=1,2,4$. Similar to the derivation of  (\ref{eq_probability_decomposition}), we can decompose $\Pr\left\{\tilde{\cal T}^* \in \tilde{G}_{\rm stop}({\cal A}^{(a)},{\cal B}^{(b)})\right\}$ as 
\begin{multline}\label{eq_T_in_G}
	\Pr\left\{\tilde{\cal T}^* \in \tilde{G}_{\rm stop}({\cal A}^{(a)},{\cal B}^{(b)})\right\} = \Pr\{E_1(a,b)\} \\
	\cdot \Pr\{E_2(a,b)\} \Pr\{E_4(a,b)|E_1(a,b)\}.
\end{multline}
Note that the first two terms in (\ref{eq_T_in_G}) are only affected by the BATS part of the graph $\tilde{\cal T}^*$ and they are exactly the same as the corresponding formulae derived in Appendix~\ref{subsection:proof_BP_std_BATS}. It remains to compute the probability $\Pr\{E_4(a,b)|E_1(a,b)\}$. Let ${\cal X}^* = \{V \in ({{\cal V} \backslash {\cal A}^{(a)}}): {\rm Neib}_{\cal E}(V) \cap ({\cal C} \backslash {\cal B}^{(b)}) = \emptyset\}$ in $\tilde{\cal T}^*$. According to Lemma \ref{Lemma_Union_Card}, we have the distribution of $|{\cal X}^*|$ given $E_1(a,b)$ as
\begin{multline}\label{eq_theorem2_1}
	\Pr\{|{\cal X}^*| = j \mid E_1(a,b)\} \\= U\left(K-a-j;K-a,n-b,{\bm \Phi}^{(K-a)}\right).
\end{multline}
There is a subtle point: Unlike ${\cal A}^{(a)}$ and ${\cal B}^{(b)}$, ${\cal X}^*$ is a random variable rather than a fixed set because $\tilde{\cal T}^*$ is a random graph. ${\cal X}^*$ can be any subset of $({\cal V}\backslash{\cal A}^{(a)})$. To avoid any confusion, let ${\cal X}^{(j)}$ be an instance of ${\cal X}^*$, which is of size $j$, i.e., ${\cal X}^{(j)}$ is a set of $j$ specified VNs in $({\cal V}\backslash{\cal A}^{(a)})$. Denote by $\check{E}({\cal A}^{(a)},{\cal X}^{(j)})$ the event that $({\cal A}^{(a)},{\cal X}^{(j)})$ is the extended stopping set with the maximal stopping set ${\cal A}^{(a)}$. The probability of this event can be bounded by Lemma~\ref{lemma:ess_recursive} and Lemma~\ref{lemma:ess}. To develop a closed-form expression, using Lemma~\ref{lemma:ess}, we have
\begin{equation}\label{eq_theorem2_2}
	\begin{aligned}
		\Pr\left\{\check{E}({\cal A}^{(a)},{\cal X}^{(j)}) \right\} \le L(a,j;K,\mathtt{l},\mathtt{r}).
	\end{aligned}
\end{equation}
Using the Total Probability Theorem, we can develop the probability $\Pr\{E_4(a,b)|E_1(a,b)\}$ as (\ref{eq:Pr_E4_under_E1}).
\begin{figure*}[!t]
	\normalsize
\begin{equation}\label{eq:Pr_E4_under_E1}
	\begin{aligned}
		&\Pr\left\{ E_4(a,b) \mid E_1(a,b)\right\} = \sum_{j=0}^{s_{\rm max}(a)} \Pr\left\{ |{\cal X}^*| = j \mid E_1(a,b)\right\} \Pr \left\{ E_4(a,b) \mid |{\cal X}^*| = j, E_1(a,b)\right\}\\
		& = \sum_{j=0}^{s_{\rm max}(a)} \Pr\left\{ |{\cal X}^*| = j \mid E_1(a,b)\right\} \left(\sum_{{\cal X}^{(j)} \subset ({\cal V}\backslash{\cal A}^{(a)})} \Pr \left\{ E_4(a,b), {\cal X}^* = {\cal X}^{(j)} \mid |{\cal X}^*| = j, E_1(a,b)\right\}\right)\\
		& = \sum_{j=0}^{s_{\rm max}(a)} \Pr\left\{ |{\cal X}^*| = j \mid E_1(a,b)\right\} \biggg(\sum_{{\cal X}^{(j)} \subset ({\cal V}\backslash{\cal A}^{(a)})} \Pr \left\{ E_4(a,b)  \mid {\cal X}^* = {\cal X}^{(j)}, |{\cal X}^*| = j, E_1(a,b)\right\} \\
		&~~~~\cdot \Pr \left\{ {\cal X}^* = {\cal X}^{(j)} \mid |{\cal X}^*| = j, E_1(a,b) \right\}\biggg)\\
		& \overset{(a)}{=} \sum_{j=0}^{s_{\rm max}(a)} \Pr\left\{ |{\cal X}^*| = j \mid E_1(a,b)\right\} \Pr\left\{\check{E}({\cal A}^{(a)},{\cal X}^{(j)}) \right\} \left(\sum_{{\cal X}^{(j)} \subset ({\cal V}\backslash{\cal A}^{(a)})} \Pr \left\{ {\cal X}^* = {\cal X}^{(j)} \mid |{\cal X}^*| = j, E_1(a,b)\right\}\right)\\
		& \overset{(b)}{=} \sum_{j=0}^{s_{\rm max}(a)} \Pr\left\{ |{\cal X}^*| = j \mid E_1(a,b)\right\} \Pr\left\{\check{E}({\cal A}^{(a)},{\cal X}^{(j)}) \right\} \\
		& \overset{(c)}{\le} \sum_{j=0}^{s_{\rm max}(a)} U\left(K-a-j;K-a,n-b,{\bm \Phi}^{(K-a)}\right) \cdot L(a,j;K,\mathtt{l},\mathtt{r}),
	\end{aligned}
\end{equation}
\hrulefill
\end{figure*}
In (\ref{eq:Pr_E4_under_E1}), equality $(a)$, $(b)$, and inequality $(c)$ are explained as follows. Given ${\cal X}^* = {\cal X}^{(j)}$, the event $E_4(a,b)$ is independent of $|{\cal X}^*|$ and $E_1(a,b)$, thus yielding
\begin{equation}\nonumber
	\begin{aligned}
	&\Pr \left\{ E_4(a,b)  \mid {\cal X}^* = {\cal X}^{(j)}, |{\cal X}^*| = j, E_1(a,b)\right\} \\
	&= \Pr \left\{ E_4(a,b)  \mid {\cal X}^* = {\cal X}^{(j)}\right\} \\
	&= \Pr\left\{\check{E}({\cal A}^{(a)},{\cal X}^{(j)}) \right\}.
	\end{aligned}
\end{equation}
Observing that the value of $\Pr\left\{\check{E}({\cal A}^{(a)},{\cal X}^{(j)}) \right\}$ is independent of the particular selections of $j$ VNs in ${\cal X}^{(j)}$, we can move this term outside the summation and then equality $(a)$ follows. Equality $(b)$ is obtained by noting that
\begin{equation}\nonumber
	\sum_{{\cal X}^{(j)} \subset ({\cal V}\backslash{\cal A})} \Pr \left\{ {\cal X}^* = {\cal X}^{(j)} \mid |{\cal X}^*| = j, E_1(a,b)\right\} = 1.
\end{equation} 
Inequality (c) is a direct application of (\ref{eq_theorem2_1}) and (\ref{eq_theorem2_2}). 

When considering standard BATS codes without precodes, i.e., set $\mathtt{l} = 0$, $\Pr\left\{ E_4(a,b) \mid E_1(a,b)\right\}$ reduces to $U\left(K-a;K-a,n-b,{\bm \Phi}^{(K-a)}\right)$ because $s_{\rm max}(a) = 0$ and $L(a,j=0;K,\mathtt{l}=0,\mathtt{r}) = 1$, which is identical to $\Pr\left\{ E_3(a,b) \mid E_1(a,b)\right\}$ derived in Appendix~\ref{subsection:proof_BP_std_BATS}. The proof is completed by substituting $\Pr\left\{ E_1(a,b)\right\}$, $\Pr\left\{ E_2(a,b)\right\}$, and $\Pr\left\{ E_4(a,b) \mid E_1(a,b)\right\}$ into (\ref{eq_T_in_G}) and then substituting (\ref{eq_T_in_G}) into (\ref{eq_union_LDPCBATS}). 

\section{Derivation of the Bound on the Error Probability of BATS Codes Under ML Decoding}\label{Sec_proof_Theorem_3}
This appendix proves Theorem~\ref{theorem_LDPC_BATS_INA}. Consider a precode $\mathscr{C}$ with dimension $K'$ and code length $K$. Let ${\bf G}_p$ denote the $K' \times K$ generator matrix of $\mathscr{C}$. Let ${\bf u} = (u_1,u_2,\ldots,u_{K'})$ be the vector of input symbols and then ${\bf v} = {\bf u} {\bf G}_p$ is the vector of intermediate symbols. In the following, a batch stands for a standard batch (i.e., B-CN), excluding the ``extra batches'' (i.e., L-CN) implied by the parity-check equations of the precode. For each batch $i$ (i.e., B-CN $C_i$), let $\tilde{\bf G}_i'$ be the $K \times M$ matrix obtained by adding $K-{\rm dg}_i$ all-zero rows to ${\bf G}_i$, such that the rows of $\tilde{\bf G}_i'$ with indices in ${\cal I}_i$ form the original generator matrix ${\bf G}_i$, where ${\cal I}_i$ is the index set of $C_i$'s neighbors. Then, the batch equation for the $i$-th batch can be rewritten as 
\begin{equation}\nonumber
	{\bf u} {\bf G}_p \tilde{\bf G}_i' {\bf H}_i = {\bf Y}_i.
\end{equation}
Recall that the transfer matrix ${\bf H}_i$ is an $M \times N_i$ matrix, where $N_i$ is the number of received symbols for the $i$-th batch. Assume the rank of each transfer matrix independently follows the empirical rank distribution ${\bf h} = (h_0,h_1,\ldots,h_M)$. Form a $K \times (\sum_{i=1}^{n} N_i)$ matrix $\tilde{\bf G} \triangleq [\tilde{\bf G}_1' {\bf H}_1, \tilde{\bf G}_2' {\bf H}_2, \ldots, \tilde{\bf G}_n' {\bf H}_n]$ and a $1 \times (\sum_{i=1}^{n} N_i)$ row vector ${\bf Y} = [{\bf Y}_1,{\bf Y}_2,\ldots,{\bf Y}_n]$. The linear system of a precoded BATS code be written as
\begin{equation}\label{eq_whole_linear_system}
	{\bf u} {\bf G}_p \tilde{\bf G} = {\bf Y}.
\end{equation}
A unique solution of the linear system (\ref{eq_whole_linear_system}) can be found if and only if ${{\rm rk} ({\bf G}_p\tilde{\bf G})} = K'$. It is clear that the choices of $\bf u$ does not affect the solvability of such a linear system. Without loss of generality, assume that we transmit the all-zero vector ${\bf u} = {\bf 0}$ thus leading to the all-zero codeword ${\bf v} = {\bf 0}$. In this case, ${\bf Y}$ must be the all-zero vector. An ML decoding failure occurs if and only if there exists ${\bf u} \in \mathbb{F}_q^{K'} \setminus \{{\bf 0}\}$ such that ${\bf u} {\bf G}_p \tilde{\bf G} = {\bf 0}$. This event is equivalent to the event that there exists ${\bf v} \in \mathscr{C} \setminus \{{\bf 0}\}$ such that ${\bf v} \tilde{\bf G} = {\bf 0}$. 

Let $\alpha$ be a primitive element of the base field and the elements in this field are represented by $\{0,\alpha^0,\alpha^1,\ldots,\alpha^{q-2}\}$. Let $w_H({\bf v})$ be the Hamming weight of a vector $\bf v$, which equals the number of non-zero elements in $\bf v$. Split the codebook into $K+1$ sets $\mathscr{C}_l = \{{\bf v} \in \mathscr{C} : w_H({\bf v}) = l\}$, $l \in [0:K]$. Let $A_l \triangleq |\mathscr{C}_l|$, which is known as the \textit{weight enumerator} of the code $\mathscr{C}$\footnote{When no precode is applied to BATS codes, we have $A_l = \binom{K}{l} (q-1)^l$. When we consider a precode chosen from the $q$-ary regular LDPC ensemble, $A_l$ is actually the \textit{expected weight enumerator}, which can be found in \cite{Weight_distribution_NBLDPC}.}. Recall $\mathbb{E}\left[P_{\textsf{E}}^{\rm ML}(\tilde{\cal T})\right]$ is the ensemble-average error probability for the $(K,n,M,{\bm \Psi},{\bf h},\mathtt{l},\mathtt{r},\nu_{\rm min})_q$ LDPC-BATS ensemble under ML decoding. We can develop the union bound on $\mathbb{E}\left[P_{\textsf{E}}^{\rm ML}(\tilde{\cal T})\right]$ as
\begin{equation}\label{eq_P_ina_1}
	\begin{aligned}
		\mathbb{E}\left[P_{\textsf{E}}^{\rm ML}(\tilde{\cal T})\right] & = \Pr\left\{\bigcup_{{\bf v}\in \mathscr{C} \setminus \{{\bf 0}\}} \left({\bf v} \tilde{\bf G} = {\bf 0}\right) \right\}\\
		& \overset{(a)}{\le} \frac{1}{q-1} \sum_{{\bf v}\in \mathscr{C} \setminus \{{\bf 0}\}} \Pr \left\{ {\bf v} \tilde{\bf G} = {\bf 0} \right\}\\
		& = \frac{1}{q-1} \sum_{l = 1}^{K} \sum_{{\bf v}\in \mathscr{C}_l} \Pr \left\{ {\bf v} \tilde{\bf G} = {\bf 0} \right\}\\
		& \overset{(b)}{=} \frac{1}{q-1} \sum_{l = 1}^{K} A_l \Pr \left\{ {\bf v} \tilde{\bf G} = {\bf 0} | w_H({\bf v}) = l \right\},
	\end{aligned}
\end{equation}
where the factor $\frac{1}{q-1}$ in inequality $(a)$ is due to that for any non-zero codeword $\bf v$
\begin{multline}\nonumber
	\Pr \left\{ \bigcup_{i \in [0:q-2]} \left( \alpha^i {\bf v} \tilde{\bf G} = {\bf 0}\right) \right\} \\=  \Pr \left\{ \alpha^{0} {\bf v} \tilde{\bf G} = {\bf 0} \right\} = \cdots = \Pr \left\{ \alpha^{q-2} {\bf v} \tilde{\bf G} = {\bf 0} \right\},
\end{multline}
and equality $(b)$ arises from the fact that $\Pr \left\{ {\bf v} \tilde{\bf G} = {\bf 0} \right\}$ only depends on $w_H({\bf v})$ since each batch $i$ randomly chooses ${\rm dg}_i$ symbols (i.e., independent of the positions of the $l$ non-zero symbols of $\bf v$) for linear combination and the elements of ${\bf G}_i$ that are coefficients for linear combination are totally random (i.e., independent of the values of the $l$ non-zero symbols of $\bf v$).

Since all batches are independent and equivalent, we only need to consider the event $({\bf v} \tilde{\bf G}_i' {\bf H}_i = {\bf 0} \mid w_H({\bf v}) = l)$ for one batch. In the following discussion, we omit the subscripts of $\tilde{\bf G}_i'$, ${\bf H}_i$, ${\rm dg}_i$, and ${\cal I}_i$ when there is no ambiguity, i.e., we will use $\tilde{\bf G}'$, ${\bf H}$, ${\rm dg}$, and $\cal I$ to represent the random variables corresponding to a randomly chosen batch. Let $I$ denote the random variable representing the number of non-zero symbols of $\bf v$ with indices in ${\cal I}$, which are called the \textit{effective non-zero symbols} used for linear combination, i.e., calculation of ${\bf v} \tilde{\bf G}'$. Due to the independence of batches, the unsolved term of (\ref{eq_P_ina_1}) can be computed as
\begin{multline}\label{eq_batch_0}
	\Pr \left\{ {\bf v} \tilde{\bf G} = {\bf 0} | w_H({\bf v}) = l \right\} \\= \left(\Pr \left\{ {\bf v} \tilde{\bf G}' {\bf H} = {\bf 0} | w_H({\bf v}) = l \right\}\right)^n.
\end{multline}

For further developing (\ref{eq_batch_0}), we provide the following lemma for the rank distribution of a random matrix.

\begin{lemma}\label{lemma_rank}
	For a $d \times k$ random matrix $\bf A$ over $\mathbb{F}_q$, ${\rm rk}({\bf A})$ follows the distribution
	\begin{equation}\label{eq_zeta_func_rdk}
		\zeta_{r}^{d,k} \triangleq \Pr\{{\rm rk}({\bf A}) = r\} = \frac{\zeta_r^d \zeta_r^k}{\zeta_r^r q^{(d-r)(k-r)}},
	\end{equation}
	where $\zeta_r^m$ is defined in (\ref{eq_zeta_function}). Also, $\zeta_{r}^{d,k}$ is equivalent to the following formula for BATS codes:
	\begin{equation}\label{eq_zeta_func_rdk_2}
		\zeta_{r}^{d,k} = \Pr\{{\rm rk}({\bf G}_i {\bf H}_i) = r | {\rm dg}_i = d, {\rm rk}({\bf H}_i) = k\},
	\end{equation}
	where ${\bf G}_i$ is a $d \times M$ totally random matrix over $\mathbb{F}_q$, and ${\bf H}_i$ is an $M$-row matrix over $\mathbb{F}_q$ that has rank $k$.
\end{lemma}
\begin{proof}
	The proof of Lemma \ref{lemma_rank} includes two part: 1) Proof of (\ref{eq_zeta_func_rdk}). 2) Proof of the equivalence between (\ref{eq_zeta_func_rdk}) and (\ref{eq_zeta_func_rdk_2}).
		
	\textit{Proof of (\ref{eq_zeta_func_rdk}):} This part is based on \cite[Theorem 25.2]{van2001course} and \cite[Corollary 2.2]{blake2006properties}. Let $V_n(q)$ denote the vector space of dimension $n$ over $\mathbb{F}_q$. The number of distinct $k$-dimensional subspaces of $V_n(q)$, denoted by $\qbin{n}{k}_q$, is
		\begin{equation}
			\qbin{n}{k}_q = \prod_{i=0}^{k-1} \frac{q^{n-i}-1}{q^{k-i}-1}.
		\end{equation}
		The quantities $\qbin{n}{k}_q$ are referred to as Gaussian numbers or Gaussian coefficients. According to \cite{van2001course,blake2006properties}, the number of $d \times k$ matrices over $\mathbb{F}_q$ that have rank $r$ is 
		\begin{equation}\nonumber
			N_q(d,k,r) = \qbin{k}{r}_q \sum_{l=0}^{r} (-1)^{r-l} \qbin{r}{l}_q q^{dl+\binom{r-l}{2}}.
		\end{equation}
		Dividing $N_q(d,k,r)$ by the total number $q^{dk}$ of $d \times k$ matrices over $\mathbb{F}_q$, the rank distribution for a $d \times k$ random matrix over $\mathbb{F}_q$ follows. Using the polynomial identity
		\begin{equation}\nonumber
			\prod_{l=0}^{r-1} (x-q^l) = \sum_{l=0}^{r} (-1)^l \qbin{r}{l}_q q^{\binom{l}{2}} x^{r-l} = \sum_{l=0}^{r} (-1)^{r-l} \qbin{r}{l}_q q^{\binom{r-l}{2}} x^{l},
		\end{equation}
		one can rewrite the expression of the rank distribution with $\zeta_r^m$
		\begin{equation}\nonumber
			\begin{aligned}
				\Pr\{{\rm rk}({\bf A}) = r\} & = \frac{N_q(d,k,r)}{q^{dk}}\\
				& = \frac{\qbin{k}{r}_q \prod_{l=0}^{r-1} (q^d-q^l)}{q^{dk}}\\
				& = \frac{(q^k-1)(q^{k-1}-1)\cdots(q^{k-r+1}-1)}{{\left( {{q^r} - 1} \right)\left( {{q^{r-1}} - 1} \right) \cdots \left( {q - 1} \right)}}\\
				& ~~~\frac{\times (q^d-1)(q^d-q)\cdots(q^d-q^{r-1})}{\times {q^{dk}}}\\
				& = \frac{{\left( {1 - {q^{ - k}}} \right) \left( {1 - {q^{ - k+1}}} \right) \cdots \left( {1 - {q^{ - k + r - 1}}} \right)}}{{\left( {1 - {q^{ - r}}} \right)\left( {1 - {q^{ - r+1}}} \right) \cdots \left( {1 - {q^{ - 1}}} \right)}}\\
				&~~~\frac{\times {\left( {1 - {q^{ - d}}} \right)\left( {1 - {q^{ - d+1}}} \right) \cdots \left( {1 - {q^{ - d + r - 1}}} \right)}}{\times q^{(d-r)(k-r)}}\\
				& = \frac{\zeta_r^d \zeta_r^k}{\zeta_r^r q^{(d-r)(k-r)}}.
			\end{aligned}
		\end{equation}
		
		\textit{Proof of the equivalence between (\ref{eq_zeta_func_rdk}) and (\ref{eq_zeta_func_rdk_2}):} There exists a non-singular square matrix ${\bf P}$ such that the first $k$ rows of ${\bf P} {\bf H}_i$ are full rank and the subsequent $M-k$ rows are all-zero rows. ${\bf G}_i{\bf P}^{-1}$ is equivalent to performing elementary column operations on ${\bf G}_i$, and so ${\bf G}_i {\bf P}^{-1}$ is still a totally random matrix. It follows that 
		\begin{equation}\nonumber
			\begin{aligned}
				&\Pr\{{\rm rk}({\bf G}_i {\bf H}_i) =r | {\rm dg}_i = d, {\rm rk}({\bf H}_i) = k\} \\
				&= \Pr\{{\rm rk}({\bf G}_i {\bf P}^{-1} {\bf P} {\bf H}_i) =r | {\rm dg}_i = d, {\rm rk}({\bf H}_i) = k\}\\
				& \overset{(a)}{=} \Pr\{{\rm rk}\big(({\bf G}_i {\bf P}^{-1})[:,1:k] \cdot ({\bf P} {\bf H}_i)[1:k,:]\big) =r | {\rm dg}_i = d,\\
				&~~~~~~~~~ {\rm rk}({\bf H}_i) = k\}\\
				& \overset{(b)}{=} \Pr\{{\rm rk}\big(({\bf G}_i {\bf P}^{-1})[:,1:k]\big) =r | {\rm dg}_i = d\}\\
				& = \zeta_{r}^{d,k},
			\end{aligned}
		\end{equation}
		where $(a)$ is because the lower $M-k$ rows of $({\bf P} {\bf H}_i)$ are all-zero rows, and $(b)$ is obtained by the fact that right-multiplying $({\bf G}_i {\bf P}^{-1})[:,1:k]$ by a matrix that has full row rank does not change the rank of $({\bf G}_i {\bf P}^{-1})[:,1:k]$.
\end{proof}

For brevity, let $\tilde{\pi}_l \triangleq \Pr \left\{ {\bf v} \tilde{\bf G}' {\bf H} = {\bf 0} | w_H({\bf v}) = l \right\}$. We now proceed to develop $\tilde{\pi}_l$ as
\begin{equation}\label{eq_batch_equals_0}
	\begin{aligned}
		\tilde{\pi}_l &= \sum_{d,i} \Pr \left\{ {\bf v} \tilde{\bf G}' {\bf H} = {\bf 0}| {\rm dg} = d, I = i, w_H({\bf v}) = l \right\} \\
		&~~~\cdot \Pr \left\{ {\rm dg} = d, I = i | w_H({\bf v}) = l \right\} \\
		 & \overset{(a)}{=} \sum_{d,i} \Pr \left\{ {\bf v} \tilde{\bf G}' {\bf H} = {\bf 0}| {\rm dg} = d, I = i\right\} {\rm hyge}(i;K,l,d) \Psi_d\\
		 & = \sum_{d = 1}^{K} \Pr \left\{ {\bf v} \tilde{\bf G}' {\bf H} = {\bf 0}| {\rm dg} = d, I = 0\right\} {\rm hyge}(0;K,l,d) \Psi_d \\
		 &~~~ + \sum_{d = 1}^{K} \sum_{i=1}^{\min\{d,l\}} \Pr \left\{ {\bf v} \tilde{\bf G}' {\bf H} = {\bf 0}| {\rm dg} = d, I =i\right\}\\
		 &~~~\cdot {\rm hyge}(i;K,l,d)\Psi_d\\
		 & \overset{(b)}{=} \sum_{d = 1}^{K} \Psi_d \frac{\binom{K-l}{d}}{\binom{K}{d}} + \sum_{d = 1}^{K} \Psi_d \Pr \left\{ {\bf g} {\bf H} = {\bf 0}| I > 0\right\}\\
		 &~~~\cdot \left(\sum_{i=1}^{\min\{d,l\}}  {\rm hyge}(i;K,l,d)\right)\\
		 & \overset{(c)}{=} \sum_{d = 1}^{K} \Psi_d \left(\frac{\binom{K-l}{d}}{\binom{K}{d}} + \left(1-\frac{\binom{K-l}{d}}{\binom{K}{d}}\right) \left(\sum_{k = 0}^{M} h_k q^{-k}\right)\right)
	\end{aligned}
\end{equation}
where $(a)$ is due to 
\begin{multline}\nonumber
	\Pr \left\{ {\bf v} \tilde{\bf G}'{\bf H} = {\bf 0}| {\rm dg} = d, I = i, w_H({\bf v}) = l\right\} \\= \Pr \left\{ {\bf v} \tilde{\bf G}'{\bf H} = {\bf 0}| {\rm dg} = d, I = i\right\},
\end{multline}
\begin{equation}\nonumber
	\Pr \left\{ I = i | {\rm dg} = d, w_H({\bf v}) = l \right\} = {\rm hyge}(i;K,l,d),
\end{equation}
\begin{equation}\nonumber
	\Pr \left\{ {\rm dg} = d | w_H({\bf v}) = l \right\} = \Pr \left\{ {\rm dg} = d \right\} = \Psi_d,
\end{equation}
and $(b)$ and $(c)$ merit more discussions as follows. Noting that $\Pr \left\{ {\bf v} \tilde{\bf G}' {\bf H} = {\bf 0}| {\rm dg} = d, I = 0\right\} = 1$, the first summation of $(b)$ follows. In the second summation of $(b)$, we define ${\bf g} = {\bf v} \tilde{\bf G}'$. Under the condition $I > 0$, we can form a $1 \times I$ subvector ${\bf v}_{\rm sub}$ of $\bf v$ with the $I$ effective non-zero symbols, and an $I \times M$ submatrix $\tilde{\bf G}'_{\rm sub}$ of $\tilde{\bf G}'$ with the $I$ rows corresponding to the $I$ effective non-zero symbols. Thus, we can write ${\bf g} = {\bf v}_{\rm sub} \tilde{\bf G}'_{\rm sub}$. Since ${\bf v}_{\rm sub}$ is a fixed non-zero vector and $\tilde{\bf G}'_{\rm sub}$ is a totally random matrix, their product ${\bf g}$ is a totally random row vector. Observing $\Pr \left\{ {\bf g} {\bf H} = {\bf 0}| {\rm dg} = d, I =i\right\}$ does not depend on $d$ and $i$ when $I > 0$, we can write for $i > 0$
\begin{equation}
	\Pr \left\{ {\bf g} {\bf H} = {\bf 0}| {\rm dg} = d, I =i\right\} = \Pr \left\{ {\bf g} {\bf H} = {\bf 0}| I > 0\right\}.
\end{equation}
Next, we can apply Lemma \ref{lemma_rank} to obtain $(c)$:
\begin{equation}\nonumber
	\begin{aligned}
		\Pr \left\{ {\bf g} {\bf H} = {\bf 0}| I > 0\right\} 
		& = \sum_{k = 0}^{M} \Pr \left\{ {\bf g} {\bf H} = {\bf 0}, {\rm rk}({\bf H}) = k | I > 0\right\}\\
		& = \sum_{k = 0}^{M} \Pr \left\{ {\rm rk}({\bf H}) = k | I > 0\right\}\\
		&~~~\cdot \Pr \left\{ {\bf g} {\bf H} = {\bf 0} | {\rm rk}({\bf H}) = k, I > 0\right\} \\
		& = \sum_{k = 0}^{M} h_k \Pr \Big\{ {\rm rk}({\bf g} {\bf H}) = 0 | {\rm rk}({\bf H}) = k,\\
		&~~~ {\bf g}\text{ is a totally random row vector}\Big\} \\
		& = \sum_{k = 0}^{M} h_k \zeta_{0}^{1,k}\\
		& = \sum_{k = 0}^{M} h_k q^{-k}.
	\end{aligned}
\end{equation}

By substituting $\Pr \left\{ {\bf v} \tilde{\bf G} = {\bf 0} | w_H({\bf v}) = l \right\} = \tilde{\pi}_l^n$ into (\ref{eq_P_ina_1}), we can obtain an upper bound on $\mathbb{E}\left[P_{\textsf{E}}^{\rm ML}(\tilde{\cal T})\right]$. Notably, when all batches do not involve the non-zero elements of the codeword ${\bf v}$, the error probability of ML decoding is independent of the chosen codeword. Based on this fact, we can further tighten this upper bound, thereby proving Theorem~\ref{theorem_LDPC_BATS_INA}. (We observe that this step primarily has an impact on the numerical results of non-binary BATS codes with a relatively high-rate precode.)

Let $\mathscr{C}^{(1)},\ldots,\mathscr{C}^{(\Xi)}$ be all codebooks in the precode ensemble, where $\Xi$ is the number of graphs in the precode ensemble, and each codebook has the same probability to be chosen. In this paper, $\Xi = 1$ if the precode is fixed (applying no precode can be regarded as applying a rate-$1$ fixed precode), and $\Xi = (K\mathtt{l})!(q-1)^{K\mathtt{l}}$ if the precode is randomly chosen from a regular LDPC ensemble. Define $\mathscr{C}_l^{(k)} = \left\{ {\bf v} \in \mathscr{C}^{(k)}: w_H({\bf v}) = l \right\}$ and $\mathscr{C}_l^{(k)}(j_1,j_2,\ldots,j_l) = \left\{ {\bf v} \in \mathscr{C}_l^{(k)}: v_{j_i} \ne 0, \forall i \in [1:l] \right\}$. It is clear that 
\begin{equation}\nonumber
	\mathscr{C}_l^{(k)} = \bigcup_{1 \le j_1 < j_2 < \cdots < j_l \le K} \mathscr{C}_l^{(k)}(j_1,j_2,\ldots,j_l).
\end{equation}
For a codeword ${\bf v} \in \mathscr{C}_l^{(k)}(j_1,j_2,\ldots,j_l)$, we use $I_n(j_1,j_2,\ldots,j_l)$ to represent the random variable that the total number of non-zero symbols of ${\bf v}$ that are chosen be encoded by $n$ batches. For example, $I_n(j_1,j_2,\ldots,j_l) = 0$ indicates that no non-zero symbols are encoded by $n$ batches, and so ${\bf v} \tilde{\bf G} = {\bf 0}$ holds for all ${\bf v} \in \mathscr{C}_l^{(k)}(j_1,j_2,\ldots,j_l)$, $\forall k$. Note that $\Pr\{I_n(j_1,j_2,\ldots,j_l) = i\}$ is independent of the specific choice of $j_1,j_2,\ldots,j_l$.

Now, following the approach used to establish the bound in (\ref{eq_P_ina_1}), we derive a new, tighter bound, as shown in (\ref{eq_P_ina2}).
\begin{figure*}[!t]
	\normalsize
\begin{equation}\label{eq_P_ina2}
	\begin{aligned}
		&~~~\mathbb{E}\left[P_{\textsf{E}}^{\rm ML}(\tilde{\cal T})\right] \\
		& = \frac{1}{\Xi} \sum_{k=1}^{\Xi} \Pr\left\{\bigcup_{{\bf v}\in \mathscr{C}^{(k)} \setminus \{{\bf 0}\}} \left({\bf v} \tilde{\bf G} = {\bf 0}\right) \right\}\\
		& \le \frac{1}{\Xi} \sum_{k=1}^{\Xi} \sum_{l=1}^{K} \sum_{1 \le j_1 < j_2 < \cdots < j_l \le K} \Pr\left\{ \bigcup_{{\bf v}\in \mathscr{C}_l^{(k)}(j_1,\ldots,j_l)} \left({\bf v} \tilde{\bf G} = {\bf 0}\right)\right\}\\
		& = \frac{1}{\Xi} \sum_{k,l,j_1,\ldots,j_l} \left( \Pr\left\{ \bigcup_{\substack{{\bf v}\in\\ \mathscr{C}_l^{(k)}(j_1,\ldots,j_l)}} \left({\bf v} \tilde{\bf G} = {\bf 0}\right), I_n(j_1,\ldots,j_l) = 0\right\} + \Pr\left\{ \bigcup_{\substack{{\bf v}\in\\ \mathscr{C}_l^{(k)}(j_1,\ldots,j_l)}} \left({\bf v} \tilde{\bf G} = {\bf 0}\right), I_n(j_1,\ldots,j_l)  > 0\right\} \right) \\
		& = \frac{1}{\Xi} \sum_{k,l,j_1,\ldots,j_l} \left(\mathds{1}_{\left\{\mathscr{C}_l^{(k)}(j_1,\ldots,j_l) \ne \emptyset\right\}} \Pr\left\{ I_n(j_1,\ldots,j_l) = 0\right\} + \Pr\left\{ \bigcup_{{\bf v}\in \mathscr{C}_l^{(k)}(j_1,\ldots,j_l)} \left({\bf v} \tilde{\bf G} = {\bf 0}\right), I_n(j_1,\ldots,j_l)  > 0\right\}\right)\\
		& \overset{(a)}{\le} \frac{1}{\Xi} \sum_{k,l,j_1,\ldots,j_l} \biggg(\mathds{1}_{\left\{\mathscr{C}_l^{(k)}(j_1,\ldots,j_l) \ne \emptyset\right\}}\Pr\left\{ I_n(j_1,\ldots,j_l) = 0\right\} \\
		&~~~+ \frac{\left|\mathscr{C}_l^{(k)}(j_1,\ldots,j_l)\right|}{q-1} \Pr\left\{ \left({\bf v} \tilde{\bf G} = {\bf 0}\right), I_n(j_1,\ldots,j_l)  > 0 \mid w_H({\bf v}) = l\right\}\biggg)\\
		& \overset{(b)}{\le} \sum_{l=1}^{K}  \underbrace{\min\left\{\frac{A_l}{q-1},\binom{K}{l}\right\} \Pr\left\{ I_n(j_1,\ldots,j_l) = 0\right\}}_{(a)} + \underbrace{\frac{A_l}{q-1} \left(\Pr\left\{ {\bf v} \tilde{\bf G} = {\bf 0}\mid w_H({\bf v}) = l\right\} - \Pr\left\{ I_n(j_1,\ldots,j_l) = 0 \right\}\right)}_{(b)}.
	\end{aligned}
\end{equation}
	\hrulefill
\end{figure*}
In (\ref{eq_P_ina2}), we use $\sum_{k,l,j_1,\ldots,j_l}$ to represent $\sum_{k=1}^{\Xi} \sum_{l=1}^{K} \sum_{1 \le j_1 < j_2 < \cdots < j_l \le K}$ to save notation. The factor $\frac{1}{q-1}$ in inequality $(a)$ follows from the same reason as that in (\ref{eq_P_ina_1}), and inequality $(b)$ deserves further discussion. First, we have for a fixed $l$
\begin{equation}\nonumber
	\frac{1}{\Xi} \sum_{k=1}^{\Xi} \sum_{1 \le j_1 < j_2 < \cdots < j_l \le K} \left|\mathscr{C}_l^{(k)}(j_1,\ldots,j_l)\right| = A_l,
\end{equation}
which is the excepted weight enumerator. By noting that $\Pr\left\{ \left({\bf v} \tilde{\bf G} = {\bf 0}\right), I_n(j_1,\ldots,j_l)  > 0 \mid w_H({\bf v}) = l\right\}$ is independent of specific $j_1,\ldots,j_l$, we can write
\begin{equation}\nonumber
	\begin{aligned}
		&\frac{1}{\Xi} \sum_{k=1}^{\Xi} \sum_{l=1}^{K} \sum_{1 \le j_1 < j_2 < \cdots < j_l \le K} \frac{\left|\mathscr{C}_l^{(k)}(j_1,\ldots,j_l)\right|}{q-1}\\
		&\cdot \Pr\left\{ \left({\bf v} \tilde{\bf G} = {\bf 0}\right), I_n(j_1,\ldots,j_l)  > 0 \mid w_H({\bf v}) = l\right\}\\
		& = \sum_{l=1}^{K} \frac{A_l}{q-1} \Pr\left\{ \left({\bf v} \tilde{\bf G} = {\bf 0}\right), I_n(j_1,\ldots,j_l)  > 0 \mid w_H({\bf v}) = l\right\}\\
		& = \sum_{l=1}^{K} \frac{A_l}{q-1} \Bigg(\Pr\left\{ \left({\bf v} \tilde{\bf G} = {\bf 0}\right) \mid w_H({\bf v}) = l\right\} \\
		&~~~- \Pr\left\{ \left({\bf v} \tilde{\bf G} = {\bf 0}\right), I_n(j_1,\ldots,j_l) = 0 \mid w_H({\bf v}) = l\right\}\Bigg)\\
		& = \sum_{l=1}^{K} \frac{A_l}{q-1} \Bigg(\Pr\left\{ \left({\bf v} \tilde{\bf G} = {\bf 0}\right) \mid w_H({\bf v}) = l\right\} \\
		&~~~- \Pr\left\{ I_n(j_1,\ldots,j_l) = 0\right\}\Bigg),
	\end{aligned}
\end{equation}
which leads to the term $(b)$ in inequality $(b)$. The term $(a)$ in inequality $(b)$ follows from that for a fixed $l$
\begin{equation}\nonumber
	\begin{aligned}
		&\frac{1}{\Xi} \sum_{k=1}^{\Xi} \sum_{1 \le j_1 < j_2 < \cdots < j_l \le K} \mathds{1}_{\left\{\mathscr{C}_l^{(k)}(j_1,\ldots,j_l) \ne \emptyset\right\}} \\
		& = \frac{1}{\Xi} \sum_{k=1}^{\Xi} \sum_{1 \le j_1 < j_2 < \cdots < j_l \le K} \min \left\{ \frac{\left|\mathscr{C}_l^{(k)}(j_1,\ldots,j_l)\right|}{q-1}, 1\right\} \\
		& \le \min\left\{\frac{A_l}{q-1},\binom{K}{l}\right\}.
	\end{aligned}
\end{equation}
Note that $\frac{\left|\mathscr{C}_l^{(k)}(j_1,\ldots,j_l)\right|}{q-1}$ must be an integer for a $q$-ary code. Substituting 
\begin{equation}\nonumber
	\Pr\left\{ I_n(j_1,\ldots,j_l) = 0\right\} = \left(\sum_{d=1}^{K} \varphi_{d,K-l}\right)^n,
\end{equation}
and
\begin{equation}\nonumber
	\Pr\left\{ {\bf v} \tilde{\bf G} = {\bf 0}\mid w_H({\bf v}) = l\right\} = \tilde{\pi}_l^n
\end{equation}
into (\ref{eq_P_ina2}), the proof of Theorem \ref{theorem_LDPC_BATS_INA} is completed.

\section{Proofs of Several Properties of Stopping Graphs}\label{appendix:proofs_for_BP}

\begin{proof}[Proof of Lemma \ref{Lemma_Combine_BATS_SS}]
For constraint ${\rm C1}$: For $V \in {\cal A}_i$, $i = 1,2$, $|{\rm Neib}_{\cal E}(V) \cap ({\cal C} \backslash ({\cal B}_1 \cup {\cal B}_2))| \le  |{\rm Neib}_{\cal E}(V) \cap ({\cal C} \backslash {\cal B}_i)| = 0$; and for $V \in {\cal A}_e$, constraint ${\rm C1}$ is satisfied by the definition of the set ${\cal A}_e$. Thus, constraint ${\rm C1}$ holds for all VNs in ${{\cal A}_1 \cup {\cal A}_2} \cup {\cal A}_e$.

For constraint ${\rm C2}$: There is a fact that for any B-CN $C$, inequality
\begin{multline}\nonumber
	|{\rm Neib}(C) \cap ({\cal A}_1 \cup {\cal A}_2)| - |{\rm Neib}(C) \cap {\cal A}_i| \\ \ge
	{\rm rk}(C^{({\rm Neib}(C) \cap ({\cal A}_1 \cup {\cal A}_2))}) - {\rm rk}(C^{({\rm Neib}(C) \cap {\cal A}_i)})
\end{multline}
holds for $i=1,2$. Therefore, we have for $C \in ({\cal B}_1 \cup {\cal B}_2)$ and $i = 1,2$
\begin{equation}\nonumber
	\begin{aligned}
		&|{\rm Neib}(C) \cap ({\cal A}_1 \cup {\cal A}_2)| \\
		&\ge {\rm rk}(C^{({\rm Neib}(C) \cap ({\cal A}_1 \cup {\cal A}_2))}) + |{\rm Neib}(C) \cap {\cal A}_i| \\
		&~~~- {\rm rk}(C^{({\rm Neib}(C) \cap {\cal A}_i)})\\
		& > {\rm rk}(C^{({\rm Neib}(C) \cap ({\cal A}_1 \cup {\cal A}_2))}).
	\end{aligned}
\end{equation}
We then can use the same argument to prove that
\begin{equation}\nonumber
	\begin{aligned}
		&|{\rm Neib}(C) \cap ({\cal A}_1 \cup {\cal A}_2 \cup {\cal A}_e)| \\
		&\ge {\rm rk}(C^{({\rm Neib}(C) \cap ({\cal A}_1 \cup {\cal A}_2 \cup {\cal A}_e))}) + |{\rm Neib}(C) \cap ({\cal A}_1 \cup {\cal A}_2)|\\
		&~~~- {\rm rk}(C^{({\rm Neib}(C) \cap ({\cal A}_1 \cup {\cal A}_2))}) \\
		&> {\rm rk}(C^{({\rm Neib}(C) \cap ({\cal A}_1 \cup {\cal A}_2 \cup {\cal A}_e))}),
	\end{aligned}
\end{equation}
which validates that constraint ${\rm C2}$ is satisfied.

For constraint ${\rm C3}$: If a VN $V$ satisfies ${\rm Neib}_{\cal E}(V) \cap ({\cal C} \backslash ({\cal B}_1 \cup {\cal B}_2)) = \emptyset$, we know that $V \in ({{\cal A}_1 \cup {\cal A}_2} \cup {\cal A}_e)$ according to the definition of ${\cal A}_e$. Thus, $V \in ({\cal V} \setminus ({{\cal A}_1 \cup {\cal A}_2} \cup {\cal A}_e))$ is a sufficient condition such that ${\rm Neib}_{\cal E}(V) \cap ({\cal C} \backslash ({\cal B}_1 \cup {\cal B}_2)) \ne \emptyset$.

According to the above arguments, $\left({{\cal A}_1 \cup {\cal A}_2} \cup {\cal A}_e,{{\cal B}_1 \cup {\cal B}_2}\right)$ is valid for $\cal T$. The proof is completed.
\end{proof}

\begin{proof}[Proof of Lemma \ref{Lemma_BATS_Stops_at_Stopping_Graph}]
We first prove the necessity. Assuming that ${\cal T}' = ({\cal V}',{\cal C}',{\cal E}')$ is a non-empty residual graph after BP decoding, in which no B-CN is decodable. Consider the reconstruction of ${\cal T}$ from ${\cal T}'$ by adding $|{\cal V}|-|{\cal V}'|$ VNs. Let ${\cal A} = {\cal V}^{(0)} = {\cal V}'$, ${\cal B} = {\cal C}^{(0)} = {\cal C}'$, ${\cal E}^{(0)} = {\cal E}'$, and let ${\cal T}^{(0)} = ({\cal V}^{(0)},{\cal C}^{(0)},{\cal E}^{(0)})$ be the initial Tanner graph. ${\cal T}^{(\ell)}$ can be reconstructed from ${\cal T}^{(\ell-1)}$ by adding one VN, denoted by $V^+$, the edges emitted by $V^+$ and the B-CNs ${\rm Neib}(V^+) \backslash {\cal C}^{(\ell-1)}$, after which ${\cal A}^{(\ell)}$, ${\cal B}^{(\ell)}$, ${\cal V}^{(\ell)}$, ${\cal C}^{(\ell)}$, and ${\cal E}^{(\ell)}$ will be updated. We have ${\cal T}' = {\cal T}^{(0)}$ and ${\cal T} = {\cal T}^{(|{\cal V}|-|{\cal V}'|)}$.

The necessity can be proved by induction. Assume that $({\cal A},{\cal B})$ is valid for ${\cal T}^{(\ell-1)}$. Since the reconstruction is the reverse process of BP decoding of a fixed $\cal T$, the shape of ${\cal T}^{(\ell)}$, i.e., ${\cal V}^{(\ell)}$, ${\cal C}^{(\ell)}$, and ${\cal E}^{(\ell)}$, is deterministic. Then, we will demonstrate that constraints $({\cal A},{\cal B})$ is still valid for ${\cal T}^{(\ell)}$ after adding one VN. There is a fact that the newly added edges are only connected to $V^+$, which indicates that during the reconstruction process neither extra edges will be connected to $\cal A$ nor will existing edges emanating from $\cal A$ be removed.
Thus, constraint ${\rm C1}$ holds for ${\cal T}^{(\ell)}$. Also due to this fact, we can see that ${\rm Neib}_{{\cal E}^{(\ell)}}(C) \cap {\cal A}$ is fixed and independent of $\ell$ for any $C \in {\cal B}$, thus ensuring that constraint ${\rm C2}$ holds for ${\cal T}^{(\ell)}$. Finally for constraint ${\rm C3}$, we only need to consider $V^+$ because the edges connected to ${\cal V}^{(\ell-1)}$ do not change. Constraint ${\rm C3}$ holds for the fact that $V^+$ is decodable for ${\cal T}^{(\ell)}$. It is clear that $V^+$ cannot be decoded by any B-CN in ${\cal B}$ because of constraint ${\rm C2}$. Therefore, we can argue that ${\rm Neib}_{{\cal E}^{(\ell)}}(V^+) \cap ({\cal C}^{(\ell)} \backslash {\cal B}) \ne \emptyset$. The last thing is to verify whether $({\cal A},{\cal B})$ is valid for ${\cal T}^{(0)}$. Constraints ${\rm C1}$ and ${\rm C3}$ respectively follow from the facts that ${\cal C}^{(0)} \backslash {\cal B} = \emptyset$ and ${\cal V}^{(0)} \backslash {\cal A} = \emptyset$, and constraint ${\rm C2}$ holds by assumption that no decodable B-CN is in the Tanner graph. The necessity is proved.

We now prove the sufficiency. If $({\cal A},{\cal B})$ is valid for $\cal T$, the nodes in ${\cal A}$ and ${\cal B}$ cannot be removed by BP decoding even if all other nodes have been removed, due to constraints ${\rm C1}$ and ${\rm C2}$. The proof is completed.
\end{proof}

\begin{proof}[Proof of Lemma \ref{Lemma_LDPCBATS_Stops_at_Stopping_Graph}]
	For $\mathtt{l} = 0$, this lemma is identical to Lemma \ref{Lemma_BATS_Stops_at_Stopping_Graph}. Therefore, we consider $\mathtt{l} > 0$ in the following. Assume that $\tilde{\cal T}' = ({\cal V}',{\cal C}',{\cal E}',\tilde{\cal C}',\tilde{\cal E}')$ is the residual graph after BP decoding, in which no B-CN or L-CN is decodable. Consider the reconstruction of $\tilde{\cal T}$ from $\tilde{\cal T}'$ by adding $|{\cal V}|-|{\cal V}'|$ VNs. Let ${\cal A} = {\cal V}^{(0)} = {\cal V}'$, ${\cal B} = {\cal C}^{(0)} = {\cal C}'$, ${\cal E}^{(0)} = {\cal E}'$, $\tilde{\cal C}^{(0)} = \tilde{\cal C}'$, and $\tilde{\cal E}^{(0)} = \tilde{\cal E}'$. Form the initial Tanner graph $\tilde{\cal T}^{(0)} = ({\cal V}^{(0)},{\cal C}^{(0)},{\cal E}^{(0)},\tilde{\cal C}^{(0)},\tilde{\cal E}^{(0)})$. $\tilde{\cal T}^{(\ell)}$ is reconstruct from $\tilde{\cal T}^{(\ell-1)}$ by adding one VN, denoted by $V^+$, the edges emitted by $V^+$, and its neighbors ${\rm Neib}(V^+) \backslash ({\cal C}^{(\ell-1)} \cup \tilde{\cal C}^{(\ell-1)})$, after which ${\cal V}^{(\ell)}$, ${\cal C}^{(\ell)}$, ${\cal E}^{(\ell)}$, $\tilde{\cal C}^{(\ell)}$, and $\tilde{\cal E}^{(\ell)}$ will be updated. We have $\tilde{\cal T}' = \tilde{\cal T}^{(0)}$ and $\tilde{\cal T} = \tilde{\cal T}^{(|{\cal V}|-|{\cal V}'|)}$.
	
	The necessity can be proved by induction. Assume that $({\cal A},{\cal B})$ is valid for $\tilde{\cal T}^{(\ell-1)}$. Due to the fact that adding a VN and its associated edges and neighbors is the inverse process of the BP decoding of $\tilde{\cal T}$, the shape of $\tilde{\cal T}^{(\ell)}$ is determined. There is a fact that the edges emitted from ${\cal A}$ do not change because the newly added edges are only connected to $V^+$. Thus, constraint ${\rm C1}$ holds for $\tilde{\cal T}^{(\ell)}$. Similarly, the set ${\rm Neib}(C) \cap {\cal A}$, $\forall C \in {\cal B}$, is fixed from $\ell-1$ to $\ell$, which indicates that constraint ${\rm C2}$ holds for $\tilde{\cal T}^{(\ell)}$. Only the validity of constraint ${\rm C4}$ for $\tilde{\cal T}^{(\ell)}$ merits a more detailed discussion. 
	
	Recalling that $\cal A$ is the remaining VN set after BP decoding, so $\cal A$ is a stopping set of the LDPC code. At the beginning, i.e., for $\tilde{\cal T}^{(0)}$, we have $\cal X = \emptyset$. From step $\ell-1$ to $\ell$, $V^+$ will be added to $\cal X$ if and only if ${\rm Neib}_{{\cal E}^{(\ell)}}(V^+) \cap ({\cal C}^{(\ell)} \backslash {\cal B}) = \emptyset$. Then there are two cases: 1) ${\rm Neib}_{{\cal E}^{(\ell)}}(V^+) \cap {\cal B} \ne \emptyset$; or 2) ${\rm Neib}_{{\cal E}^{(\ell)}}(V^+) \cap {\cal B} = \emptyset$. For case 1), let $c$ be a randomly picked neighbor of $V^+$ in ${\cal B}$. According to constraint ${\rm C2}$, we have $|{\rm Neib}(c)| = |{\rm Neib}(c) \cap {\cal A}| + |{\rm Neib}(c) \cap ({\cal V}^{(\ell)}\backslash{\cal A})| > {\rm rk}(c^{({\rm Neib}(c) \cap {\cal A})}) + {\rm rk}(c^{({\rm Neib}(c) \cap ({\cal V}^{(\ell)}\backslash{\cal A}))}) \ge {\rm rk}(c)$, which implies that $V^+$ cannot be decoded by any B-CN in ${\cal B}$. For case 2), $V^+$ is not connected to any B-CN. Therefore, we know that if $V^+$ is added to $\cal X$, there exists at least one L-CNs of degree one in $\tilde{\cal C}^{(\ell)} \backslash \tilde{\cal C}^{(\ell-1)}$, which are the neighbors of $V^+$. It follows that $({\cal A},{\cal X})$ forms an extended stopping set with the maximal stopping set ${\cal A}$. 
	
	It remains to show that $|{\cal X}| \le s_{\rm max}(a)$. Assume that we are given $|{\cal X}| = s$ for $\tilde{\cal T}^{(|{\cal V}|-|{\cal V}'|)}$ at the end of the reconstruction. Observing that $\cal A$ cannot be connected to $s$ \textit{peeling} L-CNs that are reserved for $\cal X$, we need to guarantee that the inequality $a\mathtt{l} \le (K\mathtt{l}/\mathtt{r}-s)\mathtt{r}$ holds so that the L-CN side remains enough sockets for ${\cal A}$. Since $s$ is an integer, we have $s \le \left\lfloor \frac{(K-a)\mathtt{l}}{\mathtt{r}} \right\rfloor \triangleq s_{\rm max}(a)$. 
	
	Finally we verify whether $({\cal A},{\cal B})$ is valid for $\tilde{\cal T}^{(0)}$. Constraint ${\rm C1}$ holds for ${\cal C}^{(0)} \backslash {\cal B} = \emptyset$. Constraint ${\rm C2}$ holds by assumption that no decodable B-CN remains in $\tilde{\cal T}^{(0)}$. To verify constraint ${\rm C4}$, first, we have ${\cal X} = \emptyset$ for $\tilde{\cal T}^{(0)}$. Then, $({\cal A},\emptyset)$ is also an extended stopping set. The necessity is proved.
	
	We now prove the sufficiency. If $({\cal A},{\cal B})$ is valid for $\tilde{\cal T}$, the VNs in ${\cal A}$ cannot be decoded by BP decoding even if all other VNs and associated edges are removed, due to constraints ${\rm C1}$, ${\rm C2}$, and the fact that ${\cal A}$ is a stopping set. The proof is completed.
\end{proof}

\section{Proof of Lemma \ref{Lemma_Union_Card}}\label{Appendix_Proof_Lemma_Union_Card}
We provide two methods to prove Lemma \ref{Lemma_Union_Card}. 

\noindent{\textit{Method 1:}} 
The following lemma is used to prove Lemma \ref{Lemma_Union_Card} by the first method.

\begin{lemma}\label{Lemma_Diagonalizable}
	For a Markov process with the $(m+1) \times (m+1)$ probability transition matrix ${\bf T}$ with
	\begin{equation}\nonumber
		{\bf T}[i,j] = \sum_{d = j-i}^{\min\{D,j\}} \beta_d {\rm hyge}(i+d-j;m,i,d),
	\end{equation}
	where ${\bm \beta} = (\beta_0,\beta_1,\ldots,\beta_D)$ is a probability distribution. Matrix ${\bf T}$ is diagonalizable, i.e., 
	\begin{equation}\nonumber
		{\bf T} = {\bf U}{\bf D}{\bf U}^{-1},
	\end{equation}
	where ${\bf D}$ is a diagonal matrix with ${\bf D}[i,i] = {\bf T}[i,i]$, ${\bf U}$ is an upper-triangular matrix with ${\bf U}[i,j] = \binom{m-i}{j-i}$ for $i \le j$, and ${\bf U}^{-1}$ is an upper-triangular matrix with ${\bf U}^{-1}[i,j] = (-1)^{j-i} \binom{m-i}{j-i}$ for $i \le j$.
\end{lemma}
\begin{proof}
Let ${\bf U}'[i,j] = (-1)^{j-i} \binom{m-i}{j-i}$ for $i \le j$.
\begin{equation}\nonumber
	\begin{aligned}
		({\bf U}{\bf U}')[i,j] & = \sum_{k = 0}^{m} {\bf U}[i,k] {\bf U}'[k,j]\\
		& = \sum_{k = i}^{j} {\bf U}[i,k] {\bf U}'[k,j]\\
		& = \sum_{k = i}^{j} \binom{m-i}{k-i} (-1)^{j-k} \binom{m-k}{j-k}\\
		& = \binom{m-i}{j-i} \sum_{k = i}^{j} (-1)^{j-k} \binom{j-i}{k-i}\\
		& = \binom{m-i}{j-i} \sum_{k' = 0}^{j-i} (-1)^{j-i-k'} \binom{j-i}{k'}\\
		& = \binom{m-i}{j-i} \sum_{k' = 0}^{j-i} (-1)^{j-i-k'} \binom{j-i}{j-i-k'}.
	\end{aligned}		
\end{equation}
One can verify that $({\bf U}{\bf U}')[i,i] = 1$ and $({\bf U}{\bf U}')[i,j] = 0$ for $j > i$. Therefore, we obtain ${\bf U}^{-1} = {\bf U}'$.

The binomial coefficient with negative integers will be used in the following proof, denoted by $\binom{n}{k}^*$, which is defined as
\begin{equation}\nonumber
	\binom{n}{k}^* = \begin{cases}
		\frac{n(n-1)\cdots (n-k+1)}{k!} &\mbox{if } \left\{\begin{aligned}&(k > 0 \text{ and } n < 0) \text{ or}\\ &(0 < k \le n \text{ and } n \ge 0),\end{aligned}\right. \\
		1 &\mbox{if } k=0,\\
		0 &\mbox{otherwise.}
	\end{cases}
\end{equation}

We now verify that $({\bf U}{\bf D}{\bf U}^{-1})[i,j] = {\bf T}[i,j]$ for all $i$ and $j$. Each element in ${\bf U}{\bf D}{\bf U}^{-1}$ can be computed by
\begin{equation}\nonumber
	\begin{aligned}
		&({\bf U}{\bf D}{\bf U}^{-1})[i,j]\\
		&= \sum_{k=i}^{j} \binom{m-i}{k-i} {\bf T}[k,k] (-1)^{j-k} \binom{m-k}{j-k}\\
		&= \binom{m-i}{j-i} \sum_{k=i}^{j} (-1)^{j-k} \binom{j-i}{k-i} \sum_{d = 0}^{\min\{D,k\}} \beta_d \frac{\binom{k}{d}}{\binom{m}{d}}\\
		&= \sum_{d = 0}^{\min\{D,j\}} \frac{\beta_d}{\binom{m}{d}} \binom{m-i}{j-i} \sum_{k=\max\{d,i\}}^{j} (-1)^{j-k} \binom{j-i}{k-i} \binom{k}{d}\\
		&=\sum_{d = 0}^{\min\{D,j\}} \frac{\beta_d}{\binom{m}{d}} \binom{m-i}{j-i} \\
		&~~~\cdot \sum_{k=\max\{d,i\}}^{j} (-1)^{j-k} \binom{j-i}{k-i} (-1)^{k-d} \binom{-d-1}{k-d}^*\\
		&=\sum_{d = 0}^{\min\{D,j\}} \frac{\beta_d}{\binom{m}{d}} \binom{m-i}{j-i} (-1)^{j-d} \\
		&~~~\cdot \sum_{k=\max\{d,i\}}^{j} \binom{j-i}{k-i} \binom{-d-1}{k-d}^*.
	\end{aligned}
\end{equation}
The last summation term in the above equation can be simplified by considering $d \ge i$ and $d < i$. When $d \ge i$, we have
\begin{equation}\nonumber
	\begin{aligned}
		&\sum_{k=\max\{d,i\}}^{j} \binom{j-i}{k-i} \binom{-d-1}{k-d}^*\\
		&= \sum_{k=d}^{j} \binom{j-i}{j-k} \binom{-d-1}{k-d}^*\\
		&= \sum_{k'=0}^{j-d} \binom{j-i}{j-k'-d} \binom{-d-1}{k'}^*\\
		&= \binom{j-i-d-1}{j-d}^*,
	\end{aligned}		
\end{equation}
where the last equality follows from Chu–Vandermonde identity. Similarly, when $d < i$, we have
\begin{equation}\nonumber
	\begin{aligned}
		&\sum_{k=\max\{d,i\}}^{j} \binom{j-i}{k-i} \binom{-d-1}{k-d}^*\\
		&= \sum_{k=i}^{j} \binom{j-i}{j-k} \binom{-d-1}{k-d}^*\\
		&= \sum_{k=d}^{i-1} \binom{j-i}{j-k} \binom{-d-1}{k-d}^* + \sum_{k=i}^{j} \binom{j-i}{j-k} \binom{-d-1}{k-d}^*\\
		&= \sum_{k=d}^{j} \binom{j-i}{j-k} \binom{-d-1}{k-d}^*\\
		&= \binom{j-i-d-1}{j-d}^*,
	\end{aligned}		
\end{equation}
where the second equality is obtained by $\binom{j-i}{j-k} =0$ for $k < i$. Finally, we can check that
\begin{equation}\nonumber
	\begin{aligned}
		&({\bf U}{\bf D}{\bf U}^{-1})[i,j] \\
		& = \sum_{d = 0}^{\min\{D,j\}} \frac{\beta_d}{\binom{m}{d}} \binom{m-i}{j-i} (-1)^{j-d} \binom{j-i-d-1}{j-d}^*\\
		& = \sum_{d = 0}^{\min\{D,j\}} \frac{\beta_d}{\binom{m}{d}} \binom{m-i}{j-i}  \binom{i}{j-d}\\
		& = \sum_{d = 0}^{\min\{D,j\}} \beta_d \frac{\binom{m-i}{j-i}  \binom{i}{i+d-j}}{\binom{m}{d}}\\
		& = \sum_{d = j-i}^{\min\{D,j\}} \beta_d {\rm hyge}(i+d-j;m,i,d)\\
		& = {\bf T}[i,j].
	\end{aligned}
\end{equation}
The proof of Lemma \ref{Lemma_Diagonalizable} is completed.
\end{proof}

Using Lemma \ref{Lemma_Diagonalizable}, ${\bf T}^s = {\bf U}{\bf D}^s{\bf U}^{-1}$. We can write
\begin{equation}\label{eq_T_power}
	\begin{aligned}
		({\bf T}^s)[i,j] &= \sum_{k = i}^{j} {\bf U}[i,k]({\bf T}[k,k])^s{\bf U}^{-1}[k,j]\\
		&=\binom{m-i}{j-i} \sum_{k=i}^{j} (-1)^{j-k} \binom{j-i}{k-i} \left(\sum_{d = 0}^{k} \beta_d \frac{\binom{k}{d}}{\binom{m}{d}}\right)^s.
	\end{aligned}
\end{equation}
The proof is completed by noting
\begin{equation}\nonumber
	\Pr\{|S| = j\} = {\bf T}^s[0,j].
\end{equation}

\begin{remark}\label{remark_union_card}
	In computer simulations, it may be better to compute ${\bf T}^s$ rather than directly using (\ref{eq_T_power}). All values involved in computing ${\bf T}^s$ are probabilities, thus avoiding issues of insufficient computational precision caused by large binomial coefficients in (\ref{eq_T_power}) (although these issues may be addressed through symbolic computation).
\end{remark}

\noindent{\textit{Method 2:}} 
Define a set $U_j$ of $j$ elements in $U$. Let us consider $U_j = \{1,2,\ldots,j\}$. Recalling that $S = \bigcup_{i=1}^{s} S_i$, the event $S = U_j$ can happen only if each $S_i$, $i \in [1:s]$, is a subset of $U_j$. We can write 
\begin{equation}\label{eq_lemma_union_card_1}
	\Pr\left\{ S_i \subset U_j, \forall i \in [1:s] \right\} = \left(\sum_{d=0}^{m} \beta_d \frac{\binom{j}{d}}{\binom{m}{d}}\right)^s.
\end{equation}
There are $j$ distinct subsets $U_{j-1} \subset U_j$ of $j-1$ elements, denoted by $U_{j-1,1},\ldots,U_{j-1,j}$. Let $E_i$ denote the event $S \subset U_{j-1,i}$. To ensure $S = U_j$, we need to exclude all events $E_i$, $i \in [1:j]$, from (\ref{eq_lemma_union_card_1}). Considering randomly picking $k$ subsets $U_{j-1,i_1},\ldots,U_{j-1,i_k}$, we have $ |U_{j-1,i_1} \cap \cdots \cap U_{j-1,i_k}| = j-k$, which is independent of the selection of $k$ subsets. This can be verified by
\begin{equation}\nonumber
	\begin{aligned}
		&|U_{j-1,i_1} \cap \cdots \cap U_{j-1,i_k}|\\
		 & = |U_j \setminus \left(U_j \setminus \left(U_{j-1,i_1} \cap \cdots \cap U_{j-1,i_k}\right)\right)|\\
		& = |U_j \setminus \left(\left(U_j \setminus U_{j-1,i_1}\right) \cup \cdots \cup \left(U_j \setminus U_{j-1,i_k}\right)  \right)|\\
		& = j-k.
	\end{aligned}
\end{equation}
Using the fact $|U_{j-1,i_1} \cap \cdots \cap U_{j-1,i_k}| = j-k$, we can write
\begin{equation}\nonumber
	\begin{aligned}
		&\Pr\left\{ E_1 \cup \cdots \cup E_j \right\}\\
		 & \overset{(a)}{=} \sum_{k = 1}^{j}(-1)^{k+1} \left( \sum_{1 \le i_1 < \cdots < i_k \le j} \Pr\left\{ E_{i_1} \cap \cdots \cap E_{i_k} \right\}\right)\\
		& \overset{(b)}{=} \sum_{k = 1}^{j}(-1)^{k+1} \binom{j}{k} \left(\sum_{d=0}^{m} \beta_d \frac{\binom{j-k}{d}}{\binom{m}{d}}\right)^s,
	\end{aligned}
\end{equation}
where equality $(a)$ is due to the inclusion–exclusion identity, and $(b)$ to the fact that there are $\binom{j}{k}$ ways to choose $i_1,\ldots,i_k$ and each way has the same probability 
\begin{equation}\nonumber
	\begin{aligned}
		\Pr\left\{ E_{i_1} \cap \cdots \cap E_{i_k} \right\} & = \Pr\left\{ S \subset \left(U_{j-1,i_1} \cap \cdots \cap U_{j-1,i_k}\right) \right\}\\
		& = \Pr\left\{ S_i \subset U_{j-k}, \forall i \in [1:s] \right\} \\
		& = \left(\sum_{d=0}^{m} \beta_d \frac{\binom{j-k}{d}}{\binom{m}{d}}\right)^s.
	\end{aligned}
\end{equation}

Noting that there are $\binom{m}{j}$ ways to pick $U_j$, we have
\begin{equation}\nonumber
	\begin{aligned}
		&\Pr\{|S| = j\}\\
		& = \binom{m}{j} \left(\Pr\left\{ S_i \subset U_j, \forall i \in [1:s] \right\} - \Pr\left\{ E_1 \cup \cdots \cup E_j \right\}\right)\\
		&= \binom{m}{j} \left(\left(\sum_{d=0}^{m} \beta_d \frac{\binom{j}{d}}{\binom{m}{d}}\right)^s + \sum_{k = 1}^{j}(-1)^{k} \binom{j}{k} \left(\sum_{d=0}^{m} \beta_d \frac{\binom{j-k}{d}}{\binom{m}{d}}\right)^s\right)\\
		& = \binom{m}{j} \left(\sum_{k = 0}^{j}(-1)^{k} \binom{j}{k} \left(\sum_{d=0}^{j-k} \beta_d \frac{\binom{j-k}{d}}{\binom{m}{d}}\right)^s\right)\\
		& = \binom{m}{j} \left(\sum_{k' = 0}^{j}(-1)^{j-k'} \binom{j}{k'} \left(\sum_{d=0}^{k'} \beta_d \frac{\binom{k'}{d}}{\binom{m}{d}}\right)^s\right),
	\end{aligned}
\end{equation}
where the last equality is obtained by substituting $k' = j-k$. The proof is completed.

\section{Proofs of Two Upper Bounds for Extended Stopping Sets}\label{appendix:ess}
\begin{proof}[Proof of Lemma~\ref{lemma:ess_recursive}]
We first briefly review the recursion in \cite[Ch. 3.22]{Modern_Coding_Theory}, which was proposed for computing the finite-length performance of the LDPC ensemble. Let $A(v,t,s)$ be the number of the constellations containing \textit{non-empty} stopping sets, where each constellation is formed by connecting $v$ fixed VNs to $t$ fixed L-CNs of degree at least 2 and $s$ fixed L-CNs of degree 1. Here, $v \in [1:K]$ and $t,s \in [0:K\mathtt{l}/\mathtt{r}]$. Note that for $s = 0$, $A(v,t,s=0)$ is the number of $(v,t)$ stopping sets with $v$ fixed VNs and $t$ fixed L-CNs. 

Assume that we are given $A(v,t,s=0)$, $v \in [1:K]$, $t \in [0:K\mathtt{l}/\mathtt{r}]$, and set the remaining $A(v,t,s)$ to zero. Then, these remaining $A(v,t,s)$ that are set to zero can be updated by the recursion, which accumulates the number of constellations via adding extra VNs one by one. This recursion is the inverse process of the peeling decoding, so we call this process \textit{reconstruction}. Recursively determine $A(v,t,s)$, $t \in [0:K\mathtt{l}/\mathtt{r}]$, $s \in [1:K\mathtt{l}/\mathtt{r}-t]$ for $v \in [2:K]$ by means of the following update rule:
\begin{equation}\label{eq_recursion_A}
	\begin{aligned}
		&A(v,t,s) =\sum_{\Delta s = 1}^{\mathtt{l}} \sum_{\sigma = 0}^{\mathtt{l}-\Delta s} \sum_{\Delta t = 0}^{\lfloor \frac{\mathtt{l}-\Delta s - \sigma}{2} \rfloor} \sum_{\tau = 0}^{\mathtt{l}-\Delta s -\sigma - 2\Delta t}\\
		&A(v-1,t-\Delta t-\sigma,s+\sigma-\Delta s) v(\mathtt{l})!\\
		&\cdot {\rm coef}\left\{ \left((1+x)^{\mathtt{r}-1}-1\right)^\sigma \left((1+x)^\mathtt{r}-1-\mathtt{r}x\right)^{\Delta t},x^{\mathtt{l}-\Delta s-\tau} \right\}\\
		&\cdot \binom{(t-\Delta t-\sigma)\mathtt{r}-(v-1)\mathtt{l}+s+\sigma-\Delta s}{\tau} \binom{s+\sigma-\Delta s}{\sigma} \\
		&\cdot \binom{t+s}{\Delta t + \Delta s} \binom{\Delta t + \Delta s}{\Delta t} \mathtt{r}^{\Delta s} \frac{\Delta s}{s}.
	\end{aligned}
\end{equation}
Please refer to \cite[Ch. 3.22]{Modern_Coding_Theory} for the details of the above recursion.

It remains to discuss the initialization of the recursion (\ref{eq_recursion_A}), i.e., obtain $A(v,t,s=0)$. For the $(K,\mathtt{l},\mathtt{r})_2$ LDPC ensemble, the number of the $(v,t)$ stopping sets is ${\rm coef}\{ ((1+x)^\mathtt{r}-1-\mathtt{r}x)^{t},x^{v\mathtt{l}} \}(v\mathtt{l})!$. Thus, for $v \in [1:K]$ and $t,s \in [0:K\mathtt{l}/\mathtt{r}]$, we can initialize 
\begin{equation}\nonumber
	A(v,t,s) = {\rm coef}\{ ((1+x)^\mathtt{r}-1-\mathtt{r}x)^{t},x^{v\mathtt{l}} \}(v\mathtt{l})! \mathds{1}_{\{s=0\}},
\end{equation}
where the indicator function $\mathds{1}_{\{s=0\}}$ makes all $A(v,t,s > 0) = 0$.
For the $(K,\mathtt{l},\mathtt{r},\nu_{\rm min})_2$ LDPC ensemble with $\nu_{\rm min} \ge 2$, \cite[Ch. 3.22]{Modern_Coding_Theory} suggests us to initialize as follows:
\begin{equation}\label{eq:A_v_t_s_initial}
	\begin{aligned}
	A(v,t,s) =& {\rm coef}\{ ((1+x)^\mathtt{r}-1-\mathtt{r}x)^{t},x^{v\mathtt{l}} \}(v\mathtt{l})! \\
	&\cdot \mathds{1}_{\{s=0\}} \mathds{1}_{\{v \ge \nu_{\rm min}\}},
	\end{aligned}
\end{equation}
where $\mathds{1}_{\{v \ge \nu_{\rm min}\}}$ expurgates the stopping sets of size less than $\nu_{\rm min}$. However, Johnson \cite{Johnson2009finite} pointed out that (\ref{eq:A_v_t_s_initial}) may count subgraphs that are not the expurgated stopping sets themselves, but contain them as subgraphs. To ensure that these subgraphs are excluded in the initialization, a slightly modified initialization was proposed in \cite{Johnson2009finite}. 

Note that if we still use (\ref{eq:A_v_t_s_initial}) as initialization when $\nu_{\rm min} \ge 2$, the value of $A(v,t,s)$ we obtain by the recursion becomes an upper bound on the exact value of $A(v,t,s)$. For brevity, we do not define a new symbol to represent this upper bound on $A(v,t,s)$. However, we should keep in mind that when using the upper bounds as the initialization, all values obtained through recursion are also upper bounds.

Now, we are going to bound $B(a,a')$, which is the number of $({\cal A},{\cal A}')$ extended stopping sets with $|{\cal A}| = a$ and $|{\cal A}'| = a'$. Similar to the definition of $A(v,t,s)$, let $\tilde{A}(v,t,s;a)$ denote the number of constellations that are accumulated in $A(v,t,s)$ and the $v$ VNs form the $({\cal A},{\cal A}')$ extended stopping set with $|{\cal A}| = a$ and $|{\cal A}'| = v-a$. Thus, we have 
\begin{equation}\nonumber
	A(v,t,s) = \sum_{a = 1}^{v} \binom{v}{a} \tilde{A}(v,t,s;a),
\end{equation}
and
\begin{equation}\label{eq:B_a_a_prime}
	B(a,a') = \sum_{t,s} \binom{K\mathtt{l}/\mathtt{r}}{t+s} \tilde{A}(v=a+a',t,s;a).
\end{equation}

Note that an $({\cal A},{\cal A}')$ extended stopping set can be recursively constructed by adding the VNs in ${\cal A}'$ to the stopping set ${\cal A}$. This allow us to recursively determine $\tilde{A}(v,t,s;a)$. If we are given $\tilde{A}(v,t,s;a)$ for $v=a$, $s=0$, and $0 \le t \le K\mathtt{l}/\mathtt{r}$, and set the other terms to zero, we can recursively determine $\tilde{A}(v,t,s;a)$ for $v \in [a+1:K]$ as follows:
\begin{equation}\label{eq:recursion_A_tilde}
	\begin{aligned}
		&\tilde{A}(v,t,s;a) = \sum_{\Delta s = 1}^{\mathtt{l}} \sum_{\sigma = 0}^{\mathtt{l}-\Delta s} \sum_{\Delta t = 0}^{\lfloor \frac{\mathtt{l}-\Delta s - \sigma}{2} \rfloor} \sum_{\tau = 0}^{\mathtt{l}-\Delta s -\sigma - 2\Delta t}\\
		& \tilde{A}(v-1,t-\Delta t-\sigma,s+\sigma-\Delta s;a)\\
		&\cdot \underbrace{(v-a)}_{(a)} (\mathtt{l})! \binom{t+s}{\Delta t + \Delta s} \binom{\Delta t + \Delta s}{\Delta t}\\
		&\cdot {\rm coef}\left\{ \left((1+x)^{\mathtt{r}-1}-1\right)^\sigma \left((1+x)^\mathtt{r}-1-\mathtt{r}x\right)^{\Delta t},x^{\mathtt{l}-\Delta s-\tau} \right\}\\
		&\cdot \binom{(t-\Delta t-\sigma)\mathtt{r}-(v-1)\mathtt{l}+s+\sigma-\Delta s}{\tau} \binom{s+\sigma-\Delta s}{\sigma}\\
		&\cdot \mathtt{r}^{\Delta s} \frac{\Delta s}{s}.
	\end{aligned}
\end{equation}
The modified term $(a)$ is because the position of the maximal stopping set of size $a$ is specified and only the $v-a$ newly added VNs can be permuted.

Noting that the meanings of $\tilde{A}(v = a,t,s = 0;a)$ and $A(v=a,t,s=0)$ are identical, we can initialize all $\tilde{A}(v,t,s;a)$ as
\begin{equation}\label{eq:initial_A_tilde}
	\begin{aligned}
	\tilde{A}(v,t,s;a) =& {\rm coef}\{ ((1+x)^\mathtt{r}-1-\mathtt{r}x)^{t},x^{v\mathtt{l}} \}(v\mathtt{l})! \\
	&\cdot \mathds{1}_{\{s=0\}} \mathds{1}_{\{v=a\}} \mathds{1}_{\{v \ge \nu_{\rm min}\}}.
	\end{aligned}
\end{equation}
As mentioned before, the right hand side of (\ref{eq:initial_A_tilde}) is the exact value when $\nu_{\rm min} = 1$ and may become an upper bound when $\nu_{\rm min} \ge 2$. 

After finishing the recursion, we can obtain an upper bound on $B(a,a')$ using (\ref{eq:B_a_a_prime}). Substituting this upper bound into (\ref{eq:p_a_a_prime}), Lemma~\ref{lemma:ess_recursive} is proved.
\end{proof}

\begin{proof}[Proof of Lemma~\ref{lemma:ess}]
Consider the \textit{reconstruction} of the $({\cal A},{\cal A}')$ extended stopping set from the maximal stopping set of size $a$ by adding $a'$ VNs. Once an extra VN is added, there must be a newly added degree-1 L-CN connected to this extra VN. Such an L-CN is called \textit{peeling} L-CN, which guarantees that the extra VN can be removed by the peeling decoder. If the extra VN has more than one degree-1 L-CNs, we randomly choose one of them to be the peeling L-CN. Therefore, after the reconstruction, $a'$ peeling L-CNs will be added.

Let $\tilde{B}(a,a',c)$ denote the number of $({\cal A},{\cal A}')$ extended stopping sets in the ensemble, which are with $a'$ fixed peeling L-CNs and $c$ fixed L-CNs that are the neighbors of ${\cal A}$. According to the definition of $\tilde{B}(a,a',c)$, there may be some L-CNs which are the neighbors of ${\cal A} \cup {\cal A}'$, but are not included in the fixed $a'+c$ L-CNs. The relation between $B(a,a')$ and $\tilde{B}(a,a',c)$ is
\begin{equation}\nonumber
	B(a,a') = \binom{K\mathtt{l}/\mathtt{r}}{a'+c} \tilde{B}(a,a',c).
\end{equation}

We can recursively bound $\tilde{B}(a,a',c)$ with $a'$ steps based on the reconstruction process. Let $\tilde{B}_t(a,a',c)$ denote the number of constellations after adding $t$ VNs, where $t \in [0:a']$. Note that $\tilde{B}(a,a',c) = \tilde{B}_{a'}(a,a',c)$ and $\tilde{B}_0(a,a',c) = A(a,c,0)$. For $t \in [1:a']$, we have
\begin{equation}\label{eq_B_recursion}
	\begin{aligned}
	\tilde{B}_t(a,a',c) \le & \tilde{B}_{t-1}(a,a',c) t(c+t) \mathtt{r} \\
	&\cdot \binom{(K\mathtt{l}/\mathtt{r}-a'+t-1)\mathtt{r}-(a+t-1)\mathtt{l}}{\mathtt{l}-1} (\mathtt{l})!.
	\end{aligned}
\end{equation}
If we recursively bound $\tilde{B}_t(a,a',c)$ using the above inequality, the initial condition is
\begin{equation}\nonumber
	\tilde{B}_0(a,a',c) = {\rm coef}\{ ((1+x)^\mathtt{r}-1-\mathtt{r}x)^{c},x^{a\mathtt{l}} \}(a\mathtt{l})! \mathds{1}_{\{a \ge \nu_{\rm min}\}}.
\end{equation}
Here is the explanation for inequality (\ref{eq_B_recursion}). For brevity, we use $V^+$ and $C^+$ to refer to the newly added VN and the newly added peeling L-CN in the following, respectively. The factor $t$ accounts for the fact that $V^+$ can be chosen from $t$ VNs which are not in the maximal stopping set. Likewise, we can arbitrarily chose one L-CN from $c+t$ L-CNs to be $C^+$ and then chose a socket of $C^+$ to be connected by $V^+$, which gives the factor $(c+t)\mathtt{r}$. The remaining $\mathtt{l}-1$ sockets of $V^+$ can be connected to any $\mathtt{l}-1$ available sockets. At step $t$, there are $a'-(t-1)$ L-CNs reserved as the peeling L-CNs for the current step and the subsequent steps. Therefore, we can choose $\mathtt{l}-1$ sockets from the $(K\mathtt{l}/\mathtt{r}-a'+t-1)\mathtt{r}-(a+t-1)\mathtt{l}$ available sockets, resulting in the term $\binom{(K\mathtt{l}/\mathtt{r}-a'+t-1)\mathtt{r}-(a+t-1)\mathtt{l}}{\mathtt{l}-1}$. The term $(\mathtt{l})!$ is given for the fact that we can permute the newly added $\mathtt{l}$ edges. Finally, the inequality follows from the fact that we may count some constellations more than once.

Recursively applying (\ref{eq_B_recursion}) $a'$ times, we can obtain the upper bound on $\tilde{B}(a,a',c)$. In fact, we can calculate the upper bound for $p(a,a')$ more efficiently, avoiding the issue of excessively large values of $B(a,a')$ and $T(a+a')$. Let $\tilde{p}_t(a,a',c) = \frac{\tilde{B}_t(a,a',c)}{T(a+t)}$, where $T(v) = \binom{K\mathtt{l}}{v\mathtt{l}} (v\mathtt{l})!$ (recall that we consider the binary ensemble). Noting that 
\begin{equation}\nonumber
	\begin{aligned}
		T(a+t) = T(a+t-1) \prod_{i=0}^{\mathtt{l}-1} (K\mathtt{l}-(a+t-1)\mathtt{l}-i),
	\end{aligned}
\end{equation}
we can, based on (\ref{eq_B_recursion}), write the inequality
\begin{multline}\label{eq_p_recursion}
	\tilde{p}_t(a,a',c) \\ \le \tilde{p}_{t-1}(a,a',c) \frac{t(c+t) \mathtt{r} \binom{(K\mathtt{l}/\mathtt{r}-a'+t-1)\mathtt{r}-(a+t-1)\mathtt{l}}{\mathtt{l}-1} (\mathtt{l})!}{\prod_{i=0}^{\mathtt{l}-1} (K\mathtt{l}-(a+t-1)\mathtt{l}-i)},
\end{multline}
where $t \in [1,a']$. Here, to cover the case $\mathtt{l} = 0$, we adopt the convention $\prod_{i=0}^{\mathtt{l-1}} x_i = 1$ if $\mathtt{l} \le 0$. If we recursively bound $\tilde{p}_t(a,a',c)$, the initial condition is $\tilde{p}_0(a,a',c) = \tilde{B}_0(a,a',c)/T(a)$. Thus, we obtain the upper bound on $p_{a'}(a,a',c)$ as follows:
\begin{equation}\nonumber
	\tilde{p}_{a'}(a,a',c) \le \hat{p}(a,a',c),
\end{equation}
where $\hat{p}(a,a',c)$ is defined in (\ref{eq:hat_p}) and the equality holds if $a'=0$. Noting that the $a'+c$ L-CNs involved in $\tilde{p}_{a'}(a,a',c)$, which are the neighbors of the maximal stopping set or the peeling L-CNs, are fixed, we have
\begin{equation}\label{eq:UB_p_a_a_prime_1}
	\begin{aligned}
	p(a,a')  &= \sum_{c = 0}^{K\mathtt{l}/\mathtt{r}}\binom{K\mathtt{l}/\mathtt{r}}{a'+c} \tilde{p}_{a'}(a,a',c)\\
	& \le \sum_{c = 0}^{K\mathtt{l}/\mathtt{r}}\binom{K\mathtt{l}/\mathtt{r}}{a'+c} \hat{p}(a,a',c).
	\end{aligned}
\end{equation}

At this point, we have obtained the upper bound (\ref{eq:UB_p_a_a_prime_1}). We proceed to tighten the upper bound. Consider the $({\cal A},{\cal A}')$ extended stopping set with $|{\rm Neib}({\cal A})|= c$. Noting the fact that such an extended stopping set must contain the $(a,c)$ maximal stopping set, we can write
\begin{equation}\label{eq:conditional_inequality}
	\begin{aligned}
		&{\Pr}\left\{\text{such an extended stopping set}\right\} \\
		&= {\Pr}\Big\{{\cal A}=\text{ stopping set}, |{\rm Neib}({\cal A})|= c \Big\} \\
		&~~~\times \Pr\Big\{{({\cal S} \cup {\cal S}') \ne}\text{ stopping set},\forall {\cal S}\subset {\cal A}, {\cal S}'\subset {\cal A}',\\
		&~~~ {\cal S}' \ne \emptyset \bigm| {\cal A}=\text{ stopping set}, |{\rm Neib}({\cal A})|= c\Big\}.
	\end{aligned}
\end{equation}
Here, the position of $c$ L-CNs in ${\rm Neib}({\cal A})$ are not specified. The left hand side of (\ref{eq:conditional_inequality}) equals $\binom{K\mathtt{l}/\mathtt{r}}{a'+c} \tilde{p}_{a'}(a,a',c)$ and the first term in the right hand side of (\ref{eq:conditional_inequality}) equals $\binom{K\mathtt{l}/\mathtt{r}}{c} \hat{p}(a,0,c)$. Since probability $\le 1$, we have the inequality
\begin{equation}\nonumber
	\binom{K\mathtt{l}/\mathtt{r}}{a'+c} \tilde{p}_{a'}(a,a',c) \le \binom{K\mathtt{l}/\mathtt{r}}{c} \hat{p}(a,0,c).
\end{equation}
Therefore, we can give a tighter upper bound as
\begin{equation}\label{eq_prob_eSS}
	\begin{aligned}
		p(a,a')  &\le \sum_{c = 0}^{K\mathtt{l}/\mathtt{r}} \min \left\{ \binom{K\mathtt{l}/\mathtt{r}}{c} \hat{p}(a,0,c), \binom{K\mathtt{l}/\mathtt{r}}{a'+c} \hat{p}(a,a',c) \right\}\\
		& \triangleq L(a,a';K,\mathtt{l},\mathtt{r}),
	\end{aligned}
\end{equation}
where $\hat{p}(a,a',c)$ is defined in (\ref{eq:hat_p}).
\end{proof}

\section{Tables of Empirical Rank Distributions and Degree Distributions}\label{tables}
The degree distribution ${\bm \Psi} = (\Psi_1,\Psi_2,\ldots,\Psi_K)$ is represented by the polynomial 
\begin{equation}\nonumber
	\Psi(x) = \sum_{d = 1}^{K} \Psi_d x^d.
\end{equation}

\begin{table}[htbp]
	\caption{Empirical Rank Distribution ${\bf h}_1$ for $M=16$ and $q=256$}
	\centering
	\rowcolors{1}{lightgray!20}{}
	\begin{tabular}{lllllll}
		\hline
		$h_0$ & $h_1$ & $h_2$ & $h_3$ & $h_4$ & $h_5$ & $h_6$\\
		$0$&    $0$&    $0$&    $0$&    $0$&    $0.0001$&    $0.0004$\\
		 $h_7$ & $h_8$ & $h_9$ & $h_{10}$ & $h_{11}$ & $h_{12}$ & $h_{13}$ \\
		$0.0025$&    $0.0110$ & $0.0387$&    $0.1041$&   $0.2062$& $0.2795$&   $0.2339$\\
		$h_{14}$ & $h_{15}$ & $h_{16}$ & & & &\\
		$0.1039$&   $0.0190$&   $0.0008$& & & &\\
		\hline
	\end{tabular}
	\label{table_h1}
\end{table}

\begin{table}[htbp]
	\caption{Empirical Rank Distribution ${\bf h}_2$ for $M=8$ and $q=4$}
	\centering
	\rowcolors{1}{lightgray!20}{}
	\begin{tabular}{lllllll}
		\hline
		$h_0$ & $h_1$ & $h_2$ & $h_3$ & $h_4$ & $h_5$ & $h_6$ \\
		$0$&    $0.0002$&   $0.0025$&    $0.0207$&    $0.1027$&    $0.2846$&    $0.3805$\\
		$h_7$ & $h_8$ & & & & &\\
		$0.1895$&    $0.0194$& & & & &\\
		\hline
	\end{tabular}
\label{table_h2}
\end{table}

\begin{table}[htbp]
	\caption{Empirical Rank Distribution ${\bf h}_3$ for $M=32$ and $q=2$}
	\centering
	\rowcolors{1}{lightgray!20}{}
	\begin{tabular}{lllllll}
		\hline
		$h_0$ & $h_1$ & $h_2$ & $h_3$ & $h_4$ & $h_5$ & $h_6$ \\
		$0$&    $0$&    $0$&    $0$&    $0$&    $0$&    $0$\\
		$h_7$ & $h_8$&$h_9$ & $h_{10}$ & $h_{11}$ & $h_{12}$ & $h_{13}$ \\
		$0$&    $0$&$0$&    $0$&   $0$& $0$&   $0$\\
		$h_{14}$ & $h_{15}$ & $h_{16}$ & $h_{17}$&$h_{18}$ & $h_{19}$ & $h_{20}$ \\
		$0$&   $0.0001$&   $0.0002$& $0.0009$&$0.0030$&   $0.0092$& $0.0246$\\
		$h_{21}$ & $h_{22}$ & $h_{23}$ & $h_{24}$ & $h_{25}$ & $h_{26}$&$h_{27}$\\
		$0.0564$&    $0.1082$&   $0.1689$&   $0.2081$& $0.1955$ & $0.1347$&$0.0650$\\
		$h_{28}$ & $h_{29}$ & $h_{30}$ & $h_{31}$ & $h_{32}$ & &\\
		$0.0207$&   $0.0040$&    $0.0004$&   $0$&   $0$ & & \\
		\hline
	\end{tabular}
	\label{table_h3}
\end{table}

\begin{table}[htbp]
	\caption{Degree Distributions for Standard BATS Codes}
	\centering
	\rowcolors{1}{lightgray!20}{}
	\begin{tabular}{p{0.2cm}p{4.2cm}p{3.3cm}}
		\hline
		No. & Description & Polynomial Representation\\
		${\bm \Psi}_1$ & For $\mathscr{C}_{\rm std}^{(256)}(128,n,16,\allowbreak {\bm \Psi},{\bf h}_1)$, finite-length optimization using \cite[Theorem 16]{FL_analysis_BATS} at $n_1 = 30$ & $0.3504 x^{12} + 0.0253 x^{14} + 0.0785 x^{15} + 0.0931 x^{20} + 0.0245 x^{21} + 0.1030 x^{27} + 0.0501 x^{37} + 0.0288 x^{38} + 0.0244 x^{39} + 0.0503 x^{58} + 0.0184 x^{59} + 0.1532 x^{128}$\\
		${\bm \Psi}_2$ & For $\mathscr{C}_{\rm std}^{(256)}(256,n,16,\allowbreak {\bm \Psi},{\bf h}_1)$, finite-length optimization using \cite[Theorem 16]{FL_analysis_BATS} at $n_1 = 40$ & $0.2015 x^{12} + 0.0428 x^{14} + 0.1893 x^{15} + 0.0110 x^{20} + 0.0915 x^{21} + 0.0332 x^{27} + 0.0872 x^{28} + 0.0678 x^{37} + 0.0073 x^{39} + 0.0697 x^{50} + 0.0083 x^{71} + 0.0542 x^{72} + 0.0301 x^{116} + 0.0240 x^{117} + 0.0821 x^{256}$\\
		${\bm \Psi}_{11}$ & For $\mathscr{C}_{\rm std}^{(256)}(128,n,16,\allowbreak {\bm \Psi},{\bf h}_1)$, asymptotic optimization using \cite[(P1)]{BATS} at $\bar{\eta} = 0.99$ & $0.0464 x^{14} + 0.2263 x^{15} + 0.1304 x^{20} + 0.0344 x^{21} + 0.1221 x^{27} + 0.0702 x^{37} + 0.0404 x^{38} + 0.0776 x^{58} + 0.0258 x^{59} + 0.2263 x^{128}$\\
		${\bm \Psi}_{12}$ & For $\mathscr{C}_{\rm std}^{(256)}(128,n,16,\allowbreak {\bm \Psi},{\bf h}_1)$, finite-length optimization using (\ref{eq:optimization}) at $n_1=30$ & $0.3167 x^{11} + 0.0368 x^{14} + 0.1792 x^{15} + 0.0273 x^{21} + 0.0967 x^{27} + 0.0503 x^{37} + 0.0320 x^{38} + 0.0614 x^{58} + 0.0205 x^{59} + 0.1792 x^{128}$\\
		${\bm \Psi}_{18}$ & For $\mathscr{C}_{\rm std}^{(2)}(128,n,32,\allowbreak {\bm \Psi},{\bf h}_3)$, finite-length optimization using \cite[Theorem 16]{FL_analysis_BATS} at $n_1 = 15$ & $0.0248 x^{21} + 0.4235 x^{22} + 0.0252 x^{27} + 0.0984 x^{35} + 0.0245 x^{36} + 0.0249 x^{37} + 0.0507 x^{46} + 0.0355 x^{47} + 0.0251 x^{48} + 0.0511 x^{66} + 0.2163 x^{128}$\\
		\hline
	\end{tabular}\label{table:DD_std_BATS}
\end{table}

\begin{table}[htbp]
	\caption{Degree Distributions for LDPC-BATS Codes With $K'=64$}
	\centering
	\rowcolors{1}{lightgray!20}{}
	\begin{tabular}{p{0.2cm}p{4.2cm}p{3.3cm}}
		\hline
		No. & Description & Polynomial Representation\\
		${\bm \Psi}_3$ & For $\mathscr{C}_{\rm pre}^{(256)}(128,n,16, {\bm \Psi}, {\bf h}_1, 3,6,\allowbreak \nu_{\rm min})$, finite-length optimization using (\ref{eq:optimization}) at $n_1 = 11$, $\nu_{\rm min} = 1$, $n_2 = \infty$, $\epsilon^* = 1$ & $0.1236 x^{12} + 0.5587 x^{13} + 0.1500 x^{17} + 0.1677 x^{18}$\\
		${\bm \Psi}_7$ & For $\mathscr{C}_{\rm pre}^{(4)}(128,n,8,{\bm \Psi},{\bf h}_2,3,6,\allowbreak \nu_{\rm min})$, finite-length optimization using (\ref{eq:optimization}) at $n_1 = 24$, $\nu_{\rm min} = 1$, $n_2 = \infty$, $\epsilon^* = 1$ & $0.0721 x^6 + 0.9279 x^7$\\
		${\bm \Psi}_8$ & For $\mathscr{C}_{\rm pre}^{(4)}(80,n,8,{\bm \Psi},{\bf h}_2,3,15,\allowbreak \nu_{\rm min})$, finite-length optimization using (\ref{eq:optimization}) at $n_1 = 22$, $\nu_{\rm min} = 1$, $n_2 = \infty$, $\epsilon^* = 1$ & $0.4719 x^6 + 0.0353 x^7 + 0.1727 x^8 + 0.0519 x^{11} + 0.1117 x^{17} + 0.1566 x^{18}$\\
		${\bm \Psi}_{17}$ & For $\mathscr{C}_{\rm pre}^{(256)}(128,n,16,{\bm \Psi},{\bf h}_1,3,6,\allowbreak \nu_{\rm min})$, finite-length optimization using (\ref{eq:optimization}) at $n_1 = 15$, $\nu_{\rm min} = 4$, $n_2 = \infty$, $\epsilon^* = 1$ & $0.7384 x^{12} + 0.0242 x^{17} + 0.1208 x^{18} + 0.0235 x^{20} + 0.0270 x^{22} + 0.0247 x^{24} + 0.0414 x^{27}$\\
		\hline
	\end{tabular}\label{table:DD_LDPC_BATS_64}
\end{table}

\begin{table}[htbp]
	\caption{Degree Distributions for LDPC-BATS Codes With $K'=128$}
	\centering
	\rowcolors{1}{lightgray!20}{}
	\begin{tabular}{p{0.2cm}p{4.2cm}p{3.3cm}}
		\hline
		No. & Description & Polynomial Representation\\
		${\bm \Psi}_4$ & For $\mathscr{C}_{\rm pre}^{(256)}(256,n,16,{\bm \Psi},{\bf h}_1,3,6,\allowbreak \nu_{\rm min})$, finite-length optimization using (\ref{eq:optimization}) at $n_1 = 21$, $\nu_{\rm min} = 1$, $n_2 = \infty$, $\epsilon^* = 1$ & $0.3934 x^{13} + 0.2655 x^{14} + 0.1631 x^{17} + 0.1780 x^{18}$\\
		${\bm \Psi}_5$ & For $\mathscr{C}_{\rm pre}^{(256)}(256,n,16,{\bm \Psi},{\bf h}_1,3,6,\allowbreak \nu_{\rm min})$, finite-length optimization using (\ref{eq:optimization}) at $n_1 = 26$, $\nu_{\rm min} = 7$, $n_2 = \infty$, $\epsilon^* = 1$ & $0.5172 x^{12} + 0.1992 x^{13} + 0.0818 x^{17} + 0.1177 x^{18} + 0.0304 x^{20} + 0.0281 x^{21} + 0.0255 x^{23}$\\
		${\bm \Psi}_6$ & For $\mathscr{C}_{\rm pre}^{(256)}(256,n,16,{\bm \Psi},{\bf h}_1,3,6,\allowbreak \nu_{\rm min})$, asymptotic optimization using \cite[(P1)]{BATS} at $\bar{\eta} = 0.5$ & $0.3066 x^{14} + 0.3311 x^{17} + 0.3623 x^{18}$\\
		${\bm \Psi}_9$ & For $\mathscr{C}_{\rm pre}^{(2)}(256,n,32,{\bm \Psi},{\bf h}_3,3,6,\allowbreak \nu_{\rm min})$, asymptotic optimization using \cite[(P1)]{BATS} at $\bar{\eta} = 0.5$ & $0.1723 x^{26} + 0.2921 x^{27} + 0.5356 x^{35}$\\
		${\bm \Psi}_{10}$ & For $\mathscr{C}_{\rm pre}^{(2)}(160,n,32,{\bm \Psi},{\bf h}_3,3,15,\allowbreak \nu_{\rm min})$, asymptotic optimization using \cite[(P1)]{BATS} at $\bar{\eta} = 0.8$ & $0.0666 x^{26} + 0.2494 x^{27} + 0.1059 x^{34} + 0.0814 x^{35} + 0.0266 x^{42} + 0.1168 x^{43} + 0.1532 x^{55} + 0.1563 x^{81} + 0.0438 x^{82}$\\
		${\bm \Psi}_{13}$ & For $\mathscr{C}_{\rm pre}^{(256)}(256,n,16,{\bm \Psi},{\bf h}_1,3,6,\allowbreak \nu_{\rm min})$, finite-length optimization using (\ref{eq:optimization}) at $n_1 = 26$, $\nu_{\rm min} = 1$, $n_2 = \infty$, $\epsilon^* = 1$ & $0.3057 x^{11} + 0.1413 x^{13} + 0.1561 x^{14} + 0.1690 x^{17} + 0.1844 x^{18} + 0.0230 x^{23} + 0.0204 x^{30}$\\
		${\bm \Psi}_{14}$ & For $\mathscr{C}_{\rm pre}^{(2)}(256,n,32,{\bm \Psi},{\bf h}_3,3,6,\allowbreak \nu_{\rm min})$, finite-length optimization using (\ref{eq:optimization}) at $n_1 = 15$, $\nu_{\rm min} = 7$, $n_2 = 8$, $\epsilon^* = 0.05$ & $0.4152 x^{23} + 0.0254 x^{25} + 0.0781 x^{26} + 0.1107 x^{27} + 0.3706 x^{35}$\\
		${\bm \Psi}_{15}$ & For $\mathscr{C}_{\rm pre}^{(2)}(256,n,32,{\bm \Psi},{\bf h}_3,3,6,\allowbreak \nu_{\rm min})$, finite-length optimization using (\ref{eq:optimization}) at $n_1 = 15$, $\nu_{\rm min} = 7$, $n_2 = 8$, $\epsilon^* = 1$ & $0.1029 x^{20} + 0.5192 x^{22} + 0.1956 x^{23} + 0.0781 x^{27} + 0.0251 x^{32} + 0.0241 x^{34} + 0.0210 x^{35} + 0.0339 x^{37}$\\
		${\bm \Psi}_{16}$ & For $\mathscr{C}_{\rm pre}^{(256)}(256,n,16,{\bm \Psi},{\bf h}_1,3,6,\allowbreak \nu_{\rm min})$, finite-length optimization using (\ref{eq:optimization}) at $n_1 = 21$, $\nu_{\rm min} = 7$, $n_2 = \infty$, $\epsilon^* = 1$ & $0.3979 x^{13} + 0.2478 x^{14} + 0.2351 x^{17} + 0.1192 x^{18}$\\
		\hline
	\end{tabular}\label{table:DD_LDPC_BATS_128}
\end{table}

\bibliographystyle{IEEEtran}
\bibliography{reference}

\begin{IEEEbiographynophoto}{Mingyang Zhu}
	received the B.E. degree in information engineering and the Ph.D. degree in information and communication engineering from Southeast University, Nanjing, China, in 2018 and 2024, respectively. He was a Research Assistant with the Institute of Network Coding, The Chinese University of Hong Kong, from 2023 to 2024, where he is currently a Post-Doctoral Fellow. His research interests include channel coding theory and practice.
\end{IEEEbiographynophoto}

\begin{IEEEbiographynophoto}{Shenghao Yang}
	received the B.S. degree from Nankai University in 2001, the M.S. degree from Peking University in 2004, and the Ph.D. degree in information engineering from The Chinese University of Hong Kong in 2008. He was a Visiting Student with the Department of Informatics, University of Bergen, Norway, in Spring 2007. He was a Post-Doctoral Fellow with the University of Waterloo from 2008 to 2009 and the Institute of Network Coding, The Chinese University of Hong Kong, from 2010 to 2012. He was with the Institute for Interdisciplinary Information Sciences, Tsinghua University, from 2012 to 2015. He is currently an Associate Professor with The Chinese University of Hong Kong, Shenzhen. His research interests include network coding, information theory, coding theory, and quantum information.
\end{IEEEbiographynophoto}

\begin{IEEEbiographynophoto}{Ming Jiang}
	received the B.Sc., M.S., and Ph.D. degrees in communication and information engineering from Southeast University, Nanjing, China, in 1998, 2003, and 2007, respectively. He is currently an Associate Professor with the National Mobile Communications Research Laboratory, Southeast University, Nanjing. His research interest includes coding and modulation technology.
\end{IEEEbiographynophoto}

\begin{IEEEbiographynophoto}{Chunming Zhao}
	received the B.Sc. and M.S. degrees from the Nanjing Institute of Posts and Telecommunications in 1982 and 1984, respectively, and the Ph.D. degree from the Department of Electrical and Electronic Engineering, University of Kaiserslautern, Germany, in 1993. He has been a Post-Doctoral Researcher with the National Mobile Communications Research Laboratory, Southeast University, where he is currently a Professor and the Vice Director of the Laboratory. He has managed several key projects of Chinese Communications High-Tech Program. His research interests include communication theory, coding/decoding, mobile communications, and VLSI design. Dr. Zhao received the First Prize of National Technique Invention of China in 2011 and the Excellent Researcher Award from the Ministry of Science and Technology, China.
\end{IEEEbiographynophoto}
\end{document}